\documentclass[aps,prl,floatfix,nofootinbib,superscriptaddress,twocolumn,10pt,raggedbottom]{revtex4-2}
\usepackage[english]{babel}

\usepackage{xcolor}
\usepackage{times}
\usepackage{graphicx}
\usepackage{xcolor}

\usepackage{newpxtext,newpxmath}

\let\coloneqq\relax

\usepackage[utf8]{inputenc}
\usepackage{amsthm}
\usepackage{amssymb}
\usepackage{amsmath}
\usepackage{mleftright}
\usepackage{bbold}
\usepackage{bbm}
\usepackage[pdftex, colorlinks=true, allcolors=MidnightBlue!70!black!85!TealBlue]{hyperref}
\usepackage[capitalise]{cleveref}
\usepackage{braket}
\usepackage{dsfont}
\usepackage{mathdots}
\usepackage{mathtools}
\usepackage{enumitem}
\usepackage{csquotes}
\usepackage[cal=boondox]{mathalfa}
\usepackage{graphicx}
\usepackage{stackengine}
\usepackage{scalerel}
\usepackage{tensor}       
\usepackage[dvipsnames]{xcolor}
\usepackage{array}
\usepackage{makecell}
\newcolumntype{x}[1]{>{\centering\arraybackslash}p{#1}}
\usepackage{tikz}
\usetikzlibrary{shapes.geometric, shapes.misc, positioning, arrows, arrows.meta, decorations.pathreplacing, decorations.pathmorphing, patterns, angles, quotes, calc}
\usepackage{booktabs}
\usepackage{xfrac}
\usepackage{siunitx}
\usepackage{centernot}
\usepackage{comment}
\usepackage{chngcntr}
\usepackage{titlesec}

\titleformat{\section}{\centering\bfseries\boldmath\MakeUppercase}{\thesection.}{0.5em}{}
\titleformat{\subsection}{\centering\bfseries\boldmath}{\thesubsection.}{0.5em}{}

\newtheorem{thm}{Theorem}
\newtheorem*{thm*}{Theorem}
\newtheorem{prop}[thm]{Proposition}
\newtheorem*{prop*}{Proposition}
\newtheorem{lemma}[thm]{Lemma}
\newtheorem*{lemma*}{Lemma}
\newtheorem{cor}[thm]{Corollary}
\newtheorem*{cor*}{Corollary}
\newtheorem{cj}[thm]{Conjecture}
\newtheorem*{cj*}{Conjecture}

\newtheorem*{Def*}{Definition}

\newtheorem*{question*}{Question}
\newtheorem{problem}[thm]{Problem}
\newtheorem*{problem*}{Problem}

\makeatletter
\def\thmhead@plain#1#2#3{%
  \thmname{#1}\thmnumber{\@ifnotempty{#1}{ }\@upn{#2}}%
  \thmnote{ {\the\thm@notefont#3}}}
\let\thmhead\thmhead@plain
\makeatother

\theoremstyle{definition}
\newtheorem{rem}[thm]{Remark}

\newcommand{\bb}{\begin{equation}\begin{aligned}\hspace{0pt}}
\newcommand{\bbb}{\begin{equation*}\begin{aligned}}
\newcommand{\ee}{\end{aligned}\end{equation}}
\newcommand{\eee}{\end{aligned}\end{equation*}}
\newcommand*{\coloneqq}{\mathrel{\vcenter{\baselineskip0.5ex \lineskiplimit0pt \hbox{\scriptsize.}\hbox{\scriptsize.}}} =}

\newcommand\ceil[1]{\left\lceil#1\right\rceil}
\newcommand{\eqt}[1]{\stackrel{\mathclap{\mbox{\scriptsize #1}}}{=}}
\newcommand{\leqt}[1]{\stackrel{\mathclap{\mbox{\scriptsize #1}}}{\leq}}

\newcommand{\textg}[1]{\stackrel{\mathclap{\mbox{\scriptsize #1}}}{>}}
\newcommand{\geqt}[1]{\stackrel{\mathclap{\mbox{\scriptsize #1}}}{\geq}}
\newcommand{\ketbra}[1]{\ket{#1}\!\!\bra{#1}}

\renewcommand{\epsilon}{\varepsilon}

\newcommand{\ot}{\otimes}

\newcommand{\eps}{\varepsilon}

\newcommand{\bigO}{\mathcal O}
\newcommand{\gaussian}{\mathcal G}

\def\ZZ{\mathbbm{Z}}
\def\RR{\mathbbm{R}}
\def\PP{\mathbbm{P}}
\def\CC{\mathbbm{C}}
\def\FF{\mathbbm{F}}
\def\NN{\mathbbm{N}}

\def\H{\mathcal{H}}

\def\Id{\mathbbm{1}}

\newcommand{\id}{\mathds{1}}

\newcommand{\E}{\mathds{E}}

\newcommand{\metapl}{Mp}

\DeclareMathOperator{\Tr}{Tr}

\DeclareMathOperator{\Span}{span}
\DeclareMathAlphabet{\pazocal}{OMS}{zplm}{m}{n}

\DeclareMathOperator{\U}{U}
\DeclareMathOperator{\Sp}{Sp}
\DeclareMathOperator{\SO}{SO}
\DeclareMathOperator{\Spin}{Spin}
\DeclareMathOperator{\Mp}{Mp}
\DeclareMathOperator{\Sym}{Sym}

\DeclareMathOperator{\Var}{Var}

\DeclareMathOperator{\accept}{acc}
\DeclareMathOperator{\reject}{rej}
\DeclareMathOperator{\mixed}{mixed}

\DeclarePairedDelimiter\abs{\lvert}{\rvert}
\DeclarePairedDelimiter\norm{\lVert}{\rVert}
\DeclarePairedDelimiter\angles{\langle}{\rangle}
\DeclarePairedDelimiter\parens{\lparen}{\rparen}

\newcommand{\lsmatrix}{\left(\begin{smallmatrix}}
\newcommand{\rsmatrix}{\end{smallmatrix}\right)}

\stackMath

\stackMath

\makeatletter
\newcommand*\rel@kern[1]{\kern#1\dimexpr\macc@kerna}
\newcommand*\widebar[1]{%
  \begingroup
  \def\mathaccent##1##2{%
    \rel@kern{0.8}%
    \overline{\rel@kern{-0.8}\macc@nucleus\rel@kern{0.2}}%
    \rel@kern{-0.2}%
  }%
  \macc@depth\@ne
  \let\math@bgroup\@empty \let\math@egroup\macc@set@skewchar
  \mathsurround\z@ \frozen@everymath{\mathgroup\macc@group\relax}%
  \macc@set@skewchar\relax
  \let\mathaccentV\macc@nested@a
  \macc@nested@a\relax111{#1}%
  \endgroup
}

\counterwithin*{equation}{part}
\counterwithin*{thm}{part}
\counterwithin*{figure}{part}

\tikzset{meter/.append style={draw, inner sep=10, rectangle, font=\vphantom{A}, minimum width=30, line width=.8, path picture={\draw[black] ([shift={(.1,.3)}]path picture bounding box.south west) to[bend left=50] ([shift={(-.1,.3)}]path picture bounding box.south east);\draw[black,-latex] ([shift={(0,.1)}]path picture bounding box.south) -- ([shift={(.3,-.1)}]path picture bounding box.north);}}}
\tikzset{roundnode/.append style={circle, draw=black, fill=gray!20, thick, minimum size=10mm}}
\tikzset{squarenode/.style={rectangle, draw=black, fill=none, thick, minimum size=10mm}}

\definecolor{Blues5seq1}{RGB}{239,243,255}
\definecolor{Blues5seq2}{RGB}{189,215,231}
\definecolor{Blues5seq3}{RGB}{107,174,214}
\definecolor{Blues5seq4}{RGB}{49,130,189}
\definecolor{Blues5seq5}{RGB}{8,81,156}

\definecolor{Greens5seq1}{RGB}{237,248,233}
\definecolor{Greens5seq2}{RGB}{186,228,179}
\definecolor{Greens5seq3}{RGB}{116,196,118}
\definecolor{Greens5seq4}{RGB}{49,163,84}
\definecolor{Greens5seq5}{RGB}{0,109,44}

\definecolor{Reds5seq1}{RGB}{254,229,217}
\definecolor{Reds5seq2}{RGB}{252,174,145}
\definecolor{Reds5seq3}{RGB}{251,106,74}
\definecolor{Reds5seq4}{RGB}{222,45,38}
\definecolor{Reds5seq5}{RGB}{165,15,21}

\allowdisplaybreaks

\makeatletter 
\renewcommand\onecolumngrid{
\do@columngrid{one}{\@ne}%
\def\set@footnotewidth{\onecolumngrid}
\def\footnoterule{\kern-6pt\hrule width 1.5in\kern6pt}%
}
\makeatother 

\usepackage{subcaption}
\usepackage{tabularray}
\UseTblrLibrary{booktabs} 

\allowdisplaybreaks[1] 

\newcommand*{\addFileDependency}[1]{
  \typeout{(#1)}
  \@addtofilelist{#1}
  \IfFileExists{#1}{}{\typeout{No file #1.}}
}

\makeatother

\usepackage{enumitem}

\usepackage{quantikz}

\usepackage{xcolor}
\definecolor{famBrown}{RGB}{153,76,0} 

\usepackage{mdframed}

\usepackage[most,breakable]{tcolorbox}
\renewenvironment{boxed}[1][white]%
{\expandafter\ifstrequal\expandafter{#1}{filled}{\begin{tcolorbox}[colback=gray!3,colframe=gray!20,breakable=false,enhanced,left=5.75pt,right=5.75pt,grow sidewards by=10pt]}{\begin{tcolorbox}[colback=MidnightBlue!70!black!70!TealBlue!4!white,colframe=MidnightBlue!70!black!70!TealBlue!50!white,breakable=false,enhanced,left=5.75pt,right=5.75pt,grow sidewards by=10pt]}}%
  {\end{tcolorbox}}

\begin{document}

\setcounter{secnumdepth}{2}
\setlength{\parskip}{0.1pt}
\title{Is it Gaussian? Testing bosonic quantum states}

\author{Filippo Girardi$^\clubsuit$}
\email{filippo.girardi@sns.it}
\affiliation{Scuola Normale Superiore, Piazza dei Cavalieri 7, 56126 Pisa, Italy}

\author{Freek Witteveen$^\clubsuit$}
\email{f.witteveen@cwi.nl}
\affiliation{QuSoft and Centrum Wiskunde \& Informatica, Science Park 123, Amsterdam, the Netherlands}

\author{Francesco Anna Mele}
\email{francesco.mele@sns.it}
\affiliation{NEST, Scuola Normale Superiore and Istituto Nanoscienze, Piazza dei Cavalieri 7, IT-56126 Pisa, Italy}

\author{Lennart~Bittel}
\email{l.bittel@fu-berlin.de}
\affiliation{Dahlem Center for Complex Quantum Systems, Freie Universit\"at Berlin, 14195 Berlin, Germany}

\author{Salvatore F. E. Oliviero}
\email{salvatore.oliviero@sns.it}
\affiliation{Scuola Normale Superiore, Piazza dei Cavalieri 7, 56126 Pisa, Italy}

\author{David Gross}
\email{david.gross@thp.uni-koeln.de}
\affiliation{Institute for Theoretical Physics, University of Cologne, Z\"ulpicher Stra\ss{}e 77, Cologne, Germany}

\author{Michael Walter}
\email{michael.walter@lmu.de}
\affiliation{Ludwig-Maximilians-Universit\"at M\"unchen, Theresienstra\ss{}e 37, Munich, Germany}
\affiliation{Ruhr University Bochum, Universit\"atsstra\ss{}e 150, Bochum, Germany}
\affiliation{KdVI and QuSoft, University of Amsterdam, Science Park 123, Amsterdam, the Netherlands}

\begin{abstract}
    Gaussian states are widely regarded as one of the most relevant classes
    of continuous-variable (CV) quantum states, as they naturally arise in physical systems and play a key role in quantum technologies. This motivates a fundamental question: given copies of an unknown CV state, how can we efficiently test whether it is Gaussian? We address this problem from the perspective of representation theory and quantum learning theory, characterizing the sample complexity of Gaussianity testing as a function of the number of modes. For pure states, we prove that just a \emph{constant} number of copies is sufficient to decide whether the state is \textit{exactly} Gaussian. We then extend this to the tolerant setting, showing that a polynomial number of copies suffices to distinguish states that are close to Gaussian from those that are far. In contrast, we establish that testing Gaussianity of general mixed states necessarily requires exponentially many copies, thereby identifying a fundamental limitation in testing CV systems. Our approach relies on rotation-invariant symmetries of Gaussian states together with the recently introduced toolbox of CV trace-distance bounds.
\end{abstract}

\maketitle

\def\thefootnote{$^\clubsuit$}\footnotetext{These authors contributed equally to this work.}\def\thefootnote{\arabic{footnote}}

\section{Introduction}

Property testing~\cite{montanaro2018surveyquantumpropertytesting,anshu2023survey} concerns the following broadly applicable scenario: you can sample from an unknown random variable, or have access to copies of an unknown quantum state, and you want to know whether the random variable or quantum state has some given property. Ideally, one would like to determine this with high confidence, but using as few samples or copies as possible with respect to the system size.

While property testing has been extensively studied in finite-dimensional quantum systems, its application to continuous-variable (CV) quantum systems---where states are defined over infinite-dimensional Hilbert spaces---remains largely unexplored. This gap is particularly striking given the recent surge of interest in quantum learning theory for CV systems~\cite{mele2024learning,bittel2024optimal1,fanizza2024efficient,bittel2025energyindependent,holevo2024estimatestracenormdistancequantum,gandhari_precision_2023,becker_classical_2023,oh2024entanglementenabled,liu2025quantumlearningadvantagescalable,coroi2025exponentialadvantagecontinuousvariablequantum,fawzi2024optimalfidelityestimationbinary,wu2024efficient,mobus2023diss,upreti2024efficientquantumstateverification,zhao2025complexityquantumtomographygenuine,symplectity_rank, iosue2025higher,fanizza2025gaussianunitary}. A significant portion of this literature focuses on quantum state tomography, where the efficiency landscape is sharply divided: tomography of general non-Gaussian states is provably inefficient~\cite{mele2024learning}, while Gaussian states admit efficient procedures~\cite{mele2024learning,bittel2024optimal1,bittel2025energyindependent,fanizza2024efficient}.

In this work, we address a fundamental question in this context: \emph{Can we efficiently determine whether an unknown CV state is Gaussian?} We address a refined version of this question, where we determine the precise sample complexity of this task depending on the deviation from Gaussianity.
This question is of intrinsic interest as a fundamental information-theoretic and statistical question on CV quantum systems, given the prominent role of Gaussian states in theory of CV quantum systems and their importance in quantum optics.
It also has direct applications for CV-based quantum technologies.
Gaussian states underpin a broad range of applications, including quantum metrology \cite{Oh_2019,Nichols2018metrology,Matsubara_2019}, quantum cryptography \cite{Usenko_2015,Grosshans02,Grosshans_2003,Bozzio_2019}, universal quantum computation \cite{Menicucci2006}, and Gaussian boson sampling \cite{Hamilton_2017,cazalis2024gaussianbosonsamplingbinary,andersen2025usinggaussianbosonsamplers}. A robust test for Gaussianity is therefore an valuable tool for a wide range of applications.
For example, if one can reliably identify a state or process as Gaussian, one can then apply specialized tomography algorithms to learn the full state~\cite{mele2024learning,bittel2024optimal1,bittel2025energyindependent,fanizza2024efficient}.

Previous works on property testing have addressed the problem of testing stabilizer states~\cite{gross2021schur,grewal2023improved,arunachalam2024tolerant,bao2025tolerant} and fermionic Gaussian states~\cite{bittel2024optimal,lyu2024fermionicgaussiantestinga}.
A common theme with these works is that good property tests can be derived from symmetry considerations.
We will see that this approach can be made to work in the continuous-variable setting of testing bosonic Gaussian states as well.

\subsection{Summary of results}
We now present an overview of our results.
In order to test whether a given state is Gaussian, we must use some condition which singles out Gaussian states compared to other states.
We now explain two conditions characterizing Gaussianity, and explain in more detail \cref{sec:testing pure states} how they are used for property testing.

\medskip
\textbf{Rotation-invariance:} There is a well-known characterization of Gaussian probability distributions. If $X_1$ and $X_2$ are independent random variables, then the random variables $Y_1$, $Y_2$ defined by
\bb\label{eq:kac bernstein}
\begin{pmatrix}
    Y_1 \\ Y_2
\end{pmatrix}
= \frac{1}{\sqrt2} \begin{pmatrix}
    1 & -1 \\ 1 & 1
\end{pmatrix}
\begin{pmatrix}
    X_1 \\ X_2
\end{pmatrix}
\ee
are independent if and only if $X_1$ and $X_2$ are Gaussian random variables. This is known as the Kac-Bernstein theorem, see \cite{kagan1975characterization} for a proof and related characterizations.
A similar fact is true for quantum states, where Gaussian states are characterized uniquely by their properties with respect to phase space rotations, a fact that has been explored in the context of the quantum central limit theorem \cite{cushen1971quantum} and for a Gaussian de Finetti theorem \cite{leverrier2009quantum}.
We use this to show that there exists a test which only uses three copies of the wavefunction (and two for testing zero-mean states), independently of the number of modes. This test is such that upon receiving a pure Gaussian quantum state, it accepts with certainty (we say the test has \emph{perfect completeness}), while for a non-Gaussian state it rejects with nonzero probability.
Our main result here is a bound on the robustness of this test, bounding the rejection probability in terms of the distance to the closest Gaussian state.
Note that the situation is in contrast to the classical case, where one cannot test independence with perfect completeness.
The difference is that in the quantum case rotation-invariance can be tested using interference, similar to the swap test.

\medskip
\textbf{Covariance matrix:} Gaussian states are (similar to Gaussian probability distributions) fully determined by their mean and covariance. The covariance matrix of a quantum state can not be arbitrary: it is constrained by \emph{uncertainty relations}. One way to formulate these uncertainty relations is that the symplectic eigenvalues $\nu_i$ of the covariance matrix satisfy $\nu_i \geq 1$. Here we have equality $\nu_i = 1$ for all symplectic eigenvalues if and only if the state is a pure Gaussian state. This means that determining the covariance matrix of a CV state should suffice to determine whether the state is Gaussian.
We construct a tolerant tester based on this idea.
This method not only provides a test for Gaussianity but also yields a classical description of the underlying Gaussian state. The tester still has a polynomial sample complexity, albeit higher than the symmetry-based test in general.
An advantage is that this approach only requires single-copy access and does not require quantum memory.

\medskip
\textbf{Mixed states:}
The above two properties can be used to test \emph{pure} Gaussian states. What about mixed states?
We show that there is no efficient way of testing Gaussianity for mixed states, by showing that the number of copies needed to test Gaussianity has to scale \emph{exponentially} with the number of modes of the system, even when multi-copy measurements are allowed. 

\medskip
This work is organized as follows. We first describe the setting and problem more precisely in \cref{sec:preliminaries main text}. In \cref{sec:testing pure states} we describe tests for Gaussianity testing of pure states, and in \cref{sec:intro_mixed} we give lower bounds for the case of mixed states.
Detailed analysis and proofs are deferred to the Supplementary Material.

\section{Preliminaries}\label{sec:preliminaries main text}
Testing a property of a quantum state means the following: we are given $N$ copies of a state $\rho$, and we want to determine whether $\rho$ is an element of a set of states $\mathcal P$ which have some given property (in our case this will be the set of Gaussian states).
A property test corresponds to a measurement on the $N$ copies of the states, and it either accepts or rejects.

A test has \emph{perfect completeness} if it accepts all states with property $\mathcal P$ with probability 1. Beyond this, a desirable feature is \emph{robustness}: the acceptance probability should not change abruptly with small perturbations of the input state, but rather reflect how close the state is to satisfying $\mathcal P$, up to finite experimental accuracy. To formalize this, we measure distances between states using the \emph{trace distance} $T(\rho, \sigma) = \tfrac{1}{2}\|\rho - \sigma\|_1$~\cite{NC,MARK,KHATRI}, an operationally meaningful metric~\cite{HELSTROM,Holevo1976} that characterizes how well two states can be distinguished by measurements.

\begin{figure}[t]
    \centering
    \includegraphics[width=1\linewidth]{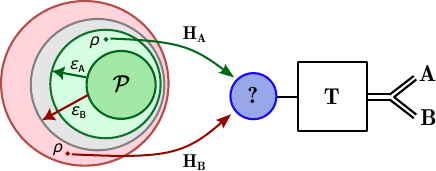}
    \caption{The property testing problem. The test has to determine, upon receiving copies of a state $\rho$, whether $\rho$ is $\varepsilon_A$-close to a state with the property $\mathcal P$, or $\varepsilon_B$-far away from all such states.}
    \label{fig:property}
\end{figure}

\begin{problem}\label{prob:testing}
Let $0\leq \eps_A<\eps_B$, let $0<\delta<1$ and let $\mathcal{P}$ be a set of states. We say that an algorithm
solves the property testing problem using $N$ samples if, for any generic state $\rho$ satisfying one of the following two mutually exclusive hypotheses
\begin{enumerate}
    \item [$(H_A)$] either $\rho$ is $\eps_A$-close to $\mathcal{P}$, i.e.\ $\displaystyle{\min_{\sigma\in\mathcal{P}}\frac 12 \|\rho-\sigma\|_1\leq \eps_A}$,
    \item [$(H_B)$] or $\rho$ is $\eps_B$-far from $\mathcal{P}$, i.e.\ $\displaystyle{\min_{\sigma\in\mathcal{P}}\frac 12 \|\rho-\sigma\|_1> \eps_B}$,
\end{enumerate}
the algorithm can identify the underlying hypothesis with failure probability at most $\delta$.

In the case where $\eps_A = 0$ and $\eps_B = \eps$ (so we are testing a scenario where the state either has the property exactly, or is~$\eps$-far from $\mathcal P$), we simply call the algorithm an $\eps$-test. 
\end{problem}

When $\eps_A>0$, we call this problem \textit{tolerant testing}, which is visualized in \cref{fig:property}.
Note that in general, it suffices to consider the case $\delta = \frac13$ (or any other constant smaller than $\frac12$), since this can always be amplified to an arbitrary success probability $\delta$ using $\log\frac1\delta$ repetitions of the test and taking a majority vote on the outcomes.

In this work, the property that we study is that of \emph{Gaussianity}.
The Hilbert space of $n$ bosonic (CV) modes is given by $\mathcal H_n = L^2(\RR^n)$~\cite{BUCCO}.
The position and momentum operators are denoted by $\hat{x}_i, \hat{p}_i$ for $i = 1, \dots, n$.
We let $\hat{E} =\frac12 \sum_{i=1}^n \left(\hat{x}_i^2 + \hat{p}_i^2\right)$.
In $\mathcal H_n$ pure Gaussian states are ground states of quadratic Hamiltonians in the position and momentum operators; mixed Gaussian states are thermal states for such Hamiltonians.
Alternatively, Gaussian states are states with a characteristic function on the phase space $\RR^{2n}$ that is Gaussian, see \cref{sec:preliminaries} for detailed definitions.
We use the convention that for pure states $\ket{\psi}$ we write $\psi = \proj{\psi}$ for the corresponding density matrix.

\section{Testing pure Gaussian states}\label{sec:testing pure states}

\subsection{Testing by symmetry}

The idea is to exploit the fact that tensor powers of Gaussian states are characterized by an enhanced symmetry as compared to arbitrary quantum states, mirroring the characterization of Gaussian probability distributions as in \cref{eq:kac bernstein}.
As the simplest example, consider $\H_n^{\ot 2} \cong L^2(\RR^{2n})$. We can apply a phase space rotation $U_{\theta}$ between the copies as
\begin{align*}
    (U_{\theta} f)(\mathbf x,\mathbf y) \coloneqq f(\mathbf x\cos\theta +  \mathbf{y}\sin\theta, -\mathbf{x}\sin\theta + \mathbf{y}\cos\theta).
\end{align*}
for $f \in L^2(\RR^{2n})$.
More generally, any orthogonal matrix $o \in O(k)$ acts on $\H_n^{\ot k} = L^2(\RR^n)^{\ot k} \cong L^2(\RR^{n \times k})$ as a unitary $U_o$.

Any such $U_o$ corresponds to a passive Gaussian unitary, and as such it can be implemented via a sequence of beam splitters acting on pairs of modes and single-mode phase shifters~\cite{BUCCO}. For example, $U_{\pi/4}$ corresponds to the $50\!:\!50$ beam splitter.

It is straightforward to see that if $\rho$ is a Gaussian state, then $U_o \rho^{\ot k} U_o^\dagger = \rho^{\ot k}$ for any orthogonal and stochastic (mapping the uniform vector $\mathbf 1 = (1,1,\dots,1)^\intercal$ to itself) $o \in O(k)$. We have invariance for any $o \in O(k)$ if $\rho$ additionally has mean zero.
To see this, we note that a Gaussian state is completely determined by its covariance and mean. If $\rho$ has covariance $V$ and mean $m$ with respect to phase space $\RR^{2n}$, $\rho^{\ot k}$ has covariance and mean on the phase space $\RR^{2n \times k}$, which represents $k$ copies of the phase space, given by $V \ot \id_k$ and $m \ot \mathbf 1$.
On the other hand, $U_o$ acts on the phase space as $\id \ot o$, so it preserves the covariance matrix, and if $m \neq 0$ it preserves the mean if $o$ is stochastic.

The crucial feature is that the converse is also true: if a state is invariant under any non-trivial orthogonal transformation, the state must be Gaussian.

\begin{thm}[(see \cref{thm:rotation invariant})]\label{thm:rotation invariant main text}
    For a pure state $\psi \in \mathcal H_n$ and $k \geq 2$, the following are equivalent:
    \begin{enumerate}[noitemsep]
        \item $U_o \ket\psi^{\ot k} = \ket\psi^{\ot k}$ for all $o \in O(k)$.
        \item $U_o \ket\psi^{\ot k} = \ket\psi^{\ot k}$ for a rotation $o \in O(k)$ that does not have all entries in $\{0,\pm 1\}$.
        \item $\psi$ is a pure Gaussian state with zero mean.
    \end{enumerate}
    If instead we allow arbitrary mean, the same is true if we consider $k \geq 3$ and restrict to stochastic orthogonal matrices.
\end{thm}

The result is similar to well-known characterizations of Gaussian probability distributions.
While likely folklore, we were not able to find \cref{thm:rotation invariant main text} in the literature and hence we provide a complete proof in \cref{sec:testing by symmetry appendix} as \cref{thm:rotation invariant}.
The theorem can also be obtained from a quantum version of the Kac-Bernstein characterization of Gaussian random variables due to \cite{springer2009conditions,Cuesta2020}, which states that $\rho_1$ and $\rho_2$ are Gaussian states if, and only if, $U_\theta (\rho_1 \ot \rho_2) U_\theta^\dagger$ is a product state for some (and then any) $\theta$ which is not a multiple of $\pi/2$.
It is also closely related to the quantum central limit theorem \cite{Cushen1971,wolf2006extremality} and the notion of convolution of bosonic states, see also~\cite{bu2025efficientmeasurementbosonic,hahn2025measuring}.

At this point, we would like to highlight a similarity to the swap test and its relation to symmetry.
The basis for the swap test is the fact that a state $\rho$ on $\CC^d$ is pure if and only if $F \rho^{\ot 2} = \rho^{\ot 2}$ where $F$ is the swap operator.
This is closely related to the fact that the swap operator generates the commutant of the action of $U^{\ot 2}$ for $U \in U(d)$, which is a special case of Schur-Weyl duality.
In the CV setting, the action of the orthogonal group $O(k)$ spans the commutant of the action of $U^{\ot k}$ for Gaussian symplectic unitaries $U$, a fact known as Howe-duality.
The general idea of the Gaussianity test, where tensor powers of Gaussian states have an enhanced symmetry, is closely analogous to the situation for stabilizer states, where studying the commutant of powers of the Clifford unitaries allows efficient stabilizer testing \cite{gross2021schur}. This can be seen as a discrete phase space variant of the approach in this work. We comment on the analogy, and the relation to representation theory in \cref{sec:relation to rep theory}.

Now, the swap test determines the purity of a state by measuring whether the state is invariant under $F$.
We can similarly test Gaussianity by measuring whether the state is invariant under phase space rotations (i.e.~measure the linear subspace of $\H_n^{\ot k}$ that is invariant under $O(k)$).
We call Gaussianity tests based on this principle \emph{rotation tests}.
The most simple variant is to measure invariance for arbitrary rotations for $k = 2$, or stochastic rotations for $k = 3$. This gives tests that reject with some non-zero probability if and only if the state is not a mean-zero Gaussian pure state, or Gaussian pure state, for $k=2$ and $k=3$ respectively.
We prove that these tests are robust, by bounding the rejection probability in terms of the distance to the closest Gaussian state, which determines how many times we need to repeat the measurement to witness the non-Gaussianity.

\begin{thm}[(see \cref{thm:appendix theorem rotation testing})]\label{thm:robust test}
    Let $\psi$ be a pure state and consider the task of testing Gaussianity. Suppose that if the state is far from Gaussian, the state is guaranteed to satisfy an energy bound $\sqrt{\Tr[\hat{E}^2 \psi]} \leq n E$.
    Then, repeating a measurement of rotation-invariance of $k=3$ copies
    \begin{align*}
        N = \bigO\mleft(\max \left\{ \eps^{-2}, n^4 E^4 \right\} \log\frac1\delta \mright)
    \end{align*}
    times, yields an $\eps$-test for Gaussianity.
\end{thm}
The same result holds for testing mean-zero Gaussian states using rotation invariance on \mbox{$k=2$~copies}.
As these tests have perfect completeness, they extend straightforwardly to the tolerant setting:
$\bigO(\gamma^{-2} \log\frac1\delta)$~samples suffice, where
\mbox{$\gamma \coloneqq C \min(\eps_B^2,(nE)^{-4}) - \eps_A$} 
and $C$ is a universal constant (see \cref{thm:appendix theorem rotation testing}).

It is not obvious how to implement the measurement that measures the rotation-invariant subspace.
An alternative is to only test invariance under rotations over a discrete angle (which also characterize Gaussianity completely).
For the zero-mean case, this angle can be taken to be $\pi/4$, which still gives a test with perfect completeness.
Testing invariance under these specific rotations can now be done with a simple cicuit, see \cref{fig:rotation test main text} for the $k = 2$ case.
Note that this circuit is in close analogy to the swap test, which applies a swap between the systems instead of rotating the phase spaces. When restricted to the even parity subspace, $U_{\pi/2} = (U_{\pi/4})^2$ in fact swaps the two copies, so this rotation test is in some sense the square root of the swap test.
This test could be more practical to implement since $U_{\pi/4}$ models a standard beam-splitter in quantum optics; it is, however, suboptimal compared to measuring full rotation-invariance.
For the case of non-zero mean, similar results are achieved with $k = 3$, and a rotation of angle $\pi/3$ along the uniform vector $\mathbf 1$.
We do not prove a complete robustness result for these tests. A result in this direction can be found in \cite{Cuesta2020}, which however only gives a robustness that is exponential in the number of modes.

For states which are close to Gaussian (and which should intuitively be the most difficult to distinguish from Gaussian), we do show a bound which is completely independent of the number of modes. This bound is also used in \cref{thm:robust test}.
\begin{thm}[(see \cref{thm:close to gaussian})]\label{thm:close intro}
    Suppose that a pure state $\psi$ is $\eps$-far from a zero-mean Gaussian state, for $\eps \leq \eps_0$ for a universal constant $\eps_0$. Then the test in \cref{fig:rotation test main text} (or its three-copy variant for Gaussian states which do not have mean zero) has acceptance probability
    \begin{align*}
        p_{\accept}(\psi) = 1 - \Omega(\eps^2).
    \end{align*}
\end{thm}
We conjecture that the restriction to $\eps < \eps_0$ is not necessary in this result, which would imply that the rotation test is independent of the number of modes (similar to the case of stabilizer testing~\cite{gross2021schur}), and which would give an improved version of \cref{thm:robust test}, which would no longer depend on $E$ or $n$.

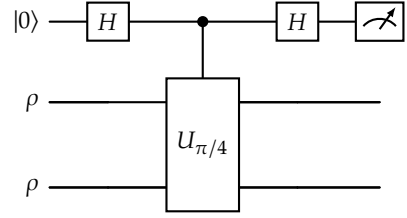
\begin{figure}[t]
    {\centering
    \begin{tikzcd}
        \lstick{$\ket{0}$} & \gate{H} & \ctrl{1} & \gate{H} & \meter{}\\
        \lstick{$\rho$} & \qw & \gate[2]{U_{\pi/4}} & \qw & \qw \\
        \lstick{$\rho$} & \qw & \qw & \qw & \qw
    \end{tikzcd}}
    \caption{Quantum circuit for the rotation test over angle $\pi/4$ (see \cref{thm:close intro}). One way to implement this test is via the Hadamard test for $U_{\pi/4}$. In the figure, $H$ denotes the Hadamard gate; as in a standard Hadamard test, one applies $H$ to the extra qubit, then the controlled-$U_{\pi/4}$, followed by another $H$ before measuring the extra qubit.}
    \label{fig:rotation test main text}
\end{figure}

\subsection{Testing by learning}
Testing Gaussianity by using the rotation symmetry relies on the ability to act on multiple copies of the state simultaneously. This may not be feasible, for instance in situations where one can just prepare a single copy of the state at a time. Here we describe an approach which tests Gaussianity by \emph{learning} the mean and covariance matrix of an $n$-mode bosonic state, and has the advantage that, if the state is (close to) Gaussian, one also obtains an estimate of the state. Another advantage of this approach is that it does not require entangling operations between multiple copies of the state as in \cref{fig:rotation test main text}. The downside is that the sample complexity can be higher.

In the non-tolerant testing, i.e.~$\eps_A = 0$ and $\eps_B=\eps$, a recent tomography algorithm for Gaussian states \cite{bittel2025energyindependent} can be used as a subroutine of a Gaussianity testing procedure which is described in \cref{app:tom}.

\begin{thm}[(see \cref{thm:pure_testing_by_learning})] Consider the task of testing Gaussianity as in \cref{prob:testing} on a (possibly-mixed) $n$-mode state $\rho$, with  $\mathcal{P}=\mathcal{G}_E$ being the set of pure Gaussian states with mean energy per mode at most $E$. There exists a $\eps$-test only using single-copy measurements which requires
    \begin{align*}
            N =\bigO\!\left(\frac{n^3 + n \log\log E}{\eps^2}\log\frac{1}{\delta}\right)
        \end{align*}
        copies of $\rho$ to decide whether $\rho$ is a Gaussian state or $\epsilon$-far from $\mathcal{G}_{E}$ with success probability at least $1-\delta$.
\end{thm}

In the tolerant setting, where $\eps_A>0$, the above-mentioned Gaussian tomography algorithm can no longer be directly applied.
Instead, we rely on a recent result on the efficient learning of the mean and covariance that applies to general states~\cite{mele2024learning}.
This fact can be combined with a recently developed toolbox regarding bounds on the trace distance between continuous-variable quantum states~\cite{mele2024learning,bittel2024optimal1,symplectity_rank,bittel2025energyindependent,fanizza2024efficient,mele2025achievableratesnonasymptoticbosonic,holevo2024estimatesburesdistancebosonic,holevo2024estimates} (see \cref{sec:tr_distance}). The synergy of these two properties of continuous-variable systems turns out to be crucial for our analysis of property testing of Gaussian states from a learning perspective.

Let $\mathcal{G}_E$ be the set of pure Gaussian states with average energy per mode at most $E$. The main idea for the testing procedure is to exploit the fact that a continuous-variable state is a pure Gaussian state if and only if the symplectic eigenvalues of its covariance matrix are all 1 
(see \cref{subsec:cv} for the definition).
More precisely, we prove that the trace distance between a generic state $\rho$ and the set $\mathcal{G}_E$ can be upper and lower bounded in terms of the symplectic eigenvalues of the covariance matrix of~$\rho$.

\begin{lemma}[(see \cref{thm:upper_bound} and \cref{lemma_lower_bound_min})]
    Let $\rho$ be an $n$-mode state with second moment of the energy upper bounded by $nE$, i.e.~$\sqrt{\Tr[\hat{E}^2\rho]}\le nE$. Moreover, let $\{\nu_i\}_i$ be the set of the symplectic eigenvalues of the covariance matrix of $\rho$ and let $\nu_{\rm{max}}$ be the largest symplectic eigenvalue. Then, it holds that
    \begin{equation*}
    \frac{(\nu_{\rm{max}}-1)^2}{c\,E^6}\leq \min_{\psi\in\mathcal{G}_E}\frac 12 \|\rho-\psi\|_1\leq \frac{1}{\sqrt 2}\sqrt{\sum_{i=1}^n(\nu_i-1})\,,
    \end{equation*}
    where $c$ is a universal constant.
\end{lemma}

As a consequence, we can leverage these bounds to provide rigorous guarantees on the success probability of the following testing procedure. Using single-copy measurements, we produce an estimator $\tilde V$ of the covariance matrix $V(\rho)$ of $\rho$ and we compare the largest symplectic eigenvalue $\tilde \nu_{\max}$ of $\tilde V$: if it turns out to be too far from $1$, then the pure state cannot be Gaussian.

\begin{thm}[(see \cref{lem:good_estimator} and \cref{thm:pure_testing})] Consider the task of testing Gaussianity as in \cref{prob:testing} on a (possibly-mixed) $n$-mode state $\rho$, with  $\mathcal{P}=\mathcal{G}_E$ being the set of pure Gaussian states with mean energy per mode at most $E$, error parameters $\eps_A$ and $\eps_B$, and failure probability $\delta$.   There is a universal constant $C>0$ such that, if
        \begin{equation*}
    \eta \coloneqq\eps_B^4-Cn^8E^6\eps_A>0,
    \end{equation*}
    and $\rho$ has the second moment bounded by $\sqrt{\Tr[\hat{E}^2 \rho]} \leq n E$,
    then single-copy measurements on
    \[ N=\bigO\!\left( \frac{n^7E^6}{\eta^2} \log\!\frac{n}{\delta} \right) \]
    copies of $\rho$ suffice to decide whether $\rho$ is $\epsilon_A$-close or $\epsilon_B$-far from $\mathcal{G}_{E}$ with success probability at least $1-\delta$.
\end{thm}

\section{Mixed state scenario}\label{sec:intro_mixed}
Let $\mathcal{G}_{E,\mixed}$ be the set of all $n$-mode (possibly mixed) Gaussian states with mean energy per mode upper bounded by $E$. We may wonder whether testing if a state $\rho$ is $\eps_A$-close or $\eps_B$-far from $\mathcal{G}_{E,\mixed}$ (with high success probability) is an easy task, in terms of the sample complexity, as in the pure state scenario discussed in the previous section. On the contrary, for the mixed case the property testing problem represented in \cref{fig:property} has a sample complexity exponentially large in the number of modes.
\begin{thm}[(see \cref{thm:hardness})]\label{thm:hardness_intro}
Let $\mathcal{G}_{E,\mixed}$ be the set of all $n$-mode (possibly mixed) Gaussian states with mean energy per mode upper bounded by $E$.
    Consider the task of testing Gaussianity as in \cref{prob:testing}, with $\mathcal{P}=\mathcal{G}_{E,\mixed}$ on $n$ modes, and with parameters $\eps_A$, $\eps_B$ and $\delta$, on a possibly mixed state $\rho$ satisfying the energy bound $\sqrt{\Tr[\hat{E}^2 \rho]} \leq n E$. Then, if $\varepsilon_{A}<\varepsilon_B<C/(nE)^4$ for a universal constant $C$, deciding whether $\rho$ is $\epsilon_A$-close or $\epsilon_B$-far from $\mathcal{G}_{E,\mixed}$ with success probability larger than $2/3$ requires at least
    \[N=\Omega\left(\frac{E^{n}}{n^4E^4\eps_B^2}\right)\]
    copies of $\rho$.
\end{thm}
This means that any procedure capable of certifying Gaussianity must use an exponential number of samples with respect to the system size. This result holds \textit{even when quantum memories are allowed}, highlighting the hardness of Gaussianity testing for mixed CV states.

We conjecture that the requirement $\eps_B< 1/c_{nE}$ can be relaxed to $\eps_B<1$ by some nontrivial refinements in the proof strategy. In Appendix \ref{subsec:mixed2}, we show that this is the case if we replace the trace distance with the quantum relative entropy, which remains an operationally meaningful measure of distinguishability thanks to the celebrated quantum Stein's lemma~\cite{Hiai1991}. In this setting, the minimum quantum relative entropy over all Gaussian states can be evaluated exactly --- coinciding with the \emph{relative entropy of non-Gaussianity}~\cite{Genoni2008,Marian2013} --- enabling sharper sample-complexity analyses compared to those obtained with the trace distance.

\section{Conclusion and future directions}
In this work, we investigated the problem of testing Gaussianity in CV quantum systems. Our results offer rigorous bounds on the sample complexity needed for this task in different regimes.
We briefly recapitulate our main results, and suggest interesting future directions.
Our first contribution is the introduction of a testing framework that exploits rotational symmetries which exactly characterize pure Gaussian states. We show that this provides robust Gaussianity tests.
Secondly, we demonstrate that by estimating the first moment and the covariance matrix, one can test for Gaussianity with polynomial sample complexity with single-copy measurements (and also learning the Gaussian state).
Finally, we analyze the general case of mixed states and prove that any property testing strategy in this regime incurs an exponential sample complexity in the number of modes. This establishes a fundamental information-theoretic limitation for Gaussianity testing.

There are several open questions in this area that remain to be addressed. We expect that even in the tolerant case, a constant number of copies of the state suffices for the symmetry-based tests. In \cref{sec:testing by symmetry appendix} we give some evidence in this direction. Another hint in this direction is the comparison to testing stabilizer states. There, it is known that a very similar strategy, based on discrete phase spaces, allows for stabilizer testing with a constant number of copies.

One can also investigate a different kind of tolerant testing.
Let
\bb
F_{\gaussian}(\psi) = \sup_{\phi \in \gaussian} \abs{\braket{\phi\vert\psi}}
\ee
be the maximal overlap of a pure state $\psi$ with a Gaussian state.
One can then ask whether $F_{\gaussian}(\psi) \leq \eps_A$ or $F_{\gaussian}(\psi) \geq \eps_B$.
This question has recently been studied in the stabilizer case, again using the discrete phase space perspective \cite{bao2025tolerant,arunachalam2025polynomial}.
An interesting open question is whether these techniques also apply to CV states.

While we conjecture that an improved analysis can yield symmetry-based tests needing only a constant number of copies, it is clear that any approach based on learning the covariance matrix will require a number of copies scaling with the number of modes $n$. An open question is to give a lower bound for the number of copies needed for Gaussianity testing using single-copy measurements (see \cite{hinsche2025single} for such a lower bound in the analogous stabilizer case).
Further future directions are to apply our methods to the problem of testing Gaussian unitaries and to testing states which are given by applying Gaussian unitaries to Fock states \cite{iosue2025higher}.

Finally, for the lower bounds on testing mixed Gaussian states, our proof requires polynomially small $\eps_B$, a condition that could potentially be relaxed.

\section{Acknowledgements}
The authors would like to acknowledge Cl\'ement Canonne, Marco Fanizza, Ludovico Lami, Lorenzo Leone, Antonio Anna Mele, Sepehr Nezami, Mich\`ele Vergne, and Chirag Wadhwa for helpful discussions.
F.G.\ and F.A.M.\ acknowledge financial support from the European Union (ERC StG ETQO, Grant Agreement no.\ 101165230).
M.W.\ acknowledges financial support from the European Union (ERC StG SYMOPTIC, Grant Agreement no.\ 101040907).
Views and opinions expressed are however those of the author(s) only and do not necessarily reflect those of the European Union or the European Research Council. Neither the European Union nor the granting authority can be held responsible for them.
F.A.M.\ thanks California Institute of Technology for hospitality. F.A.M.\ acknowledges financial support from the project: PRIN 2022 "Recovering Information in Sloppy QUantum modEls (RISQUE)", code 2022T25TR3, CUP E53D23002400006.
S.F.E.O.\ acknowledge financial support by MUR (Ministero dell'Istruzione, dell'Universit\`a e della Ricerca) through the following projects: PNRR MUR project PE0000023-NQSTI.
M.W.\ also acknowledges the German Federal Ministry of Research, Technology and Space (QuSol, 13N17173).

\bibliographystyle{apsrev4-1}
\bibliography{biblio,library}

\clearpage
\appendix
\onecolumngrid

\tableofcontents
\newpage

\section{Preliminaries, notation and overview of the results}\label{sec:preliminaries}

\begin{table}[h]
    \caption{Recap of the notation.}
    \label{table1}
    \begin{tblr}{
            colspec = {p{0.10\linewidth} X Q[r, wd=0.18\linewidth]},
            rows = {m}, 
            row{1} = {font=\bfseries}, 
        }
        \toprule
        Symbol                   & Description                                                                                                                                     & Introduced in                           \\
        \midrule
        $\mathcal{D}(\H)$        & set of the density matrices (states) of the Hilbert space $\H$                                                                                  &                                         \\
        $\psi$                   & density matrix associated to the pure state $\ket{\psi}$, namely $\psi = \ketbra \psi$                                                          &                                         \\
        $\mathcal P$             & subset of states with some property of interest                                                                                                 & \cref{sec:intro_to_property_testing}    \\
        $d_{\mathcal P}(\rho)$   & trace distance between $\rho$ and the family of states $\mathcal{P}$                                                                            & \textit{ibid.}                          \\
        $n$                      & number of modes of the bosonic system                                                                                                           & \cref{subsec:cv}                        \\
        $\mathcal{H}_n$          & Hilbert space for a continuous variable (CV) system of $n$ modes $\H_n = L^2(\RR^n)$                                                            & \textit{ibid.}                          \\
        $\hat{x}_i$, $\hat{p}_i$ & position and momentum operators, respectively, of the $i$th mode                                                                                & \textit{ibid.}                          \\
        $\hat{\mathbf{R}}$       & quadrature operator vector $\hat{\mathbf{R}} = (\hat x_1,\dots,\hat x_n,\hat p_1,\dots,\hat p_n)$                                               & \textit{ibid.}                          \\
        $\mathbf{m}(\rho)$       & first moment of the quantum state $\rho$                                                                                                        & \textit{ibid.}                          \\
        $V(\rho)$                & covariance matrix of the quantum state $\rho$                                                                                                   & \textit{ibid.}                          \\
        $D_{\mathbf{r}}$         & displacement operator                                                                                                                           & \textit{ibid.}                          \\
        $\chi_\rho$              & characteristic function $\chi_\rho \colon \RR^{2n} \to \CC$ of the state $\rho$, defined as $\chi_{\rho}(\mathbf{r}) = \Tr[D_{\mathbf{r}}\rho]$ & \textit{ibid.}                          \\
        $U_S$                    & symplectic Gaussian unitary associated with a symplectic matrix \( S \)                                                                         & \textit{ibid.}                          \\
        $\hat{E}$                & energy operator                                                                                                                                 & \textit{ibid.}                          \\
        $\rho(V,\mathbf{m})$     & the Gaussian state with covariance matrix $V$ and first moment $\mathbf{m}$                                                                     & \cref{sec:tr_distance}                  \\
        $U_{\theta}$             & unitary implementing a rotation of an angle $\theta$ on $\mathcal{H}_n^{\otimes 2}$                                                             & \cref{sec:rotation}                     \\
        $G$                      & generator of the rotations on $\mathcal{H}_n^{\otimes 3}$                                                                                       & \cref{sec:testing_rotation}             \\
        $\tilde G$               & generator of $U_{\theta}$                                                                                                                       & \textit{ibid.}                          \\
        $\mathcal{G}$            & set of pure Gaussian states                                                                                                                     & \cref{sec:testing by symmetry appendix} \\
        $\mathcal{G}_E$          & set of pure Gaussian states $\psi$ with energy constraint $\Tr[\psi \hat E]\leq nE$                                                             & \cref{sec:pure}                         \\
        $\mathcal{G}_{E,\mixed}$ & set of (possibly mixed) Gaussian states $\rho$ with energy constraint $\Tr[\rho \hat E]\leq nE$                                                 & \cref{sec:mixed}                        \\
        $\mathcal{G}_{\mixed}$   & set of all Gaussian states                                                                                                                      & \cref{subsec:mixed2}                    \\
        \bottomrule
    \end{tblr}
\end{table}

\cref{table1} summarises the notation used throughout the paper.

\subsection{Property testing}\label{sec:intro_to_property_testing}

We briefly recall the formulation of property testing of quantum states from the main text. An introduction to quantum property testing can be found in \cite{montanaro2018surveyquantumpropertytesting}.
We consider a Hilbert space $\H$, and we let $\mathcal P$ be a subset of states with some property of interest.

The question is: given $N$ copies of a state $\rho \in \mathcal{D}(\H)$, is $\rho \in \mathcal P$ (or close to it), or far from it?
We can formulate this in terms of two hypotheses, where given $0 \leq \eps_A < \eps_B$, we assume that $\rho$ satisfies one of two mutually exclusive options:
\begin{enumerate}
    \item [$(H_A)$] either $\rho$ is $\eps_A$-close from $\mathcal{P}$, i.e.~$\displaystyle{\min_{\sigma \in \mathcal P}} \frac12\norm{\rho - \sigma}_1 \leq \eps_A$
    \item [$(H_B)$] or $\rho$ is $\eps_B$-far from $\mathcal{P}$, i.e.~$\displaystyle{\min_{\sigma \in \mathcal P}} \frac12\norm{\rho - \sigma}_1> \eps_B$.
\end{enumerate}
It is useful to think of a test for distinguishing these two hypotheses as a two-outcome POVM $\{P, I-P\}$ with outcomes \emph{accept} and \emph{reject}, where
\begin{align*}
    p_{\accept}(\rho) = \Tr[ \rho^{\ot N} P] \qquad p_{\reject}(\rho) = \Tr[\rho^{\ot N}(I - P)].
\end{align*}
The goal is that $p_{\accept}(\rho)$ is large under hypothesis $H_A$, and small under $H_B$. To be precise, we say that the test has probability of error $\delta$ if $p_{\reject}(\rho) \leq \delta$ under hypothesis $H_A$ and $p_{\accept}(\rho) \leq \delta$ under hypothesis $H_B$.\\

For notational simplicity, we will sometimes write
\bb
d_{\mathcal P}(\rho) = \min_{\sigma \in \mathcal P} \frac12\norm{\rho - \sigma}_1
\ee
to denote the trace distance of a state $\rho$ to the set $\mathcal P$.\\

A test has \emph{perfect completeness} if $p_{\accept}(\rho) = 1$ for $\rho \in \mathcal P$ (so the test never incorrectly rejects $H_A$).
If the property $\mathcal P$ is a set of pure states, there is an optimal test using $N$ copies with perfect completeness, and it is given by letting $P$ be the projection onto
\bb\label{eq:optimal perfect completeness}
V_{\mathcal P, N} = \Span \{ \ket{\psi}^{\ot N} \, : \, \ket{\psi} \in \mathcal P \}.
\ee
The optimality here is that, among all tests with perfect completeness, this test has the highest probability of rejecting any state that is not in~$\mathcal P$ (see e.g~ Lemma 2.1 in \cite{lovitz2024nearly}).
By construction, this test is also such that if $\min_{\sigma \in \mathcal P} \frac12\norm{\rho - \sigma}_1\leq \eps$, then $p_{\reject}(\rho) \leq N \eps$ by a telescoping bound on the tensor powers.\\

If we have a test which has a gap in acceptance probability between two hypotheses, we can always repeat the test to amplify this gap, and obtain a test with a small probability of error.

\begin{lemma}\label{lem:amplify error test}
    Suppose we have a property testing procedure, using $N$ copies of $\rho$, where $p_{\accept}(\rho) \geq p_A$ under hypothesis $H_A$, and $p_{\accept} \leq H_B$ under $H_B$.
    Let $\eps = p_A - p_B$
    Then, there exists a test which has probability of error $\delta$ and uses $k = \bigO(\eps^{-2} \log(\delta))$ rounds of the $N$-copy test.
    If $p_A = 1 -
        \bigO(\eps)$, $k = \bigO(\eps^{-1}\log(\delta))$ rounds suffice.
\end{lemma}

\begin{proof}
    The new testing algorithm is as follows: perform $N$ rounds of the $N$-copy test, which define Bernoulli random variables $X_i$, which take outcome $0,1$ corresponding to rejecting or accepting respectively.
    Then, $\bar X = \frac{1}{k}(X_1 + \dots + X_k)$ is an estimator for $p_{\accept}(\rho)$.
    Let $p = (p_A + p_B)/2$. The new test accepts if and only if $\bar X \geq p$. 
    By assumption on $p_A$ and $p_B$, the test has probability of error at most $\delta$ if
    \bb
    \Pr\mleft(\abs{\bar X - p_{\accept}(\rho)} \geq \eps/2 \mright) \leq \delta.
    \ee
    By the Bernstein inequality,
    $k = \bigO(\eps^{-2} \log(\delta))$ suffices, and if $p_A = 1 -
        \bigO(\eps)$, this can be reduced to $k = \bigO(\eps^{-1}\log(\delta))$.
\end{proof}

\cref{table2} reports all the tests discussed in this paper with the corresponding sample complexity, if available.

\begin{table}[t]
    \caption{Overview of tests proposed in this work to certify Gaussianity of a quantum state.}
    \label{table2}
    \begin{tblr}{
            width=\textwidth,
            colspec = {p{0.12\linewidth} X Q[l, wd=0.22\textwidth]},
            row{1} = {font=\bfseries}, 
        }
        \toprule
        Test & Description& Sample complexity \\
        & & {\scriptsize $N = $ copies of the state needed,}    \\
        & & {\scriptsize $k$ = copies used in one measurement}      \\
         & & {\scriptsize $\delta$ = failure probability}      \\
        \midrule
        \SetCell[c=3]{c}{\textbf{Testing by symmetry}} \\
        \SetCell[c=3]{c}{\textit{These tests all have perfect completeness. For tolerant testing, we assume an energy bound if the state is far from Gaussian.}}   \\
        \midrule
        \SetCell[c=3]{c}{\textit{(We test a pure state $\psi$ with respect to the set of pure zero-mean Gaussian states $\mathcal{G}_0$)}}  \\
        \midrule
        Test 1    &
        Take two copies of a quantum state and measures whether the state in the rotation-invariant subspace. In the sample complexity, \mbox{$\eps = \Omega\left(\min\left(\eps_B^2, \frac{1}{n^4 E^4} \right) - \eps_A\right)$}. &
        $k=2$,\qquad \mbox{$N = \bigO(\eps^{-2} \log(\delta^{-1}))$}.   \\
        Test 2   
        & Same as Test 1, but testing only rotation invariance over an angle $\pi/4$.   & $k=2$.                                             \\
        \midrule
        \SetCell[c=3]{c}{\textit{(We test a pure state $\psi$ with respect to the set of pure Gaussian states $\mathcal{G}$.)}}      \\
        \midrule
        Test 3    & Apply a phase space rotation by $\pi/4$ and test if the result is a product state.  & $k=2$.   \\
        Test 4    & Take three copies of a quantum state and measures whether the state in the subspace invariant under rotations along the axis $(1,1,1)$. In the sample complexity, $\eps$ is as in Test 1.  & $k = 3$, \qquad\mbox{$N = \bigO(\eps^{-2} \log(\delta^{-1}))$}. \\
        Test 5  & Same as Test 4, but testing only rotation invariance over an angle $\pi/3$.  & $k = 3$.  \\
        \midrule
        \SetCell[c=3]{c}{\textbf{Testing by learning}}\\
        \SetCell[c=3]{c}{\textit{(We test an energy constrained state $\rho$ with respect to the set of energy constrained pure Gaussian states $\mathcal{G}_E$.)}}       \\
        \midrule
        Covariance test  & This test produces an estimator of the covariance matrix of the unknown state. If the largest symplectic eigenvalue of the estimator is close to $1$, then the state is guessed to be close to the set of pure Gaussian states. In the sample complexity $\eta \coloneqq\eps_B^4-\frac{c}{2} n^8E^6\eps_A$. &
        $k=1$, \qquad\mbox{$N=\bigO\left( \log\!\left(\frac{n^2}{\delta}\right)\frac{n^7E^6}{\eta^2}\right)$}.  \\
        \midrule
        Tomography test  & This test leverages the energy-independent tomography method described in~\cite{bittel2025energyindependent}. &
        $k=1$, \qquad\mbox{$N =\bigO\left(\frac{n^3 + n \log\log E}{\eps_B^2}\log\frac{1}{\delta}\right)$}.  \\
        \bottomrule
    \end{tblr}
\end{table}

\subsection{Continuous-variable states}\label{subsec:cv}
The Hilbert space for a continuous variable (CV) system of $n$ modes is $\H_n = L^2(\RR^n)$.
We label coordinates of $\RR^{2n}$ as $\mathbf{r} = (x_1,\dots,x_n, p_1,\dots,p_n)$, and the quadrature operator vector, denoted as $\hat{\mathbf{R}}$, is defined as
\begin{align*}
    \hat{\mathbf{R}} = (\hat x_1,\dots,\hat x_n,\hat p_1,\dots,\hat p_n)
\end{align*}
where $\hat{x}_i$ and $\hat{p}_i$ are the position and momentum operators of the $i$th mode. Using the convention $\hbar = 1$, the canonical commutation relations for position and momentum are
\bb
[\hat x_j,\hat p_k]=i\delta_{jk} \qquad j,k=1,\dots, n.
\ee
It is  possible to define the creation and annihilation operators as functions of the position and momentum operators.
\begin{align*}
    a_j = \frac{\hat x_j + i \hat p_j}{\sqrt2}, \quad a_j^\dagger = \frac{\hat x_j - i \hat p_j}{\sqrt2} \qquad j = 1, \dots, n.
\end{align*}
The corresponding commutation relations are
\bb
[\hat a_j,\hat a^\dagger_k]=\delta_{jk} \qquad j,k=1,\dots, n.
\ee

The Hilbert space $\H_n$ has a basis consisting of the \emph{number states}
\begin{align*}
    \ket{\mathbf k} = \frac{(a_{1}^\dagger)^{k_1} \dots (a_n^\dagger)^{k_n}}{\sqrt{k_1 ! \dots k_n !}} \ket{0} \qquad \mathbf{k} = (k_1, \dots, k_n) \in \ZZ_{\geq 0}^n
\end{align*}
where $\ket{0}$ is the vacuum state, as it corresponds to the occupation numbers being all zero.
The number of particles is given by $\abs{\mathbf k} = k_1 + \dots + k_n$.
We say $\ket{\psi}$ has even parity if it invariant under inversion of $(x_1,\dots,x_n)$; the space of even-parity states is spanned by number states with an even number of particles.

The \textit{energy operator} is defined as
\bb
\hat E \coloneqq \frac{1}{2}\mathbf{R}^\intercal \mathbf{R} = \sum_{i=1}^n\frac{\hat x_i^2+\hat p_i^2}{2}.
\ee

Consider a density operator~$\rho$ on $L^2(\RR^n)$.
The first moment $\mathbf{m}(\rho)$ and the covariance matrix $V(\rho)$ of the quantum state $\rho$ are defined as
\bb
\mathbf{m}(\rho)&\coloneqq \Tr[\mathbf{R}\rho]\\
V(\rho)& \coloneqq \Tr[\{\mathbf{R}-\mathbf{m}(\rho),(\mathbf{R}-\mathbf{m}(\rho))^{\intercal}\}\rho]
\ee
where $\{\cdot,\cdot\}$ is the anticommutator, and $(\cdot)^{\intercal}$ denotes the transpose operation. Any covariance matrix $V$ is strictly positive and satisfies the uncertainty relation $V+i\Omega \succeq 0$, with $\Omega$ the symplectic form
\bb
\Omega=\begin{pmatrix}
    0      & \id_n \\
    -\id_n & 0
\end{pmatrix}
\ee
Consequently, the covariance matrix \( V \) admits a Williamson decomposition: there exists a symplectic matrix \( S \) and real numbers \( \nu_1 \ge \nu_2 \ge \ldots \ge \nu_n \), called the \emph{symplectic eigenvalues} of \( V \), such that
\(
V = S D S^{\intercal}, \quad \text{with} \quad D \coloneqq \operatorname{diag}(\nu_1, \nu_1, \ldots, \nu_n, \nu_n).
\)
The symplectic eigenvalues of a covariance matrix are unique and satisfy \( \nu_i \ge 1 \) for all \( i \).

We define the characteristic function as
\bb
\chi_\rho \colon \RR^{2n} \to \CC, \quad \chi_{\rho}(\mathbf{r}) = \Tr[D_{\mathbf{r}}\rho].
\ee
where $D_{\mathbf{r}} = \exp(i \mathbf{r}^\intercal\Omega \mathbf{R})$ are the displacement operators.
The characteristic function is a continuous function and $\chi_{\rho}(0) = \Tr[\rho] = 1$.

The Wigner function $W_\rho: \mathbb{R}^{2n} \to \mathbb{R}$ of the $n$-mode state $\rho$ is defined as the inverse Fourier transform
of the characteristic function \( \chi_\rho \), evaluated at \( \Omega\mathbf{r} \)
\begin{equation}
    W_\rho(\mathbf{r}) \coloneqq \frac{1}{(2\pi)^{2n}} \int_{\mathbf{r}' \in \mathbb{R}^{2n}} \mathrm{d}^{2n}\mathbf{r}' \, \chi_\rho(\mathbf{r}') e^{i {\mathbf{r}'}^{\intercal} \Omega \mathbf{r}}.
\end{equation}
If $\rho$ is a pure state with wavefunction $\psi:\mathbb{R}^n\to \mathbb{C}$ in terms of the position coordinates, then $W_\rho(\mathbf{r})$ can be expressed as
\bb\label{eq:wigner_psi}
W_\rho(\mathbf{r}) = \frac{1}{\pi^{n}} \int_{\mathbf{x}' \in \mathbb{R}^{n}} \mathrm{d}^{n}\mathbf{x}' \, \psi^\ast\left(\mathbf{x}+\mathbf{x}'\right)\psi\left(\mathbf{x}-\mathbf{x}'\right) e^{2i {\mathbf{x}'}^{\intercal}\mathbf{p}} \qquad \mathbf{r}=(\mathbf{x},\mathbf{p}).
\ee

\subsubsection{Gaussian unitaries and Gaussian states}
A \emph{Gaussian unitary} \( G \) is a unitary operator generated by quadratic Hamiltonians in the quadrature operators \( \hat{\mathbf{R}} \). Any such unitary can be written as
\[
    G = U_S D_{\mathbf{r}},
\]
where \( D_{\mathbf{r}} \) is the displacement operator and \( U_S \) is the symplectic Gaussian unitary associated with a symplectic matrix \( S \).

The action of each operator on the quadrature vector \( \hat{\mathbf{R}} \) is given by
\bb
D_{\mathbf{r}}^{\dagger} \hat{\mathbf{R}} D_{\mathbf{r}} = \hat{\mathbf{R}} + \mathbf{r} \, \id, \quad
U_S^{\dagger} \hat{\mathbf{R}} U_S = S \hat{\mathbf{R}}.
\ee

Correspondingly, their action on the first moment \( \mathbf{m} \) and covariance matrix \( V \) is
\bb\label{eq:action gaussian unitaries}
\mathbf{m}(D_{\mathbf{r}} \rho D_{\mathbf{r}}^{\dagger}) &= \mathbf{m}(\rho) + \mathbf{r}, &\quad
V(D_{\mathbf{r}} \rho D_{\mathbf{r}}^{\dagger}) &= V(\rho), \\
\mathbf{m}(U_S \rho U_S^{\dagger}) &= S \mathbf{m}(\rho), &\quad
V(U_S \rho U_S^{\dagger}) &= S V(\rho) S^{\intercal}.
\ee

The characteristic function transforms as
\bb
\chi_{D_{\mathbf{r}} \rho D_{\mathbf{r}}^\dagger}(\mathbf{s}) = e^{i \mathbf{s}^\intercal \Omega \mathbf{r}} \chi_{\rho}(\mathbf{s}), \\
\chi_{U_S \rho U_S^\dagger}(\mathbf{s}) = \chi_{\rho}(S^{-1} \mathbf{s}).
\ee

The characteristic function of a Gaussian state $\rho$ can be written in terms of the mean vector $\mathbf{m}(\rho) $ and the covariance matrix $V\!(\rho)$ of the state $\rho$ as~\cite{BUCCO}
\bb\label{charact_gaussian}
\chi_{\rho}(\mathbf{r})=\exp\!\left( -\frac{1}{4}(\Omega\mathbf{r})^{\intercal}V\!(\rho)\Omega \mathbf{r}+i(\Omega\mathbf{r})^{\intercal}\mathbf{m}(\rho) \right)\,.
\ee
Therefore, the Wigner function of a Gaussian state is given by the Gaussian probability distribution:
\bb\label{eq:Wigner_gaussian}
W_\rho(\textbf{r})&=\frac{e^{- (\textbf{r}-\textbf{m}(\rho))^\intercal [V(\rho) ]^{-1}  (\textbf{r}-\textbf{m}(\rho))}}{\pi^{n}\sqrt{\det [V(\rho) ]}}=\mathcal{N}\!\left(\textbf{m}(\rho),\frac{V(\rho)}{2}\right)\!(\textbf{r})\,,
\ee
where $\mathcal{N}[\textbf{m},V]$ represents the Gaussian probability distribution with first moment $\textbf{m}$ and covariance matrix $V$:
\bb
\mathcal{N}(\textbf{m},V)(\textbf{r})&\coloneqq\frac{e^{-\frac12 (\textbf{r}-\textbf{m})^\intercal V^{-1}  (\textbf{r}-\textbf{m})}}{(2\pi)^{n}\sqrt{\det V }}\,.
\ee

We let $\gaussian$ represent the set of pure Gaussian states, and $\gaussian_0$ the set of zero-mean pure Gaussian states.

A pure state $\ket{\psi}
    \in\H_n$ is \emph{Gaussian}
\begin{itemize}
    \item if and only if its characteristic function is Gaussian;
    \item if and only if its symplectic eigenvalues are all equal to $1$;
    \item if and only if it can be written as the action of a Gaussian unitary $G$ on the vacuum state $\ket{0}$;
    \item if and only if its wavefunction is a Gaussian wavefunction~\cite{Arvind_1995}. Specifically, the wavefunction of a pure Gaussian state $\ket{\psi}$, defined as $\psi(\mathbf{x}) \coloneqq \braket{\mathbf{x}|\psi}$ for all $\mathbf{x}\in\mathbb{R}^n$ (with $\ket{\mathbf{x}}$ being a common generalised eigenvector of all position operators), can be written in terms of its first moment $\mathbf{m}=(\bar{\mathbf{x}},\bar{\mathbf{p}})$ and covariance matrix $V=\begin{pmatrix} V_{xx} & V_{xp} \\ V_{xp}^\intercal & V_{pp}\end{pmatrix}$ as
          \bb\label{this_is_the_wave}
          \psi(\mathbf{x}) \;=\; \left(\frac{\det W}{\pi^n}\right)^{1/4}
          \exp\!\Bigg(
          -\tfrac12 (\mathbf{x}-\bar{\mathbf{x}})^\intercal (W+iU)(\mathbf{x}-\bar{\mathbf{x}})
          + i\,\bar{\mathbf{p}}^\intercal(\mathbf{x}-\bar{\mathbf{x}})
          \Bigg),
          \qquad \forall\,\mathbf{x}\in\mathbb{R}^n,
          \ee
          where $W \coloneqq V_{xx}^{-1}$ and $U \coloneqq -\tfrac{1}{2}\!\left(V_{xx}^{-1}V_{xp}+V_{xp}^\intercal V_{xx}^{-1}\right)$. Conversely, any pure state with wavefunction of the form in Eq.~\eqref{this_is_the_wave}, where $\bar{\mathbf{x}},\bar{\mathbf{p}}\in\mathbb{R}^n$ and $W,U\in\mathbb{R}^{n\times n}$ are symmetric with $W$ strictly positive definite, is necessarily Gaussian.
\end{itemize}

\subsection{Review on CV trace-distance bounds}\label{sec:tr_distance}
In this section, we provide an overview of a recently developed toolbox regarding bounds on the trace distance between continuous-variable quantum states~\cite{mele2024learning,bittel2024optimal1,holevo2024estimates,symplectity_rank}. This new toolbox turns out to be crucial for our analysis of property testing of Gaussian states.

Specifically, we consider two (possibly non-Gaussian) quantum states $\rho$ and $\sigma$ with covariance matrices $V$ and $W$, and first moments $\mathbf{m}$ and $\mathbf{t}$, respectively. We will review how to bound the trace distance $\frac12\|\rho-\sigma\|_1$ in terms of the norms of $V-W$ and $\mathbf{m}-\mathbf{t}$. The bounds vary depending on whether $\rho$ and $\sigma$ are Gaussian states. Throughout this section, we will use the notation $\rho(V,\mathbf{m})$ to denote the Gaussian state with covariance matrix $V$ and first moment $\mathbf{m}$.

Let us start with the following upper bound on the trace distance between two Gaussian states in terms of the norm between their first moments and covariance matrices~\cite{bittel2024optimal1}.
\begin{lemma}[(Upper bound on the trace distance between Gaussian states~\cite{bittel2024optimal1})] \label{lem:upper_gaussian}
    The trace distance between two Gaussian states $\rho(V,\mathbf{m})$ and $\rho(W,\mathbf{t})$ can be upper bounded as follows:
    \begin{align}
        \frac{1}{2}\|\rho(V,\mathbf{m})-\rho(W,\mathbf{t})\|_1 & \le \frac{1+\sqrt{3}}{8}\max(\Tr V ,\Tr W) \,\|V-W\|_\infty+ \sqrt{\frac{\min(\|V\|_\infty,\|W\|_\infty)}{2} }\|\mathbf{m}-\mathbf{t}\|\,.\label{bound_hollllderb}
    \end{align}
\end{lemma}
This upper bound, proven to be tight~\cite{bittel2024optimal1}, represents a significant improvement over bounds previously established in~\cite{mele2024learning,holevo2024estimates}. Another similar upper bound has been recently introduced in~\cite{bittel2025energyindependent}, but it will not be useful for this work. To complete the overview of CV trace-distance bounds, we also mention that Ref.~\cite{mele2025achievableratesnonasymptoticbosonic} presents an algorithm for estimating the trace distance between Gaussian states up to a fixed precision.

What can we say about the case of non-Gaussian states? Of course, the trace distance between two arbitrary quantum states cannot be upper bounded solely in terms of their first moments and covariance matrices, as these quantities do not uniquely determine a general quantum state, unlike in the case of Gaussian states. However, it is possible to derive an upper bound in the case in which one of the two states is assumed to be a pure Gaussian state~\cite{mele2024learning}.
\begin{lemma}[(Upper bound on the trace distance between an arbitrary state and a pure Gaussian state~\cite{mele2024learning})] \label{lemma_upp_bound_purea}
    Let $\psi_G$ be a pure Gaussian state with covariance matrix $V$ and first moment $\mathbf{m}$. Moreover, let $\sigma$ be a (possibly non-Gaussian and possibly mixed) state with covariance matrix $W$ and first moment $\mathbf{t}$. Then, it holds that
    \bb\label{upper_bound_d_tr_pure}
    \frac12\|\psi_G-\sigma\|_1 \le \sqrt{E}\sqrt{ \|V-W\|_\infty+2\|\mathbf{m}-\mathbf{t}\|^2}\,,
    \ee
    where $E\coloneqq\max(\Tr[\sigma\hat{E}],\Tr[\psi_G\hat{E}])$ is the maximum energy and $\hat{E}$ denotes the energy operator.
\end{lemma}
The intuitive reason why it is possible to derive a bound between an arbitrary state and a pure Gaussian state, as shown in the above Lemma~\ref{lemma_upp_bound_purea}, lies in the simple-to-prove fact that if a state $\sigma$ has the same first moment and covariance matrix of a pure Gaussian state $\psi_G$, then $\sigma$ must also be equal to $\psi_G$. This fact also serves as the basic idea underlying our algorithm for property testing of Gaussian states in the pure-state setting.

Now, let us consider lower bounds on the trace distance, starting with the case of Gaussian states~\cite{mele2024learning}.
\begin{lemma}[(Lower bound on the trace distance between Gaussian states~\cite{mele2024learning})]\label{lemma:lower_trace}
    The trace distance between two Gaussian states $\rho(V,\mathbf{m})$ and $\rho(W,\mathbf{t})$ can be lower bounded in terms of the norm distance between their first moments and the norm distance between their covariance matrices as
    \bb\label{thm_trace_distance_lower_bound}
    \frac{1}{2}\|\rho(V,\mathbf{m})-\rho(W,\mathbf{t})\|_1&\ge \frac{1}{200}\min\!\left\{1,\frac{\| \mathbf{m}-\mathbf{t} \|}{\sqrt{4\min(\|V\|_\infty,\|W\|_\infty)+1}}\right\} \,,\\
    \frac{1}{2}\|\rho(V,\mathbf{m})-\rho(W,\mathbf{t})\|_1&\ge \frac{1}{200}\min\!\left\{1,\frac{\| V-W \|_2}{4\min(\|V\|_\infty,\|W\|_\infty)+1}\right\} \,.
    \ee
\end{lemma}
Finally, let us state a lower bound on the trace distance between two arbitrary states~\cite{symplectity_rank}.
\begin{lemma}[(Lower bound on the trace distance between arbitrary states~\cite{symplectity_rank})]\label{lower_bound_NG}
    Let $\rho$ be a (possibly non-Gaussian) state with first moment $\mathbf{m}$ and covariance matrix $V$. Moreover, let $\sigma$ be a (possibly non-Gaussian) state with first moment $\mathbf{t}$ and covariance matrix $W$. The trace distance can be lower bounded as
    \begin{align}
        \frac12\|\rho-\sigma\|_1 & \ge \frac{\|\mathbf{m}-\mathbf{t}\|^2}{32 \max\!\left({\Tr\!\left[\hat{E}  \rho \right]},{\Tr\!\left[\hat{E}\sigma\right]}\right)}\,,\label{ineq_first} \\
        \frac12\|\rho-\sigma\|_1 & \ge \frac{\|V
        -W\|_{\infty}^2}{3098\,\max(\Tr[\hat{E}^2\rho],\Tr[\hat{E}^2\sigma])}\label{ineq_cov}\,,
    \end{align}
    where $\hat{E}$ denotes the energy operator.
\end{lemma}

\section{Testing by symmetry}\label{sec:testing by symmetry appendix}

Let $\mathcal{G}$ be the set of pure Gaussian states. Leveraging the symmetries of Gaussian states, the aim of this section is to provide tests for \cref{prob:testing} (see also \cref{sec:intro_to_property_testing}), where $\mathcal{P}=\mathcal{G}$.

\subsection{Rotation invariance properties of Gaussian states}\label{sec:rotation}
We will now introduce a testing framework to determine if a given CV state is Gaussian. Initially, we will focus on the scenario where the mean is zero.
Consider two copies of $\H_n = L^2(\RR^n)$, then we have $\H_n^{\ot 2} = (L^2(\RR^n))^{\ot 2} \cong L^2(\RR^{n \times 2}) \cong L^2(\RR^{2n}) = \H_{2n}$.
We let
\begin{align*}
    r_{\theta} = \begin{pmatrix}
                     \cos(\theta) & -\sin(\theta) \\ \sin(\theta) & \cos(\theta)
                 \end{pmatrix}
\end{align*}
which acts on $\RR^2$ and hence on $\RR^{2n} = (\RR^n)^2 = \RR^n \ot \RR^2$ by rotating the two copies.
This defines a symplectic map on two copies of $\RR^{2n}$ by
\bb
\begin{pmatrix} \mathbf{r} \\ \mathbf{s} \end{pmatrix} \mapsto \begin{pmatrix}  \cos(\theta)\mathbf{r} - \sin(\theta) \mathbf{s} \\ \sin(\theta) \mathbf{r} + \cos(\theta)\mathbf{s} \end{pmatrix} .
\ee
We denote by $U_{\theta}$ the associated unitary on $\H_n^{\ot 2}$.
Concretely, if $\psi \in L^2(\RR^{2n})$, then for $\mathbf x, \mathbf y\in \RR^n$
\begin{align}\label{eq:U_theta concrete}
    (U_{\theta} \psi)(\mathbf x,\mathbf y) = \psi(\cos(\theta)\mathbf x + \sin(\theta)\mathbf y, -\sin(\theta)\mathbf x + \cos(\theta)\mathbf y).
\end{align}
For the characteristic functions, if $\rho$ is a state on two copies of $\H_n$, then
\begin{align}\label{eq:rotate characteristic function}
    \chi_{U_{\theta} \rho U_{\theta}^\dagger}(\mathbf r,\mathbf s) = \chi_{\rho}(\cos(\theta)\mathbf r + \sin(\theta)\mathbf s, -\sin(\theta)\mathbf r + \cos(\theta)\mathbf s)
\end{align}
for $\mathbf r,\mathbf s \in \RR^{2n}$.

It is easy to see that $U_\theta$ commutes with $U^{\ot 2}$ for every symplectic Gaussian unitary~$U$. This can be checked easily at the level of characteristic functions, using \cref{eq:action gaussian unitaries} and \cref{eq:rotate characteristic function}.
We can also see this as an instance of a more general phenomenon:
for any number of copies~$k$ (not just $k=2$), we have
\begin{align*}
    \H_n^{\ot k} = L^2(\RR^n)^{\ot k} \cong L^2(\RR^{n \times k}),
\end{align*}
where $\RR^{n \times k}$ are the real $n \times k$ matrices.
We have the tensor power action of $\Sp(2n,\RR)$ on this space, but also an action by $O(k)$, which from the perspective of $\H_n^{\ot k}$ acts as `rotations between the copies' and from the perspective of $L^2(\RR^{n \times k})$ simply by right multiplication on the variable.
Concretely, any $o \in O(k)$ acts by the unitary $U_o$ defined as
\begin{align}\label{eq:o action}
    (U_o \psi)(\mathbf x) = \psi(\mathbf x o) \qquad \psi \in L^2(\RR^{n \times k}), \, \mathbf x \in \RR^{n \times k},
\end{align}
generalizing \cref{eq:U_theta concrete}.
Let $O_s(k)$ denote the group of stochastic rotation matrices, preserving the vector $\mathbf 1_k = (1,\dots,1)^{\intercal}$, or equivalently, having rows summing to 1.

It is easy to see that the actions of $O(k)$ and $\Sp(2n,\RR)$ commute:

\begin{lemma}\label{lem:commuting two copies}
    Let $U$ be a symplectic Gaussian unitary on $\mathcal H_n$ and $o \in O(k)$ a real orthogonal $k\times k$-matrix.
    Then, the operators $U^{\ot k}$ and $U_o$ on $\mathcal H_n^{\ot k} \cong L^2(\RR^{n \times k})$ commute.
    If $o \in O_s(k)$ and $D_{\mathbf r}$ is a displacement operator, the operators $D_{\mathbf r}^{\ot k}$ and $U_o$ on $\mathcal H_n^{\ot k} \cong L^2(\RR^{n \times k})$ commute.
\end{lemma}
\begin{proof}
    Note that the action of $U^{\ot k}$ of $\Sp(2n)$ corresponds to a symplectic matrix of the form~$u \ot I_k \in \Sp(2nk)$, where $u \in \Sp(2n)$, while the rotation action corresponds to the symplectic matrix $I_{2n} \ot o$.
    Clearly, these commute.
    (Concretely, if we denote the bosonic operators on $k$ copies by an operator valued matrix $R = (R_{a,b})$, for $a=1,\dots,2n$ and $b=1,\dots,k$, then the action of $U^{\ot k}$ corresponds to $R \mapsto u R$, while $U_o$ corresponds to $R \mapsto R o^{\intercal}$.)
    The action on $R$ determines the transformation up to a scalar.

    For the claim about the displacement operator, note that $D_{\mathbf r}^{\ot k}$ acts as a displacement operator by $\mathbf r \ot \mathbf 1_k$.
    This commutes with the action of $I_{2n} \ot o$ if $o$ is stochastic.
\end{proof}

The relation between the actions of $U^{\ot k}$ for symplectic Gaussian unitaries and rotations $o \in O(k)$ is similar to the relation between the actions of $U^{\ot k}$ for $U \in \U(n)$ and $\pi \in S_k$ acting on $(\CC^n)^{\ot k}$ in Schur-Weyl duality. We explain the similar status, known as \emph{Howe duality} in more detail in \cref{sec:relation to rep theory}.

We say that $\psi \in L^2(\RR^{2n})$ is rotation-invariant if $U_{\theta} \psi = \psi$ for all $\theta \in [0,2\pi]$.
More generally, the action of $U_{\theta}$ defines a representation of $SO(2)$ on $\H^{\ot 2}$, which can be decomposed into subspaces $V_l$ for $l \in \ZZ$, where $U_{\theta}$ acts as $\exp( i l \theta)$ on $V_l$. 

We will now prove a close relation between rotation-invariance and Gaussianity.
Similar results characterizing when classical random variables are Gaussian are well-known \cite{kagan1975characterization}.
The relation between rotation-invariance and Gaussianity for bosonic quantum states has been used in the context of the quantum central limit theorem \cite{Hudson1974,becker2021convergence} and in studying extremality properties of Gaussian states \cite{wolf2006extremality}, using similar techniques as in below derivation.
We use the following Lemma.

\begin{lemma}\label{lem:rotation invariant quadratic}
    Let $f : \RR^n \to \CC$ be twice-differentiable at the origin.
    Let $\alpha = (\alpha_1,\dots,\alpha_k) \in \RR^k$, with $\norm{\alpha} = 1$.
    Suppose that $\abs{\alpha_i} \neq 1$ for all $i$. If
    \bb\label{eq:condition quadratic}
    f(\mathbf r) = \sum_{i=1}^k f(\alpha_i \mathbf r)   \qquad  \text{ for all } \mathbf r \in \RR^n
    \ee
    then $f$ is a quadratic polynomial. If $\alpha_1 + \dots + \alpha_k \neq 1$, the linear term of $f$ vanishes.
\end{lemma}

\begin{proof}
    Write a Taylor expansion $f(\mathbf r) = f_0 + r^\intercal f_1 + r^\intercal f_2 r + g(r)$ where $g(r) = o(\norm{\mathbf r}^2)$.
    If $\alpha_1 + \dots + \alpha_k \neq 1$, then taking the first derivative of \cref{eq:condition quadratic} gives that $f_1 = 0$.
    We now note that $\tilde f(\mathbf r) = f_0 + r^\intercal f_1 + r^\intercal f_2 r$ satisfies the condition $\tilde f(\mathbf r) = \sum_{i=1}^k \tilde f(\alpha_i \mathbf r)$ and hence $g(\mathbf r) = \sum_{i=1}^k g(\alpha_i \mathbf r)$.
    Choose $\eps > 0$. Then there exists $\delta$ such that for $\norm{\mathbf r} \leq \delta$, $g(\mathbf r) \leq \eps \norm{\mathbf r}^2$.
    Now let $\alpha = \max_i \abs{\alpha_i} < 1$.
    Then for $r \leq \delta/\alpha$, we have that $\norm{\alpha_i \mathbf r} \leq \delta$ for $i = 1,\dots,k$ and hence
    \bb
    g(\mathbf r) = \sum_{i=1}^k g(\alpha_i \mathbf r) \leq \sum_{i=1}^k \eps \norm{\alpha_i \mathbf r}^2 = \eps \norm{\mathbf r}^2
    \ee
    Iterating this, we find $g(\mathbf r) \leq \eps \norm{\mathbf r}^2$ for all $\mathbf r \in \RR^n$. Since $\eps > 0$ was arbitrary, $g(\mathbf r) = 0$.
\end{proof}

We now apply this lemma to the logarithm of the characteristic function to show a relation between rotation-invariance and Gaussianity.
We say that $o \in O(k)$ is \emph{nontrivial} if not all of its entries are in the set $\{0, \pm 1\}$.

\begin{boxed}{}
    \begin{thm}\label{thm:rotation invariant}
        Let $\rho$ be an arbitrary state on $\H_n$ with finite second moments.
        Then the following are equivalent for every $k \geq 2$:
        \begin{enumerate}
            \item\label{it:rot invariant} $U_o \rho^{\ot k} U_o^\dagger = \rho^{\ot k}$ for all $o \in O(k)$.
            \item\label{it:one rot invariant} $U_o \rho^{\ot k} U_o^\dagger = \rho^{\ot k}$ for a nontrivial $o \in O(k)$.
            \item\label{it:gaussian states} $\rho$ is Gaussian with zero mean.
        \end{enumerate}
        Similarly, the following are equivalent for every $k \geq 3$:
        \begin{enumerate}
            \item\label{it:rot invariant} $U_o \rho^{\ot k} U_o^\dagger = \rho^{\ot k}$ for all $o \in O_s(k)$.
            \item\label{it:quarter pi invariant} $U_o \rho^{\ot k} U_o^\dagger = \rho^{\ot k}$ for a nontrivial $o \in O_s(k)$.
            \item\label{it:equal gaussian states} $\rho$ is Gaussian with zero mean.
        \end{enumerate}
    \end{thm}
\end{boxed}

\begin{proof}
    In both cases, the implication (\ref{it:rot invariant}) $\Rightarrow$ (\ref{it:one rot invariant}) is trivial.
    The implication (\ref{it:gaussian states}) $\Rightarrow$ (\ref{it:rot invariant}) can be checked easily at the level of characteristic functions.
    The condition $U_o \rho^{\ot k} U_o^\dagger = \rho^{\ot k}$ for $o \in O(k)$ is equivalent to
    \bb\label{eq:rotation characteristic function}
    \prod_{j=1}^k \chi_{\rho}\mleft(\sum_{l=1}^k o_{lj} \mathbf r_l\mright) = \prod_{j=1}^k \chi_{\rho}(\mathbf r_j)
    \ee
    for $\mathbf r_j \in \RR^{2n}$. If the characteristic function is given as in \cref{charact_gaussian}, then this is invariant under all rotations is the mean is zero, and under all stochastic rotations if the mean is nonzero.
    Next, we assume (\ref{it:one rot invariant}).
    Since $\rho$ has finite second moments, we can expand the characteristic function as
    \begin{align*}
        \chi_{\rho}(r) = 1 + i\Tr[\mathbf{r}^\intercal\Omega \hat{\mathbf{R}} \rho] - \frac{1}{2}\Tr[(\mathbf{r}^\intercal\Omega \hat{\mathbf{R}})^2 \rho] + o(\norm{\mathbf r}^2)
    \end{align*}
    and $f(\mathbf r) = \log(\chi_{\rho}(\mathbf r))$ is twice-differentiable around $\mathbf r = 0$.
    Since $o$ is nontrivial, we may choose $l$ such that the vector given by $\alpha_i = o_{il}$ satisfies $\abs{\alpha_i} \neq 1$ for all $i$.
    Then setting $\mathbf r_j = \mathbf r$ for $j = l$ and $\mathbf r_j = 0$ for $j \neq l$ in \cref{eq:rotation characteristic function} gives
    \bb
    \sum_{i=1}^k f(\alpha_i \mathbf r) = f(\mathbf r)
    \ee
    which implies by \cref{lem:rotation invariant quadratic} that $f$ has to be quadratic (and hence $\chi_{\rho}$ is Gaussian).
    If $o$ is not stochastic, we may additionally choose $l$ similar to above such that $\alpha_1 + \dots + \alpha_k \neq 1$.
    There are two options: either $\abs{\alpha_i} \neq 1$ for all $i$, in which case we conclude from \cref{eq:rotation characteristic function} that the state has mean zero, or $\alpha_i = -1$ for some $i$, which implies $\chi_{\rho}(\mathbf r) = \chi_{\rho}(-\mathbf r)$ and hence the state has mean zero.
\end{proof}

Note that the above result holds for all Gaussian states, not just pure ones.
However, for testing purposes we are mostly interested in testing pure states because testing invariance under $U_\theta$ can be tested using similar ideas as the swap test.

\begin{cor}\label{cor:rotation invariant}$\phantom{.}$
    For a pure state $\psi$ on $\H_n$, the following are equivalent for $k \geq 2$.
    \begin{enumerate}[resume]
        \item\label{it:rot invariant pure} $U_o \ket\psi^{\ot k} =  \ket\psi^{\ot k}$ for all $o \in O(k)$.
        \item\label{it:quarter pi invariant pure} $U_o \ket\psi^{\ot k} =  \ket\psi^{\ot k}$ for some nontrivial $o \in O(k)$.
        \item\label{it:equal gaussian states pure} $\psi$ is a pure Gaussian state with zero mean.
    \end{enumerate}
    Similarly, for $k \geq 3$, the following are equivalent:
    \begin{enumerate}[resume]
        \item\label{it:rot invariant pure} $U_o \ket\psi^{\ot k} =  \ket\psi^{\ot k}$ for all $o \in O_s(k)$.
        \item\label{it:quarter pi invariant pure} $U_o \ket\psi^{\ot k} =  \ket\psi^{\ot k}$ for some nontrivial $o \in O_s(k)$.
        \item\label{it:equal gaussian states pure} $\psi$ is a pure Gaussian state.
    \end{enumerate}
\end{cor}

\begin{proof}
    In all three cases, the only implication that does not directly follow from \cref{thm:rotation invariant}, is (\ref{it:equal gaussian states pure}) $\Rightarrow$ (\ref{it:rot invariant pure}) (since a global phase is not excluded).
    This can be checked as follows: the vacuum state satisfies rotation invariance. Any zero-mean state can be written as $U_S \ket{0}$, and $U_S^{\ot k}$ commutes with the action of $O(k)$ by \cref{lem:commuting two copies}.
    An arbitrary pure Gaussian state can be written as $D_{\mathbf r} U_S \ket{0}$, and $(D_{\mathbf r}U_S)^{\ot k} $ commutes with the action of $O_s(k)$ by \cref{lem:commuting two copies}.
\end{proof}

We also note the following variant. This result is also proven in \cite{Cuesta2020,springer2009conditions}. For completeness, we provide a simple proof, similar to the result above.

\begin{thm}\label{thm:quantum ds theorem}
    Let $\rho$ be an arbitrary state on $\H_n$ with finite second moments.
    Let $U_{\theta}$ be as in \cref{eq:U_theta concrete}. Then the following are equivalent:
    \begin{enumerate}
        \item\label{it:all rotations} $U_{\theta} \rho^{\ot 2} U_{\theta}^\dagger$ is a product state for all $\theta \in [0,2\pi]$.
        \item\label{it:rotation prod state} $U_{\theta} \rho^{\ot 2} U_{\theta}^\dagger$ is a product state for some $\theta$ which is not a multiple of $\pi/2$.
        \item\label{it:gaussian nonzero mean} $\rho$ is a Gaussian state.
    \end{enumerate}
\end{thm}

\begin{proof}
    Again, the implication (\ref{it:all rotations}) $\Rightarrow$ (\ref{it:rotation prod state}) is trivial and the implication (\ref{it:gaussian nonzero mean}) $\Rightarrow$ (\ref{it:all rotations}) is a straightforward calculation.
    Assume (\ref{it:rotation prod state}), this means that for some angle $\theta$ we have
    \bb
    \chi_{\rho}(\cos(\theta) \mathbf r_1 + \sin(\theta) \mathbf r_2) \chi_{\rho}(-\sin(\theta) \mathbf r_1 + \cos(\theta) \mathbf r_2) = \chi_{\sigma_1}(\mathbf r_1) \chi_{\sigma_2}(\mathbf r_2)
    \ee
    for some states $\sigma_1, \sigma_2$.
    Setting $\mathbf r_1 = \mathbf r$, $\mathbf r_2 = 0$, or $\mathbf r_1 = 0$, $\mathbf r_2 = \mathbf r$, or $\mathbf r_1 = \cos(\theta)\mathbf r$, $\mathbf r_2 = \sin(\theta) \mathbf r$ yields respectively
    \bb
    \chi_{\sigma_1}(\mathbf r) &= \chi_{\rho}(\cos(\theta) \mathbf r ) \chi_{\rho}(-\sin(\theta) \mathbf r)
    \chi_{\sigma_2}(\mathbf r) &= \chi_{\rho}(\sin(\theta) \mathbf r ) \chi_{\rho}(\cos(\theta) \mathbf r)\\
    \chi_{\rho}(\mathbf r) &= \chi_{\sigma_1}(\cos(\theta)\mathbf r) \chi_{\sigma_2}(\sin(\theta)\mathbf r).
    \ee
    Combining these yields
    \bb
    \chi_{\rho}(\mathbf r) = \chi_{\rho}(\cos(\theta)^2\mathbf r)\chi_{\rho}(\sin(\theta)^2 \mathbf r)\chi_{\rho}(\cos(\theta)\sin(\theta)\mathbf r)\chi_{\rho}(-\cos(\theta)\sin(\theta)\mathbf r)
    \ee
    We may now apply \cref{lem:rotation invariant quadratic} to $f(\mathbf r) = \log \chi_{\rho}(\mathbf r)$ with
    \bb
    \alpha_1 = \cos(\theta)^2, \quad \alpha_2 = \sin(\theta)^2, \quad \alpha_3 = \cos(\theta)\sin(\theta), \quad \alpha_4 = -\cos(\theta)\sin(\theta), \quad.
    \ee
    Note that if $\theta$ is not a multiple of $\pi/2$, $\abs{\alpha_i} < 1$, and they sum to 1 for all $\theta$.
\end{proof}

\begin{cor}
    Let $\psi$ be a pure state. Then $\psi$ is Gaussian if and only if $\Tr_2[U_{\theta} \psi^{\ot 2} U_{\theta}^\dagger]$ is a pure state for some $\theta$ which is not a multiple of $\pi/2$.
\end{cor}

\subsection{Testing Gaussianity by rotation invariance}\label{sec:testing_rotation}

Based on the above characterizations of Gaussian states, we now describe a number of possible tests for Gaussianity.
We first describe the tests, and then state their correctness as \cref{thm:correctness tests}.

We consider tests that are based on measuring rotation invariance.
In general, if we have a (finite) group $G$, acting on a Hilbert space $\H$ by a unitary representation $U_g$ for $g \in G$,
\bb
P = \frac{1}{\abs{G}} \sum_{g \in G} U_g
\ee
is the projection onto the subspace of states invariant under the action of $G$.
This measurement can be implemented by the circuit in \cref{fig:invariance group action}.
In the special case where the group is $\ZZ/2\ZZ$, with $U_1 = F$ acting as the swap operator on two copies of a Hilbert space, this gives the swap test.

We start with the situation where we want to test whether the state is a zero-mean Gaussian state.
Let $\tilde V_0$ be the subspace of the symmetric subspace $\Sym^2(\H_n) \subseteq \H_n^{\ot 2}$ that is rotation-invariant.
In \cref{thm:span gaussian commutator} below, we state that this space equals the span of all states $\ket{\phi}^{\ot 2}$ over all zero-mean Gaussian states $\ket{\phi}$. This means that measuring rotation-invariance is the \emph{optimal} two-copy test with perfect completeness, as in \cref{eq:optimal perfect completeness}.

\begin{mdframed}
    \noindent\textbf{Test 1} (rotation invariance, mean zero).
    A test which takes two copies of a quantum state and measures $\{P_0,I-P_0\}$, where $P$ is the projection onto the rotation-invariant subspace $\tilde V_0$.
\end{mdframed}

This test can in principle be implemented as in \cref{fig:invariance group action}, but it would require preparing a state corresponding to a uniform superposition over all angles $\ket{\theta}$ (a quantum rotor state \cite{albert2017general}), and applying a rotation $U_{\theta}$ controlled on this angle.
This is probably challenging on physical systems.

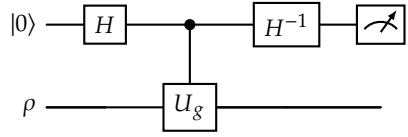
\begin{figure}[t]
    \centering
    \begin{tikzcd}
        \lstick{$\ket{0}$} & \gate{H} & \ctrl{1} & \gate{H^{-1}} & \meter{}\\
        \lstick{$\rho$} & \qw & \gate{U_{g}} & \qw & \qw
    \end{tikzcd}
    \caption{Quantum circuit for measuring invariance under the action of a group $G$ by $U_g$ for $g \in G$. Here, $H$ can be seen as the Fourier transform over the group, and prepares a uniform state $\abs{G}^{-1}\sum_{g \in G} \ket{g}$, and next this state controls the application of $U_g$.}
    \label{fig:invariance group action}
\end{figure}

An alternative is to only test invariance under rotations over an angle of $\pi/4$.
The main advantage of this is that it could be more practical to implement; it is however suboptimal compared to Test 1 (in the sense that if $\psi$ is not a Gaussian state, the acceptance probabilities satisfy $\Tr[P_0 \psi^{\ot 2}] \leq \Tr[P \psi^{\ot 2}] < 1$.

\begin{mdframed}
    \noindent\textbf{Test 2} ($\pi/4$-rotation invariance, mean zero).
    A test which takes two copies of a quantum state and measures $\{P,I-P\}$, where $P$ is the projection onto the subspace invariant under $U_{\pi/4}$.
\end{mdframed}

This test can now be implemented as in \cref{fig:invariance group action}, leading to the first circuit in \cref{fig:rotation test discrete}.
We only need to prepare an auxiliary system of dimension 8, and control the application of $U_{\pi/4}^k$ for $k = 0, 1, \dots, 7$.
In the special case where the state has even parity (a linear condition, which can be tested efficiently), it is easy to see that $U_{\pi/2}$ acts identically to a swap operator, and hence as the identity on the symmetric subspace.
That means that the second circuit in \cref{fig:rotation test discrete} suffices, which only has an auxiliary qubit, which controls whether $U_{\pi/4}$ (which in an optical system corresponds to a beam-splitter) is applied.
We call this variant Test 2', and in fact, when the state is not assumed to have even parity, this is also detected by Test 2', which we discuss in more detail in \cref{sec:analysis}.

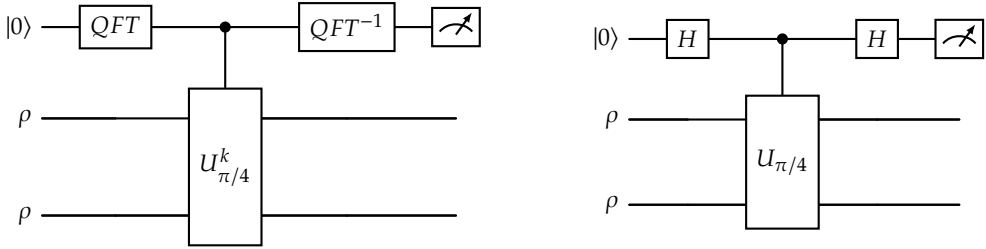
\begin{figure}[t!]
    \begin{subfigure}{0.4\textwidth}
        \begin{tikzcd}
            \lstick{$\ket{0}$} & \gate{QFT} & \ctrl{1} & \gate{QFT^{-1}} & \meter{}\\
            \lstick{$\rho$} & \qw & \gate[2]{U_{\pi/4}^k} & \qw & \qw \\
            \lstick{$\rho$} & \qw & \qw & \qw & \qw
        \end{tikzcd}
    \end{subfigure}
    \vspace*{1cm}
    \begin{subfigure}{0.4\textwidth}
        \begin{tikzcd}
            \lstick{$\ket{0}$} & \gate{H} & \ctrl{1} & \gate{H} & \meter{}\\
            \lstick{$\rho$} & \qw & \gate[2]{U_{\pi/4}} & \qw & \qw \\
            \lstick{$\rho$} & \qw & \qw & \qw & \qw
        \end{tikzcd}
    \end{subfigure}
    \caption{Quantum circuit for the rotation test over angle $\pi/4$. For the left diagram, the top wire has dimension $8$, with basis $\ket{k}$, $k = 0, 1, \dots, 7$, which controls the power by which $U_{\pi/4}$ is applied. The right diagram shows Test 2', a quantum circuit for an approximation of the rotation test over angle $\pi/4$, now with an auxiliary qubit.}
    \label{fig:rotation test discrete}
\end{figure}

We now discuss the case where we want to test Gaussian states which do not necessarily have mean zero.
We can do so based on \cref{thm:quantum ds theorem}, in a way that requires four rather than two copies.
Given four copies of a state $\rho$, let
\bb\label{eq:convolution mean zero state}
\sigma = \Tr_2[U_{\pi/4} \rho^{\ot 2} U_{\pi/4}^\dagger]
\ee
Now apply a swap test to two copies of $\sigma$ to determine whether $\sigma$ is pure.
This test is shown in \cref{fig:rotation test nonzero mean}.

\begin{mdframed}
    \noindent\textbf{Test 3} ($\pi/4$-rotation gives pure state).
    A test which takes four copies of a quantum state and measures with a swap test on two copies of $\sigma$ defined in \cref{eq:convolution mean zero state} whether $\sigma$ is pure.
\end{mdframed}

\begin{figure}[t]
    \centering
    \begin{tikzcd}
        \lstick{$\ket{0}$} & \gate{H} & \ctrl{2} & \gate{H} & \meter{}\\
        \lstick{$\rho$} & \gate[2]{U_{\pi/4}} & \qw & \qw \\
        \lstick{$\rho$} & \qw & \gate[2]{F} & \qw \\
        \lstick{$\rho$} & \gate[2]{U_{-\pi/4}} & \qw & \qw \\
        \lstick{$\rho$} & \qw & \qw & \qw
    \end{tikzcd}
    \caption{Quantum circuit for the rotation test in the case of non-zero mean. Here $F$ denotes a swap-operator.}
    \label{fig:rotation test nonzero mean}
\end{figure}
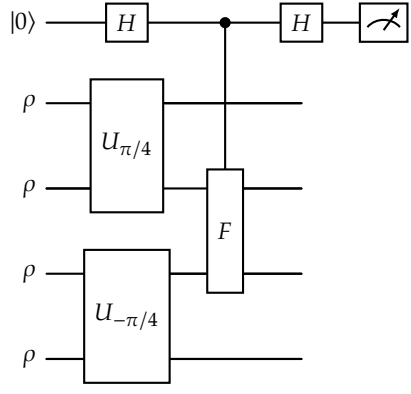

Alternatively, we can define 3-copy tests, based on \cref{thm:rotation invariant}.
In this case, we let $r_{\theta}$ denote the rotation over an angle $\theta$ along the axis $\frac{1}{\sqrt3} (1,1,1)^{\intercal}$.
Concretely, we let
\bb\label{eq:rotation 3d}
r_{\theta} = \frac13 \begin{pmatrix}
    1 + 2\cos(\theta)                     & (1- \cos(\theta) - \sqrt3\sin(\theta) & (1- \cos(\theta) + \sqrt3\sin(\theta) \\
    (1- \cos(\theta) + \sqrt3\sin(\theta) & 1 + 2\cos(\theta)                     & (1- \cos(\theta) - \sqrt3\sin(\theta) \\
    (1- \cos(\theta) - \sqrt3\sin(\theta) & (1- \cos(\theta) + \sqrt3\sin(\theta) & 1 + 2\cos(\theta).
\end{pmatrix}
\ee
which acts on $\RR^3$ and defines a symplectic unitary $U_\theta$ acting on $\H_n^{\ot 3}$.
As special cases we have
\bb\label{eq:square is permutation}
r_{2\pi/3} = \begin{pmatrix}
    0 & 0 & 1 \\
    1 & 0 & 0 \\
    0 & 1 & 0
\end{pmatrix} \, , \qquad
r_{\pi/3} = \frac13 \begin{pmatrix}
    2  & -1 & 2  \\
    2  & 2  & -1 \\
    -1 & 2  & 2
\end{pmatrix}.
\ee

\begin{mdframed}
    \noindent\textbf{Test 4} (rotation invariance).
    A test which takes three copies of a quantum state and measures whether the three copies are in the subspace invariant under $U_\theta$ as in \cref{eq:rotation 3d}.
\end{mdframed}
The analog of Test 2 is then the following.
\begin{mdframed}
    \noindent\textbf{Test 5} ($\pi/3$ rotation invariance).
    A test which takes three copies of a quantum state and measures whether the three copies are in the subspace invariant under $U_{\pi/3}$ as in \cref{eq:rotation 3d}.
\end{mdframed}

Finally, we give a different perspective, which may also give a practically useful test, and is convenient for analysis.
Here we consider the (unbounded) operator which is the generator of the group of rotations $U_{\theta}$ on two copies of $\H_n$:
\begin{align*}
    \left(
    i\partial_\theta|_0
    U_\theta \psi
    \right)(\mathbf x,\mathbf y)
     & =
    i\partial_\theta|_0
    \psi\big(
    \cos\theta\, \mathbf x + \sin\theta\, \mathbf y
    ,
    - \sin\theta\, \mathbf x + \cos\theta\, \mathbf y
    \big)                                                           \\
     & =
    -
    (\mathbf x \cdot i\nabla_{\mathbf y})\psi(\mathbf x,\mathbf y)
    +
    (\mathbf y \cdot i\nabla_{\mathbf x}) \psi(\mathbf x,\mathbf y) \\
     & =
    \left(\sum_{i=1}^n \hat{\mathbf x}_i \otimes \hat{\mathbf p}_i
    -
    \sum_{i=1}^n \hat{\mathbf p}_i \otimes \hat{\mathbf x}_i \right)
    \psi
    (\mathbf x,\mathbf y)
\end{align*}
and we define
\begin{align*}
    G : & = \sum_{i=1}^n \hat{\mathbf x}_i \otimes \hat{\mathbf p}_i
    -
    \sum_{i=1}^n \hat{\mathbf p}_i \otimes \hat{\mathbf x}_i          \\
        & = i \sum_{i=1}^n a_i \ot a_i^\dagger - a_i^\dagger \ot a_i.
\end{align*}

This operator commutes with $U^{\ot 2}$ for all Gaussian unitaries.
A state is rotation-invariant if and only if $G \ket{\psi} = 0$, that is, if $\bra{\psi}{G^2}\ket{\psi} = 0$, so this suggests the following.

\begin{mdframed}
    \noindent\textbf{Test 6} (Generator of rotations, two copies).
    A test that takes two copies of a state and measures the observable $G^2$, where the outcome 0 leads to acceptance.
\end{mdframed}

By \cref{thm:rotation invariant}, this test has the property that
\begin{align*}
    \langle G^2 \rangle_{\psi^{\ot 2}} = 0 \Leftrightarrow P_0 \ket{\psi}^{\ot 2} = \ket{\psi}^{\ot 2} \Leftrightarrow \psi \text{ is a pure zero-mean Gaussian state.}
\end{align*}
where $P_0$ is the projector defined in Test 1.

We can take the same approach for three copies, where we take the generator for the rotation corresponding to \cref{eq:rotation 3d}.

In this case, we have a generator $\tilde G$ given by
\begin{align*}
    \tilde G : & = \sum_{i=1}^n \hat{\mathbf x}_i \otimes \hat{\mathbf p}_i \ot \id
    - \hat{\mathbf p}_i \otimes \hat{\mathbf x}_i  +  \id \ot \hat{\mathbf x}_i \otimes \hat{\mathbf p}_i
    - \id \ot \hat{\mathbf p}_i \otimes \hat{\mathbf x}_i +  \hat{\mathbf p}_i \ot \id \ot \hat{\mathbf x}_i
    -  \hat{\mathbf x}_i  \ot \id \ot \hat{\mathbf p}_i                             \\
               & = G^{(12)} + G^{(23)} + G^{(31)}.
\end{align*}
where $G^{(ij)}$ is acting as $G$ on copies $i$ and $j$.
\begin{mdframed}
    \noindent\textbf{Test 7} (Generator of rotations, three copies).
    A test that takes three copies of a state and measures the observable $\tilde G^2$, where the outcome 0 leads to acceptance.
\end{mdframed}

By \cref{thm:rotation invariant} this test has the property that
\begin{align*}
    \langle \tilde G^2 \rangle_{\psi^{\ot 3}} = 0 \Leftrightarrow \psi \text{ is a pure Gaussian state.}
\end{align*}
and $\tilde G$ is closely related to Test 4.

We now state the correctness of these tests; all have perfect completeness.

\begin{thm}\label{thm:correctness tests}
    Tests 1,2 and 6 accept a state $\rho$ with probability 1 if and only if the state is a pure zero-mean Gaussian state. Tests 3,4,5 and 7 accept a state $\rho$ with probability 1 if and only if the state is a pure Gaussian state.
\end{thm}

\begin{proof}
    By the first part of \cref{thm:rotation invariant}, Tests 1 and 2 accept a pure state $\psi$ with probability 1 if and only if it is a zero-mean Gaussian state. The second part implies that Tests 4 and 5 accept a pure state $\psi$ with probability 1 if and only if it is a Gaussian state. By \cref{thm:quantum ds theorem}, Test 3 accepts a pure state $\psi$ with probability 1 if and only if it is a Gaussian state.
    It now suffices to argue that these tests have a non-zero rejection probability when applied to a state that is not pure.
    To this end, we note that the acceptance probability of Test 1 is at most that of Test 2. We show in \cref{lem:purity} that if $\Tr[\rho^2] \neq 1$ (so the state is not pure), Test 2 has nonzero rejection probability.
    The acceptance probability of Test 4 is at most that of Test 5. By \cref{eq:square is permutation}, invariance under a rotation over $\pi/3$ implies invariance under a cyclic permutation of the copies, which implies $\Tr[\rho^3] = 1$ and hence the state is pure.
    The fact that Test 4 does not accept with probability 1 if $\rho$ is not pure follows directly from \cref{thm:quantum ds theorem}.
    Finally, Tests 6 and 7 accept $\rho$ with probability 1 if and only if Test 1 and Test 4 accept with probability 1 respectively.
\end{proof}

\subsection{Analysis}\label{sec:analysis}
We now proceed to analyze the approach described in the previous section.
We have already proven in \cref{thm:correctness tests} that they have a non-zero probability of rejecting if and only if the state is not Gaussian.
Our goal in this section will be to prove that Tests 1 and 4 have a good robustness properties, and can be turned into tolerant tests.
In case we are testing whether the state is close to a zero-mean state, we can assume that the state being tested has mean zero as well, or the state is even parity, which implies in particular the state has zero mean.
This is a harmless assumption, since it is a linear condition and can be easily tested.
We start with Test 6 and 7, which will then serve to prove the robustness of Tests 1 and 4.

\begin{lemma}\label{lem:norm symplectic spectrum}
    \begin{enumerate}
        Let $\rho$ be a state on $\H_n$. Let $D = (\nu_1, \dots, \nu_n)$ be the symplectic eigenvalues of the covariance matrix of $\rho$.
        \item\label{it:bound G} Suppose $\rho$ has mean zero.
              Then
              \begin{align*}
                  \angles{G^2}_{\rho^{\ot 2}} = \frac12 \sum_{j=1}^n (\nu_j^2 - 1).
              \end{align*}
        \item\label{it:bound G three copies}
              For arbitrary $\rho$,
              \begin{align*}
                  \angles{\tilde G^2}_{\rho^{\ot 3}} = \frac32 \sum_{j=1}^n (\nu_j^2 - 1).
              \end{align*}
    \end{enumerate}
\end{lemma}

\begin{proof}
    Because $G$ commutes with $U \ot U$ for any Gaussian symplectic unitary, we may assume that the covariance matrix of~$\rho$ is in Williamson normal form.
    Since we assume the state has vanishing first moments, this means that there are numbers $\nu_j$ for $j\in[n]$ such that
    \begin{align*}
        \angles{\hat{\mathbf x}_j^2} = \angles{\hat{\mathbf p}_j^2} & = \frac12 \nu_j > 0
    \end{align*}
    and moreover
    \begin{align*}
        \angles{\frac12 \{ \hat{\mathbf x}_j, \hat{\mathbf p}_j \}} & = 0 \text{ and hence}\quad \angles{\hat{\mathbf x}_j \hat{\mathbf p}_j} = -\angles{\hat{\mathbf p}_j \hat{\mathbf x}_j} = \frac i 2
    \end{align*}
    and for $j \neq k$
    \begin{align*}
        \angles{\hat{\mathbf x}_j \hat{\mathbf x}_k}                & = \angles{\hat{\mathbf p}_j \hat{\mathbf p}_k} = 0,                              \\
        \angles{\frac12 \{ \hat{\mathbf x}_j, \hat{\mathbf p}_k \}} & = 0 \quad\text{and hence}\quad \angles{\hat{\mathbf x}_j \hat{\mathbf p}_k} = 0.
    \end{align*}
    It follows that
    \begin{align*}
        \Var[G]
         & = \angles{G^2}
        = \sum_{j,k} \angles*{ (\hat{\mathbf x}_j \ot \hat{\mathbf p}_j - \hat{\mathbf p}_j \ot \hat{\mathbf x}_j) (\hat{\mathbf x}_k \ot \hat{\mathbf p}_k - \hat{\mathbf p}_k \ot \hat{\mathbf x}_k) }    \\
         & = 2 \sum_{j,k} \angles{\hat{\mathbf x}_j\hat{\mathbf x}_k} \angles{\hat{\mathbf p}_j\hat{\mathbf p}_k} - \angles{\hat{\mathbf x}_j\hat{\mathbf p}_k} \angles{\hat{\mathbf p}_j\hat{\mathbf x}_k} \\
         & = \frac12 \sum_j \parens[\Big]{ 4 \angles{\hat{\mathbf x}_j^2} \angles{\hat{\mathbf p}_j^2} - 1 }                             = \frac12 \sum_j \parens[\Big]{\nu_j^2 - 1 }.
    \end{align*}
    This proves (\ref{it:bound G}).
    For $\tilde G$, we note that $\tilde G$ not only commutes with $U^{\ot 3}$ for Gaussian symplectic unitaries $U$, but also with $D_{\mathbf{r}}^{\ot 3}$ for all displacement operators $D_{\mathbf r}$.
    This means that we can again assume the state has vanishing first moments and covariance matrix is in the same Williamson normal form.
    It is now easy to see that
    \bb
    \langle \tilde G^2 \rangle = \langle  (G^{(12)} + G^{(23)} + G^{(31)})^2 \rangle = 3 \langle  G^2 \rangle
    \ee
    since terms like $\langle G^{(12)}G^{(23)} \rangle$ vanish if the first moments vanish.
    The result in (\ref{it:bound G three copies}) now follows from (\ref{it:bound G}).
\end{proof}

A state is a pure Gaussian state if and only if all its symplectic eigenvalues satisfy $\nu_j^2 = 1$, corresponding to $\angles{G^2} = 0$. If $\angles{G^2}$ is small, then the state must be close to a pure Gaussian state.

\begin{lemma}[(Approximate rotation invariance implies nearby Gaussian state)]\label{lem:rotation invariance G}
    Let $\rho$ be a state on $L^2(\RR^n)$ with vanishing first moments such that $\angles{G^2}_{\rho^{\ot 2}} \leq \eps$.
    Then there is a centered pure Gaussian state~$\psi$ such that
    \begin{align*}
        F(\psi,\rho)^2 = \bra{\psi} \rho \ket{\psi} \geq 1 - \frac{\eps}{2}.
    \end{align*}
\end{lemma}

\begin{proof}
    As in the proof of \cref{lem:norm symplectic spectrum} we may assume the covariance matrix is in Williamson normal form, which gives
    \begin{align*}
        \angles{G^2}
         & = \frac12 \sum_j \parens[\Big]{ \nu_j^2 - 1 }                                               \\
         & = \frac12 \sum_j \parens[\Big]{ \angles{\hat{\mathbf x}_j^2 + \hat{\mathbf p}_j^2}^2 - 1 }.
    \end{align*}
    Recall that the harmonic oscillator $H = \hat{\mathbf p}^2 + \hat{\mathbf x}^2$ has a Gaussian ground state~$\ket0$, with eigenvalue 1, and all other eigenvalues are $\geq3$.
    That is, $H \geq \proj0 + 3 (I - \proj0) = 3 - 2 \proj0$, and hence we have, with $\rho_j = \Tr_{j^c}[\rho]$ the reduced states,
    \begin{align*}
        \angles{G^2} & = \frac12 \sum_j \parens[\Big]{ \Tr[H \rho_j]^2 - 1 }                                    \\
                     & \geq \frac12 \sum_j \parens[\Big]{ (3 - 2 \angles{0|\rho_j|0})^2 - 1 }                   \\
                     & = \frac12 \sum_j \parens[\Big]{ \parens[\big]{ 1 + 2 (1 - \angles{0|\rho_j|0}) }^2 - 1 } \\
                     & \geq 2 \sum_j \parens[\Big]{ 1 - \angles{0|\rho_j|0} }.
    \end{align*}
    Since $I - \proj0^{\ot n} \leq \sum_{j=1}^n I - \proj0_j \ot I_{j^c}$, it holds that
    \begin{align*}
        2 \sum_j \parens[\Big]{ 1 - \angles{0|\rho_j|0} }
        \geq 2 \parens[\Big]{ 1 - \angles{0^{\ot n} | \rho | 0^{\ot n}} }.
    \end{align*}
    Thus, $\angles{G^2} \leq \eps$ implies that $\angles{0^{\ot n} | \rho | 0^{\ot n}} \geq 1 - \eps/2$.
\end{proof}

This shows that if Test 6 accepts with high probability, the state is close to Gaussian.
Next we use this to give bounds for Test 1.
Our goal is to show that if the test rejects with probability bounded away from 1, then we can also bound the distance to $\gaussian$, the set of Gaussian states.

\begin{lemma}\label{lem:C}
    Let $\psi$ be a pure state, and let $C$ be a constant bounding
    \begin{align*}
        \abs{\bra{\psi}^{\ot 2} G^4 \ket{\psi}^{\ot 2}} \leq C.
    \end{align*}
    \begin{enumerate}
        \item If $\ket{\psi}$ has mean zero, $P$ is the accepting projector of Test 1, and Test 1 accepts with high probability on $\psi$
              \begin{align*}
                  p_{\accept} = \Tr[\psi^{\ot 2} P] \geq 1 - \eps,
              \end{align*}
              there exists a zero mean Gaussian state $\ket{\phi}$ such that $\abs{\angles{\psi|\phi}} \geq 1 - \bigO(\sqrt{C \eps})$.
        \item If $P$ is the accepting projector of Test 4, and Test 4 accepts with high probability on $\psi$
              \begin{align*}
                  p_{\accept} = \Tr[\psi^{\ot 3} P] \geq 1 - \eps,
              \end{align*}
              there exists a Gaussian state $\ket{\phi}$ such that $\abs{\angles{\psi|\phi}} \geq 1 - \bigO(\sqrt{C \eps})$.

    \end{enumerate}

\end{lemma}

\begin{proof}
    If $\Tr[\psi^{\ot 2} P] \geq 1 - \eps$, then there exists $\ket{\phi} \in V_0$ such that $\abs{\bra{\phi} (\ket{\psi}^{\ot 2})} \geq 1 - \eps$.
    In particular, $G \ket{\phi} = 0$, which we can use to see that
    \begin{align*}
        \langle G^2 \rangle_{\psi} = \bra{\psi}^{\ot 2} G^2 ( \ket{\psi}^{\ot 2} - \ket{\phi}).
    \end{align*}
    We may now use the Cauchy-Schwarz inequality to bound
    \begin{align*}
        \langle G^2 \rangle_{\psi} \leq \norm{G^2 \ket{\psi}^{\ot 2}} \norm{\ket{\psi}^{\ot 2} - \ket{\phi}} \leq \sqrt{2C\eps}.
    \end{align*}
    The result now follows from \cref{lem:rotation invariance G}.
    The argument for Test 4 is the same (but using $\tilde G$).
\end{proof}

We now bound the fourth moment of $G$ in terms of moments of the energy. Let
\bb
E_2(\rho) = \frac1n \min_{S} \sqrt{\Tr[U_S \rho U_S^\dagger \hat E^2]}
\ee
where the minimum is over all $S \in \Sp(2n,\RR)$.

\begin{prop}[(New inequality on $G$)]\label{prop:bound generator} The following inequality holds:
    \bb
    \bra{\psi}^{\otimes 2}G^4\ket{\psi}^{\otimes 2}
    \leq \left(4n^2 E_2(\psi)^2+2n\right)^2.
    \ee
\end{prop}

\newcommand{\rr}{\hat{\mathbf R}}

\begin{proof}
    We prove that
    \bb
    \bra{\psi}^{\otimes 2}G^4\ket{\psi}^{\otimes 2}
    \leq \left(4\braket{\psi|\hat E^2|\psi}+2n\right)^2.
    \ee
    Since $G$ commutes with $U_S^{\ot 2}$, this implies the result.
    We use the following notation:
    \bb
    G : & = \sum_{i=1}^n \hat{\mathbf x}_i \otimes \hat{\mathbf p}_i-\sum_{i=1}^n \hat{\mathbf p}_i \otimes \hat{\mathbf x}_i
    = \sum_{i=1}^{2n} s_i\hat{\mathbf R}_i \otimes \hat{\mathbf R}_{i\oplus n},
    \ee
    where
    \bb
    \hat{\mathbf R}\coloneqq(\hat{\mathbf x}_1,\dots,\hat{\mathbf x}_n,\hat{\mathbf p}_1,\dots, \hat{\mathbf p}_n)^\intercal, \qquad i\oplus n \coloneqq i+j \mod{2n}, \qquad s_i=\begin{cases}
        +1 & 1\leq i\leq n    \\
        -1 & n+1\leq i\leq 2n
    \end{cases}.
    \ee
    \textbf{Claim 1.} The following operator identity holds:
    \bb
    \sum_{l,k=1}^{2n}\rr_l\rr_k^2\rr_l=4\hat E^2 + 2n\id,
    \ee
    where $\hat E\coloneqq\frac{1}{2}\rr^\intercal \rr$.\\

    \noindent\textit{Proof of Claim 1.} We notice that
    \bb\label{eq:comm}
    [A,B^2]&=[A,B]B+B[A,B]\\
    [\rr_k,\rr_l]&=i\Omega_{kl}\id\\
    \Rightarrow \;[\rr_k,\rr_l^2]&=2i\Omega_{kl}\rr_l.
    \ee
    Furthermore, remembering also that
    \bb\label{eq:omega}
    \Omega^{\intercal}=-\Omega,\qquad \text{and}\qquad \rr^\intercal \Omega^\intercal \rr = -in\id,
    \ee
    we get
    \bb
    \sum_{l,k=1}^{2n}\rr_l\rr_k^2\rr_l&
    =\sum_{l,k=1}^{2n}\left(\rr_k^2\rr_l+[\rr_l,\rr_k^2]\right)\rr_l\\
    &=\sum_{l,k=1}^{2n}\rr_k^2\rr_l^2+\sum_{l,k=1}^{2n}2i\Omega_{lk}\rr_k\rr_l\\
    &=(\rr^\intercal\rr)^2 +2i \rr^\intercal \Omega^\intercal \rr\\
    &=4\hat E^2 + 2n\id,
    \ee
    and this concludes the proof of the claim. \hfill$\blacksquare$\\

    Now,
    \begin{align*}
        |\bra{\psi}^{\otimes 2}G^4\ket{\psi}^{\otimes 2}|
         & \leqt{(i)} \sum_{i,j,k,l=1}^{2n}|\braket{\psi|\rr_i\rr_j\rr_k\rr_l|\psi}|\,|\braket{\psi|\rr_{i\oplus n}\rr_{j\oplus n}\rr_{k\oplus n}\rr_{l\oplus n}|\psi}| \\
         & \leqt{(ii)} \sqrt{\sum_{i,j,k,l=1}^{2n}|\braket{\psi|\rr_i\rr_j\rr_k\rr_l|\psi}|^2}
        \sqrt{\sum_{i,j,k,l=1}^{2n}|\braket{\psi|\rr_{i\oplus n}\rr_{j\oplus n}\rr_{k\oplus n}\rr_{l\oplus n}|\psi}|^2}                                                 \\
         & =\sum_{i,j,k,l=1}^{2n}|\braket{\psi|\rr_i\rr_j\rr_k\rr_l|\psi}|^2                                                                                            \\
         & \leqt{(iii)}\sum_{i,j,k,l=1}^{2n}\braket{\psi|\rr_i\rr_j^2\rr_i|\psi}\braket{\psi|\rr_l\rr_k^2\rr_l|\psi}                                                    \\
         & =\braket{\psi|\sum_{i,j=1}^{2n}\rr_i\rr_j^2\rr_i|\psi}\braket{\psi|\sum_{k,l=1}^{2n}\rr_l\rr_k^2\rr_l|\psi}                                                  \\
         & \eqt{(iv)} \braket{\psi|4\hat E^2+2n\id|\psi}^2                                                                                                              \\
         & =\left(4\braket{\psi|\hat E^2|\psi}+2n\right)^2
    \end{align*}
    where (i) is just triangle inequality, (ii) and (iii) are Cauchy-Schwarz inequality, and (iv) follows from Claim 1. This concludes the proof.
\end{proof}

The following is now a direct consequence of \cref{lem:C}, \cref{prop:bound generator} and the Fuchs-van de Graaf inequalities.

\begin{cor}\label{cor:accept implies close to gaussian}
    Let $\psi$ be a pure state, and suppose that $\E_2(\rho) \le E$.
    \begin{enumerate}
        \item If $\ket{\psi}$ has mean zero, and Test 1 accepts with probability $p_{\accept} = 1 - \eps$ on $\psi$, then $d_{\gaussian_0}(\psi) = \bigO(nE \eps^{1/4})$.
        \item If Test 4 accepts with probability $p_{\accept} = 1 - \eps$ on $\psi$, then $d_{\gaussian_0}(\psi) = \bigO(nE \eps^{1/4})$.
    \end{enumerate}
\end{cor}

We can give a sharper bound on the relation between $p_{\accept}$ and closeness to Gaussian if $\psi$ is constantly close to a Gaussian state.
We will now prove this for Test 2 (and hence Test 1), and for Test 5 (and hence Test 4).

\begin{thm}\label{thm:close to gaussian}$\phantom{.}$
    \begin{enumerate}
        \item\label{it:close to gaussian mean zero} Suppose that a even-parity pure state $\psi$ is $\eps$-far from a zero-mean Gaussian state, so $d_{\gaussian_0}(\psi) = \eps$, for $\eps \leq \eps_0$. Then Test 2 (and hence Test 1) has acceptance probability
              \begin{align*}
                  p_{\accept}(\psi) = 1 - \Omega(\eps^2).
              \end{align*}
        \item\label{it:close to gaussian mean nonzero} Suppose that a pure state $\psi$ is $\eps$-far from Gaussian, so $d_{\gaussian}(\psi) = \eps$, for $\eps \leq \eps_0$. Then Test 5 (and hence Test 4) has acceptance probability
              \begin{align*}
                  p_{\accept}(\psi) = 1 - \Omega(\eps^2).
              \end{align*}
    \end{enumerate}
\end{thm}

\begin{proof}
    We start with proving (\ref{it:close to gaussian mean zero}).
    If $d_{\gaussian_0}(\psi) = \eps$, we may write
    \begin{align}\label{eq:gaussian + pert}
        \ket{\psi} = \sqrt{1 - \eps^2} \ket{\phi} + \eps\ket{\alpha}
    \end{align}
    for a zero-mean Gaussian $\ket{\phi}$ and $\braket{\phi | \alpha} = 0$.
    Since the test is invariant under Gaussian unitaries, we may take $\ket{\phi} = \ket{0}$ without loss of generality.

    The idea is to expand $p_{\accept}(\psi) = \Tr[P \psi^{\ot 2}]$ according to \cref{eq:gaussian + pert}.
    Since $\phi$ is Gaussian,
    \begin{align*}
        \Tr[P \phi^{\ot 2}] = 1.
    \end{align*}
    Next, we observe that
    \begin{align*}
        \bra{\phi}\bra{\alpha} P \ket{\phi} \ket{\phi} = 0
    \end{align*}
    since $P \ket{\phi}^{\ot 2} = \ket{\phi}^{\ot 2}$ and $\ket{\alpha}$ is orthogonal to $\ket{\phi}$.
    Additionally, since $V_0$ is contained in the symmetric subspace, $P \ket{\phi}\ket{\alpha} = P \ket{\alpha} \ket{\phi}$.
    Taking this into account we see that
    \begin{align*}
        p_{\accept}(\psi) = \Tr[P \psi^{\ot 2}] & = (1 - \eps^2)^2 + 4\eps^2(1-\eps^2) \Tr[P (\alpha \ot \phi)]      \\
                                                & \qquad + \bigO(\eps^3)                                             \\
                                                & = 1 - 2\eps^2 + 4\eps^2 \Tr[P  (\alpha \ot \phi)] + \bigO(\eps^3).
    \end{align*}
    To prove the result, it suffices to bound $\Tr[P (\alpha \ot \phi)] \leq c$ for a constant $c < \frac12$.

    We may expand
    \bb\label{eq:expand alpha}
    \ket{\alpha} = \sum_{\mathbf m} \alpha_{\mathbf m} \ket{\mathbf m}.
    \ee
    The coefficient $\alpha_{\mathbf 0}$ is zero (since $\alpha$ is orthogonal to $\phi$, which we assumed to be the vacuum).
    We write $\abs{\mathbf m} = m_1 + \dots + m_n$.
    By assumption, the state only has support on the even parity subspace, so $\alpha_{\mathbf m} = 0$ for $\abs{\mathbf{m}}$ odd.
    Additionally, we claim that $\alpha_{\mathbf m} = 0$ for $\abs{\mathbf{m}} = 2$ (so we only have nonzero coefficients for $\abs{\mathbf m} \geq 4$).
    The proof of this claim is that since $\phi$ is the \emph{closest} zero-mean Gaussian state, for any quadratic Hamiltonian $H$, we must have
    \bb\label{eq:critical point}
    \partial_t \abs{\bra{0} \exp(i H t) \ket{\psi}}^2 \Big\vert_{t = 0} = 0.
    \ee
    Expanding according to \cref{eq:gaussian + pert} we get
    \begin{align*}
        \partial_t \abs{\bra{0} \exp(i H t) \ket{\psi}}^2 \Big\vert_{t = 0} = \eps \sqrt{1 - \eps^2} \text{Re}(\bra{0} H \ket{\alpha}).
    \end{align*}
    Taking $H = a_i^\dagger a_j^\dagger + a_j a_i$ then shows that $\alpha_{\mathbf m} = 0$ for $\abs{\mathbf m} = 2$ and $m_i = m_j = 1$ (or $m_i = 2$ if $i = j$).

    Given two copies of $\H_n$, we denote by $a_{i,1}$ and $a_{i,2}$ for $i = 1, \dots, n$ the annihilation operators on the first and second copy respectively.
    Let $\hat m_i = a_{i,1}^\dagger a_{i,1} + a_{i,2}^\dagger a_{i,2}$ be the particle number operator on mode $i$ on both copies.
    Then $[\hat m_i, P] = 0$, since $[\hat m_i, U_{\theta}]$ for all $\theta$.
    This means that
    \begin{align*}
        \bra{\mathbf m, \mathbf 0} P \ket{\mathbf n, \mathbf 0} = 0
    \end{align*}
    if $\mathbf m \neq \mathbf n$ and hence
    \begin{align*}
        \bra{\alpha} \bra{\phi} P \ket{\alpha} \ket{\phi} & = \sum_{\mathbf m} \abs{\alpha_{\mathbf m}}^2 \bra{\mathbf m, \mathbf 0} P \ket{\mathbf m, \mathbf 0} \\
                                                          & \leq \max_{\abs{\mathbf m} \geq 4} \bra{\mathbf m, \mathbf 0} P \ket{\mathbf m, \mathbf 0}.
    \end{align*}
    We now write
    \begin{align*}
        P = \frac{1}{8}\sum_{l=0}^7 U_{\pi/4}^l = \frac{1}{8}\sum_{l=0}^7 U_{l\pi/4}.
    \end{align*}
    It is easy to see that for $l = 2, 6$ we have
    \begin{align*}
        \bra{\mathbf m, \mathbf 0} U_{\pi/2} \ket{\mathbf m, \mathbf 0} = \bra{\mathbf m, \mathbf 0} U_{3\pi/2} \ket{\mathbf m, \mathbf 0} = 0
    \end{align*}
    by using that we have even parity states.
    For $l = 0, 4$ we have
    \begin{align*}
        \bra{\mathbf m, \mathbf 0} U_{0} \ket{\mathbf m, \mathbf 0} = \bra{\mathbf m, \mathbf 0} U_{\pi} \ket{\mathbf m, \mathbf 0} = 1.
    \end{align*}
    The case of odd $l$ remains, we take $l = 1$ (the other ones are similar).
    First, note that
    \begin{align*}
        \bra{\mathbf m, \mathbf 0} U_{\pi/4} \ket{\mathbf m, \mathbf 0} = \prod_{i=1}^n \bra{m_i,0} U_{\pi/4} \ket{m_i, 0}
    \end{align*}
    where on the right-hand side $U_{\pi/4}$ acts on two copies of a single mode.
    We are now going to bound the terms in the product.
    Consider a two copies of a single mode with annihilation operators $a_1, a_2$, then
    \begin{align*}
        \bra{m,0} R_{\pi/4} \ket{m,0} = \bra{m,0} R_{\pi/4} \frac{(a_1^\dagger)^m}{\sqrt{m!}} \ket{0,0} \\
        \bra{m,0} \frac{(R_{\pi/4} a_1^\dagger R_{\pi/4}^\dagger)^m}{\sqrt{m!}} \ket{0,0}
    \end{align*}
    using that $R_{\pi/4} \ket{0,0} = \ket{0,0}$.
    Now $R_{\pi/4} a_1^\dagger R_{\pi/4}^\dagger = \frac{1}{\sqrt2}(a_1^\dagger + a_2^\dagger)$ and hence
    \begin{align*}
        \bra{m,0} R_{\pi/4} \ket{m,0} = \bra{m,0}\left(\frac{1}{2^{m/2}} \frac{(a_1^\dagger)^m}{\sqrt{m!}} + \ldots\right) \ket{0,0} = \frac{1}{2^{m/2}}.
    \end{align*}
    where we omitted the terms lower order in $a_1^\dagger$ which evaluate to zero.
    We conclude that
    \begin{align*}
        \bra{\mathbf m, \mathbf 0} U_{\pi/4} \ket{\mathbf m, \mathbf 0} = 2^{-\abs{\mathbf m} / 2} \leq \frac{1}{4}
    \end{align*}
    since $\abs{\mathbf m} \geq 4$ and hence
    \begin{align*}
        \max_{\abs{\mathbf m} \geq 4} \bra{\mathbf m, \mathbf 0} P \ket{\mathbf m, \mathbf 0} = \frac{1}{8}(2 \times 1 + 2 \times 0 + 4 \times \frac14) = \frac{3}{8} < \frac12
    \end{align*}
    concluding the proof of (\ref{it:close to gaussian mean zero}).

    The proof of (\ref{it:close to gaussian mean nonzero}) follows the exact same strategy. We expand $\ket{\psi}$ as in \cref{eq:gaussian + pert} in a component which is the closest Gaussian state, which we may take to be $\ket{\phi} = \ket{0}$ without loss of generality, and an orthogonal state $\ket{\alpha}$.
    We let $P$ denote the accepting measurement and expand
    \bb
    p_{\accept}(\psi) = \Tr[P \psi^{\ot 3}] &= (1 - \eps^2)^3 + 9\eps^2(1 - \eps^2)^2 \Tr[P \alpha \ot \phi^{\ot 2}] + \bigO(\eps^3)\\
    &= 1 - 3\eps^2 + 9\eps^2\Tr[P \alpha \ot \phi^{\ot 2}] + \bigO(\eps^3)
    \ee
    using that $P \ket{\phi}^{\ot 3} = \ket{\phi}^{\ot 3}$ (so the term linear in $\eps$ vanishes), and $P$ projects onto a space which is invariant under cyclic permutations. This means it suffices to show that $\Tr[P \alpha \ot \phi^{\ot 2}] \leq c$ for $c < \frac13$.
    We expand $\ket{\alpha}$ as in \cref{eq:expand alpha} (but now not necessarily in the even paritty subspace).
    In this case, \cref{eq:critical point} is also required for \emph{linear} Hamiltonians, which implies that $\alpha_{\mathbf m} = 0$ if $\abs{\mathbf m} \leq 2$.
    We now consider $\hat m_i = a_{i,1}^\dagger a_{i,1} + a_{i,2}^\dagger a_{i,2} + a_{i,3}^\dagger a_{i,3}$.
    Then $[\hat m_i, U_{\theta}] = 0$ and hence $[\hat m_i, P] = 0$, which implies that
    \begin{align*}
        \bra{\mathbf m, \mathbf 0, \mathbf 0} P \ket{\mathbf n, \mathbf 0, \mathbf 0} = 0
    \end{align*}
    if $\mathbf m \neq \mathbf n$ and hence
    \begin{align*}
        \bra{\alpha} \bra{\phi}^{\ot 2} P \ket{\alpha} \ket{\phi}^{\ot 2} & = \sum_{\mathbf m} \abs{\alpha_{\mathbf m}}^2 \bra{\mathbf m, \mathbf 0, \mathbf 0} P \ket{\mathbf m, \mathbf 0, \mathbf 0} \\
                                                                          & \leq \max_{\abs{\mathbf m} \geq 3} \bra{\mathbf m, \mathbf 0, \mathbf 0} P \ket{\mathbf m, \mathbf 0, \mathbf 0}.
    \end{align*}
    We write
    \begin{align*}
        P = \frac{1}{6}\sum_{l=0}^5 U_{\pi/3}^l = \frac{1}{6}\sum_{l=0}^5 U_{l\pi/3}.
    \end{align*}
    For $l = 2, 4$ the unitary $U_{\pi/3}^l$ is a nontrivial cyclic permutation, and hence
    \begin{align*}
        \bra{\mathbf m, \mathbf 0, \mathbf 0} U_{\pi/3}^l \ket{\mathbf m, \mathbf 0, \mathbf 0} = 0
    \end{align*}
    for $\abs{\mathbf m} \neq 0$.
    For $l = 0$ we obviously have $\bra{\mathbf m, \mathbf 0, \mathbf 0} U_{\pi/3}^l \ket{\mathbf m, \mathbf 0, \mathbf 0} = 1$.
    The case of odd $l$ remains. We take $l = 1$ (the other ones are similar), and the situation reduces as before to the single-mode case, with annihilation operators $a_1, a_2, a_3$. We have
    \begin{align*}
        \bra{m,0,0} R_{\pi/3} \ket{m,0}
        = \bra{m,0,0} \frac{(R_{\pi/4} a_1^\dagger R_{\pi/3}^\dagger)^m}{\sqrt{m!}} \ket{0,0,0}
    \end{align*}
    using that $R_{\pi/3} \ket{0,0} = \ket{0,0}$.
    Now $R_{\pi/3} a_1^\dagger R_{\pi/3}^\dagger = \frac{2}{3} a_1^\dagger + \frac{2}{3}a_2^\dagger - \frac{1}{3}$ and hence
    \begin{align*}
        \bra{m,0,0} R_{\pi/3} \ket{m,0,0} = \bra{m,0,0}\left(\frac{2^m}{3^m} \frac{(a_1^\dagger)^m}{\sqrt{m!}} + \ldots\right) \ket{0,0,0} = \left(\frac{2}{3} \right)^m.
    \end{align*}
    (omitting terms that evaluate to zero).
    We conclude that
    \begin{align*}
        \bra{\mathbf m, \mathbf 0, \mathbf 0} U_{\pi/3} \ket{\mathbf m, \mathbf 0, \mathbf 0} = \left(\frac23 \right)^m \leq \frac{8}{27}
    \end{align*}
    since $\abs{\mathbf m} \geq 3$.
    A similar calculation, and using that $U_{\pi/3}^2$ is a cyclic permutation, shows that
    \begin{align*}
        \bra{\mathbf m, \mathbf 0, \mathbf 0} U_{\pi/3}^3 \ket{\mathbf m, \mathbf 0, \mathbf 0} & = \left(\frac23 \right)^m \leq \frac{8}{27} \\
        \bra{\mathbf m, \mathbf 0, \mathbf 0} U_{\pi/3}^5 \ket{\mathbf m, \mathbf 0, \mathbf 0} & = \left(\frac13 \right)^m \leq \frac{1}{27}
    \end{align*}
    and hence
    \begin{align*}
        \max_{\abs{\mathbf m} \geq 3} \bra{\mathbf m, \mathbf 0, \mathbf 0} P \ket{\mathbf m, \mathbf 0, \mathbf 0} \leq \frac{1}{6}(1 + 2 \times 0 + 2 \times \frac{8}{27} + 1 \times \frac{1}{27}) = \frac{22}{81} < \frac13.
    \end{align*}
    This proves (\ref{it:close to gaussian mean nonzero}).
\end{proof}

We can now state and prove the complete robustness result for Test 1 and Test 4, proving the result in \cref{thm:robust test} in the main text.
The idea is that it suffices to make sure that $p_{\accept}(\psi)$ distinguish between the two hypotheses, and then we can repeat the test and apply \cref{lem:amplify error test} to get a test with small probability of error.

\begin{boxed}{}
    \begin{thm}\label{thm:appendix theorem rotation testing}
        Consider the task of testing Gaussianity as in \cref{sec:intro_to_property_testing} where $\mathcal P = \gaussian$ is the set of pure Gaussian states on $n$ modes, and with parameters $0 \leq \eps_A < \eps_B$ and $\delta \in (0,1)$, on a pure state $\psi$.
        Suppose that under hypothesis $H_B$ (if $\psi$ is far from Gaussian), the state satisfies $E_2(\psi) \leq E$.
        There exists a constant $C$, such that if
        \bb\label{eq:define eps for rotation test}
        \eps = C \min\left(\eps_B^2, \frac{1}{n^4 E^4} \right) - \eps_A
        \ee
        is positive, then
        \bb
        N = \bigO(\eps^{-2} \log(\delta^{-1}))
        \ee
        rounds of Test 4 (each using 3 copies of $\psi$), suffice to determine whether the state is $\eps_A$-close to, or $\eps_B$-far from, $\gaussian$ with probability of error at most $\delta$.
        If $\eps_A = \bigO(\eps)$,
        \bb
        N = \bigO(\eps^{-1} \log(\delta^{-1}))
        \ee
        rounds suffice.

        The same statement is true for Test 1, if we assume that $\psi$ has even parity, and we consider $\mathcal P = \gaussian_0$, the set of zero-mean Gaussian states.
    \end{thm}
\end{boxed}

\begin{proof}
    We need to bound $p_{\reject} = 1 - p_{\accept}$.
    If $d_{\gaussian}(\psi) \leq \eps_A$, then since the test has perfect completeness,
    \bb\label{eq:small p reject}
    p_{\reject} \leq 3\eps_A.
    \ee
    If, on the other hand, $d_{\gaussian} \geq \eps_B$, we distinguish two cases.
    If $d_{\gaussian}(\psi) = \eta \leq \eps_0$ as in \cref{thm:close to gaussian} (which is of course only relevant if $\eps_B < \eps_0$), then by \cref{thm:close to gaussian}
    \bb
    p_{\reject} = \Omega(\eta^2).
    \ee
    For the situation where $d_{\gaussian}(\psi) > \eps_0$, we can use that by \cref{cor:accept implies close to gaussian}, we get
    \bb
    p_{\reject} = \Omega((nE)^{-4}).
    \ee
    In conclusion, for pure states, there is a constant $C$ such that if $d_{\gaussian} \geq \eps_B$,
    \bb\label{eq:large p reject}
    p_{\reject} \geq C \min \left(\eps_B^{-2},(nE)^{-4} \right).
    \ee
    The result now follows from \cref{lem:amplify error test}.

    The same argument applies to Test 1 (replacing $d_{\gaussian}$ by $d_{\gaussian_0}$).
\end{proof}

We conjecture the following. This would directly imply that the Gaussianity tests we propose require a number of copies that does not scale with the number of modes $n$ at all (i.e.~in \cref{eq:define eps for rotation test} one could take $\eps = \Omega(\eps_B^2 - \eps_A)$).

\begin{cj}\label{con:two copy test}$\phantom{.}$
    \begin{enumerate}
        \item Let $\psi$ be a pure even-parity state. If Test 1 or 2 accepts with high probability on $\psi$
              \begin{align*}
                  p_{\accept} = \Tr[\psi^{\ot 2} P] \geq 1 - \eps,
              \end{align*}
              there exists a zero mean Gaussian state $\ket{\phi}$ such that $\abs{\angles{\psi|\phi}} \geq 1 - f(\eps)$, where $f$ does not depend on $n$ and $f(x) = o(1)$.
        \item Let $\psi$ be a pure state. If Test 4 or 5 accepts with high probability on $\psi$
              \begin{align*}
                  p_{\accept} = \Tr[\psi^{\ot 3} P] \geq 1 - \eps,
              \end{align*}
              there exists a Gaussian state $\ket{\phi}$ such that $\abs{\angles{\psi|\phi}} \geq 1 - f(\eps)$, where $f$ does not depend on $n$ and $f(x) = o(1)$.
    \end{enumerate}
\end{cj}

Additionally, we comment that for Test 2, Test 3, and Test 5, we have only shown limited robustness properties. These tests have perfect completeness, so if the state is close to Gaussian, they accept with high probability. We have only shown the converse direction in the case where the state is constantly close to Gaussian, for Test 2 and Test 5, in \cref{thm:close to gaussian}.
Since these tests are of a more practical nature than Test 1 and Test 4 (they have a relatively straightforward circuit), it would be interesting to understand their properties better.
For Test 3, there is a robustness result proven in \cite{Cuesta2020}.
There, it is shown that if $\norm{U_{\pi/4} \rho^{\ot 2} U_{\pi/4}^\dagger - \sigma_1 \ot \sigma_2}_1 \leq \eps$, there exists a (mixed) Gaussian state $\omega$ such that
\bb
\norm{\rho - \omega}_2 \leq c_1 \eps^{\frac13} + \frac{c_2}{\sqrt{\log 1/\eps}}
\ee
where $c_1$ and $c_2$ depend polynomially on the energy and number of modes (\cite{Cuesta2020} gives explicit expressions, where it should be noted that they contain a spurious factor $\pi^{-n}$ due to a normalization error in the normalization conventions).
If $\rho$ is pure, $\norm{\rho - \omega}_1 \leq 2\sqrt{\norm{\rho - \omega}_2}$, and this gives a robustness result which can be applied to analyze the test.
However, the term which depends logarithmically on $1/\eps$ means that if one uses this result in similar fashion to the proof of \cref{thm:appendix theorem rotation testing}, the resulting bounds depend \emph{exponentially} on the number of modes, which is not satisfying.
Closely related bounds have been derived through the \emph{quantum central limit theorem} \cite{Cushen1971,becker2021convergence,beigi2023towards}, which may be helpful for an improved analysis.

\medskip
Below we state some useful facts for completeness.
First we show that Test 2', the second circuit in \cref{fig:rotation test discrete}, is closely related to Test 2, at a small cost in the performance.

\begin{lemma}\label{lem:accept pi over 4 circuit}
    Suppose that Test 2' accepts with probability at least $1 - \eps$ for a (not necessarily pure) state $\rho$.
    Then Test 2 accepts $\rho$ with probability at least
    \begin{align*}
        \Tr[P \rho^{\ot 2}] \geq 1 - \bigO(\eps).
    \end{align*}
\end{lemma}

\begin{proof}
    The acceptance probability of Test 2' is given by
    \begin{align*}
        \Tr\left[\frac12(\Id + U_{\pi/4}) \rho^{\ot 2} \frac12(\Id + U_{\pi/4})^\dagger\right].
    \end{align*}
    Since $U_{\pi/4}^8 = \Id$ and it is unitary, $U_{\pi/4}$ has spectrum contained in $\{e^{2\pi ik/8}\}_{k=0}^7$.
    This means that $Q = \frac14(\Id + U_{\pi/4})(\Id + U_{\pi/4}^\dagger)$ has eigenvalues $\sigma_k = \frac14\abs{1 + e^{2\pi i k/8}}^2 = \frac12(1 + \cos(\pi k / 4))$ for $k = 0,\dots,4$.
    In particular, $k = 0$ gives eigenvalue $\sigma_k = 1$. For $k \neq 0$ we have $0 \leq \sigma_k \leq \frac12(1 + \cos(\pi / 4)) = \alpha < 1$.
    Since $P$ is the projector onto the $+1$ eigenspace of $Q$, we have
    \begin{align*}
        1 - \eps \leq \Tr[Q \rho^{\ot 2}] \leq \Tr[P \rho^{\ot 2}] + \alpha \Tr[(\Id - P) \rho^{\ot 2}] = (1 - \alpha) \Tr[P \rho^{\ot 2}] + \alpha
    \end{align*}
    which implies
    \begin{align*}
        \Tr[P \rho^{\ot 2}] \geq 1 - \frac{\eps}{1 - \alpha}.
    \end{align*}
\end{proof}

The tests are tests for \emph{pure} Gaussian states, and we analyze them for pure states.
This is not a significant restriction.
First of all, testing purity can be done efficiently, but additionally, the Gaussianity tests also directly detect if the state is not pure. The reason is that (up to parity) rotation over an angle $\pi/2$ is a swap, so testing rotation invariance (either over any angle, or over $\pi/4$) is stronger than the swap test.

\begin{lemma}\label{lem:purity}
    Let $\rho$ be an arbitrary state, and suppose that Test 2 has an acceptance probability of at least $1 - \eps$.
    Then $\rho$ has purity at least $\Tr[\rho^2] = 1 - \bigO(\sqrt{\eps})$.
\end{lemma}

\begin{proof}
    By the gentle measurement lemma, if $\sigma$ is the post-measurement state $P\rho^{\ot 2}P / \Tr[P \rho^{\ot 2}]$ we have
    \begin{align*}
        \norm{\rho^{\ot 2} - \sigma}_1 \leq \bigO(\sqrt{\eps}).
    \end{align*}
    In particular, since $U_{\pi/4} \sigma = \sigma$,
    \begin{align*}
        \Tr[U_{\pi/4} \rho^{\ot 2}] \geq \Tr[U_{\pi/4} \sigma] - \norm{\rho^{\ot 2} - \sigma}_1 \geq 1 - \bigO(\sqrt{\eps}).
    \end{align*}
    and similarly
    \begin{align*}
        \Tr[U_{\pi/4}^2 \rho^{\ot 2}] \geq 1 - \bigO(\sqrt{\eps})
    \end{align*}
    The operator $U_{\pi/4}^2$ implements a rotation by angle $\pi/2$, which is equivalent to a swap of the copies and a reflection on the second copy.
    In particular,
    \begin{align*}
        \Tr[U_{\pi/2} \rho^{\ot 2}] = \Tr[\rho \tilde\rho]
    \end{align*}
    where $\tilde \rho = J \rho J$ and $J f(x) = f(-x)$ for $f \in L^2(\RR^n)$.
    In particular, this implies that $\rho$ has large purity, as
    \begin{align*}
        \Tr[\rho \tilde\rho] \leq \norm{\rho}_{\infty} \norm{\tilde\rho}_1 = \norm{\rho}_{\infty}.
    \end{align*}
    Additionally,
    \begin{align*}
        F(\rho, \tilde\rho) = \norm{\sqrt{\rho}\sqrt{\tilde\rho}}_1^2 \geq \norm{\sqrt{\rho}\sqrt{\tilde\rho}}_2^2 = \Tr[\rho \tilde\rho] \geq 1 - \bigO(\sqrt{\eps}).
    \end{align*}
\end{proof}

This also directly implies the same result for Test 1 (which never has a higher acceptance probability than Test 2), and a similar argument also applies to Test 4 and Test 5, where squaring the rotation over $\pi/3$ gives a cyclic permutation.

\subsection{A brief excursion on representation theory}\label{sec:relation to rep theory}
The mathematical structure of the relation between \emph{rotation invariance} and \emph{Gaussian quantum states}, has been known since the 1970s in representation theory.
Here we briefly explain the general framework, as we believe it will be of interest for further refinement of Gaussianity testing and other quantum information processing tasks for Gaussian states.
See \cite{kashiwara1978segal} and references therein for a complete exposition.
The mapping $S \mapsto U_S$ is a projective representation of the symplectic group~$\Sp(2n)$ on $\H_n = L^2(\RR^n)$.
The zero-mean states are orbits of the vacuum state under this action:
\begin{align*}
    \gaussian = \Sp(2n) \cdot \ket0\bra0 \subseteq \PP(\H_n)
\end{align*}
Such a projective representation lifts to a linear representation of the metaplectic group~$\Mp(2n)$, its double cover.
As discussed, for $k$ copies we have
\begin{align*}
    \H_n^{\ot k} = L^2(\RR^n)^{\ot k} \cong L^2(\RR^{n \times k}),
\end{align*}
where $\RR^{n \times k}$ are $n$ by $k$ real matrices and the tensor power action of $\Sp(2n,\RR)$ commutes with the action of $O(k)$ (\cref{lem:commuting two copies}).

In fact, the two actions generate each others' commutants, which implies that we have a Schur-Weyl type decomposition $L^2(\RR^{n \times k}) \cong \bigoplus_{\lambda \in \Sigma_{n,k}} W_{\metapl(2n),\tilde\lambda} \ot W_{O(k),\lambda}$, where we write $W_{G,\mu}$ for the irreducible representation of $G$ with highest weight $\mu$, and $\lambda \mapsto \tilde\lambda$ is some injective map.
This decomposition was worked out explicitly in \cite{kashiwara1978segal}:

\begin{thm}[{\cite[Theorem~(7.2)]{kashiwara1978segal}}]\label{thm:KV}
    As a representation of $\metapl(2n) \times O(k)$, one has the multiplicity-free decomposition
    \begin{align*}
        L^2(\RR^{n \times k}) \cong \bigoplus_{\lambda \in \Sigma_{n,k}} W_{\metapl(2n),\tilde\lambda} \ot W_{O(k),\lambda}
    \end{align*}
    for some explicit index set $\Sigma_{n,k}$ and some explicit injection $\lambda \mapsto \tilde\lambda$ described in their paper.
\end{thm}

For our approach to testing zero-mean states we use $k=2$ copies, where
\begin{align*}
    \Sigma_{n,2} = \begin{cases}
                       \NN \cup \{-\} & \text{if $n\geq2$}, \\
                       \NN            & \text{if $n=1$},
                   \end{cases}
\end{align*}
and the irreducible representations of $O(2)$ can be described as follows:
for $\lambda \in \NN_{>0}$, $W_{O(2),\lambda}$ is a two-dimensional representation that when restricted to $\SO(2)$ consists of weights $\lambda$ and $-\lambda$;
$W_{O(2),0}$ is the one-dimensional trivial representation;
$W_{O(2),-}$ is the one-dimensional determinant representation.
What is special about the case $n=1$ is that in this case a wave function in $L^2(\RR^{1 \times 2}) = L^2(\RR)^{\ot 2}$ that is rotation invariant is automatically in the symmetric subspace, because it only depends on the norm of the argument.
For $n>1$, we can have rotation-invariant wave functions in $L^2(\RR^{n \times 2}) = L^2(\RR^n)^{\ot 2}$ that are odd.
In particular, we have the following result:

\begin{cor}\label{thm:span gaussian commutator}
    For every $n\geq1$,
    \begin{align*}
        (\mathcal H_n^{\ot 2})^{O(2)}
        = (\Sym^2(\mathcal H_n))^{\SO(2)}
    \end{align*}
    is an irreducible representation of $\Mp(2n)$ and hence the above is also equal to
    \begin{align*}
        \Span \parens*{ \ket\psi^{\ot 2} \;\middle|\; \psi \in \gaussian_0 }.
    \end{align*}
\end{cor}
\begin{proof}
    For the first equality, note that invariance under $\tau = \begin{psmallmatrix}0 & 1 \\ 1 & 0\end{psmallmatrix} \in S_2 \subseteq O(2)$ means that the wave function is in the symmetric subspace, and we have $O(2) / \SO(2) \cong \{I, \tau\}$.
    And this space is irreducible since as an $\Mp(2n)$-representation it is isomorphic to the irreducible representation $W_{\Mp(2n),\tilde 0}$ in the notation of \cref{thm:KV}.
    For the second equality, one notes that (because $\gaussian$ is closed under $\Mp(2n)$) the span of $\ket\psi^{\ot 2}$ for $\psi \in \gaussian$ is a (nonzero) invariant subspace of this irreducible representation, hence the two must be equal.
\end{proof}

\begin{rem}
    More generally, for any $n\geq1$, it holds that
    \begin{align*}
        \Sym^2(\mathcal H_n) = L^2(\mathcal H_n)^\tau \cong \bigoplus_{\lambda \in \NN} W_{\metapl(2n),\tilde\lambda}
    \end{align*}
    as a $\Mp(2n)$-representation, where $\tau = \begin{psmallmatrix}0 & 1 \\ 1 & 0\end{psmallmatrix} \in S_2 \subseteq O(2)$.
    Note that despite the suggestive notation used in \cref{thm:span gaussian commutator}, the action of $\SO(2)$ does not restrict to the symmetric subspace, because $O(2) = \SO(2) \rtimes \langle \tau \rangle$ is only a semidirect product or, relatedly, the subspaces $W_{O(2),\lambda}^\tau$ are not $\SO(2)$-representations when $\lambda > 0$.
    Still these subspaces are one-dimensional and so we get the claim above.
\end{rem}

We can also consider the action of the displacement operators $D_{\mathbf r}$, which form a representation of the Heisenberg group $H_n$. This gives the semidirect product $G_n = \Mp(2n) \ltimes H_n$ of affine transformations.
It is now easy to see that elements in $O(k)$ commute with the action of $G_n$ if and only if they are stochastic.

The decomposition in \cref{thm:KV} is an instance of a more abstract duality studied by Howe \cite{howe1995perspectives}.
In our concrete setting and language, the basic idea behind the bosonic setting is that $\Sp(2n)$ and $O(k)$ generate each others centralizer in $\Sp(2nk)$, and in nice situations (e.g.\ when the field is $\RR$) this implies that the action of these subgroups in the oscillator representation generate each others' commutants.
Howe was interested in this for general fields and many results are known.
Note that for finite-dimensional quantum systems, formally the very same situation arises, but with the finite field~$\FF_p$ instead of the reals~$\RR$.
Here the situation is less nice, but we know the answer from \cite{nebe2006self,gross2021schur,bittel2025complete}, namely that the commutant is slightly larger than $O(k)$ and can be described by certain Lagrangian subspaces in the discrete phase space.

Finally we note that the fermionic Gaussian case is just as classical as the bosonic Gaussian case.
Here, the relevant groups are $\Spin(2n)$ (with Lie algebra $\mathfrak{so}(k)$) and the relevant version of Howe duality states the following:

\begin{thm}[{\cite[4.3.4, 4.3.5]{howe1995perspectives}}]
    As a representation of $\Spin(2n) \times O(k)$, one has the multiplicity-free decomposition
    \begin{align*}
        \bigwedge(\CC^n \ot \CC^k) \cong \bigoplus_{\lambda \in \Sigma_{n,k}} W_{\Spin(2n),\tilde\lambda} \ot W_{O(k),\lambda}
    \end{align*}
    for some explicit index set $\Sigma_{n,k}$ and some explicit injection $\lambda \mapsto \tilde\lambda$ described in \cite{howe1995perspectives}.
\end{thm}

Note that again $\Spin(2n) \times O(k) \subseteq \Spin(2nk)$, but because of the anti-commuting nature of the generators one has to be slightly careful in translating the action of $O(k)$ from $\bigwedge(\CC^n \ot \CC^k)$ to the $k$-copy Hilbert space $(\bigwedge \CC^n)^{\ot k}$.

\section{Testing by learning}\label{sec:pure}
In this section we analyze the task of testing whether an unknown quantum state is close to, or far from, a pure Gaussian state from a different perspective.
Unlike the approach of the previous section---which was based on exploiting \emph{symmetry} properties of the state---the strategy here relies on \emph{learning} the statistical moments of the unknown state. Intuitively, while symmetry tests verify invariance properties under certain group actions, learning-based tests proceed by estimating the covariance matrix and the first moments, and then checking whether these are consistent with those of a pure Gaussian state.
A difference with the symmetry-based tests is that the learning approach only requires single-copy measurements. It may therefore be more convenient for experimental implementation, especially when it is not possible to produce multiple copies of the same state simulataneously and entangle them.

\medskip
Throughout this section, we denote by $\mathcal{G}_E$ the set of all $n$-mode pure Gaussian states with mean energy per mode at most $E$.
Formally, a state $\rho$ has mean energy bounded by $E$ whenever
\begin{align}
    \Tr[\rho \hat{E}] \leq n E\,,
\end{align}
where $\hat{E}$ denotes the energy operator.
Imposing such an energy constraint is natural both from a physical and an analytical perspective: it reflects the finite-energy resources available in realistic experimental settings, and at the same time guarantees that all relevant quantities (first moments, covariance matrix, etc.) are finite.
In particular, the bound on the energy ensures that the covariance matrix of $\rho$ can be learned with a finite number of samples, where the required sample complexity depends explicitly on~$E$.

\medskip

The testing problem we are going to study is the instance of \cref{prob:testing} (see also \cref{sec:intro_to_property_testing}) where $\mathcal{P}=\mathcal{G}_E$.

\subsection{Analysis of tolerant pure Gaussian property testing by learning the moments}
The key idea of the learning-based approach is that the covariance matrix of a continuous-variable quantum state
contains enough information to determine its proximity to the set of pure Gaussian states.
In fact, a state $\rho$ is a pure Gaussian state if and only if all the symplectic eigenvalues of its covariance matrix are equal to one.
Hence, deviations of the symplectic spectrum from unity can be directly related to the trace distance between $\rho$
and the set $\mathcal{G}_E$.
Our strategy for tolerant property testing is therefore straightforward:
estimate the covariance matrix of the unknown state with sufficient accuracy,
use this estimate to approximate the symplectic eigenvalues,
and decide whether $\rho$ is close to or far from the Gaussian set according to the size of these deviations.
This reduces the quantum property testing problem to a statistical task of learning the moments of the state.

To formalize this approach, two ingredients are required.
First, we establish upper and lower bounds on the minimum trace distance between an arbitrary state and the set of pure Gaussian states,
expressed explicitly in terms of its symplectic eigenvalues (see Lemma~\ref{thm:upper_bound} for the upper bound
and Lemma~\ref{lemma_lower_bound_min} for the lower bound).
Second, we derive robust (perturbed) versions of these bounds that remain valid when only an approximate covariance matrix is available
(see Lemma~\ref{lem:upper_perturbed} and Lemma~\ref{lemma_low_approx}).
Together, these results show how accurate covariance estimation yields rigorous guarantees on the distance from pure Gaussian states,
thus enabling a tolerant property testing strategy.
Theorem~\ref{thm:pure_testing}, stated at the end of this subsection, summarizes this learning-based analysis and provides
the main result of this section on the sample complexity of pure Gaussian property testing.

We begin our analysis by establishing a first quantitative relation between the trace distance from the set of pure Gaussian states and the covariance matrix. In particular, the following lemma shows such distance can be upper bounded in terms of the deviations of the symplectic eigenvalues from unity.
\begin{lemma}[(Upper bound on the minimum trace distance from the set of pure Gaussian states)]\label{thm:upper_bound}
    Let $\rho$ be an $n$-mode state. Let $\nu_1,\nu_2,\ldots,\nu_n$ be the symplectic eigenvalues of its covariance matrix $V(\rho)$. Then, it holds that
    \begin{align}
        \min_{\psi\in\mathcal{G}}\frac 12\|\rho-\psi\|_1\leq \frac{1}{\sqrt 2}\sqrt{\sum_{i=1}^n(\nu_i-1})\,,
    \end{align}
    where the minimum is taken over the set $\mathcal{G}$ of all $n$-mode pure Gaussian states.
\end{lemma}
\begin{proof} Let $V(\rho)=SDS^{\intercal}$ be the Williamson decomposition of the covariance matrix of $\rho$. Let $\bar{\psi}$ be the pure Gaussian state with covariance matrix $V(\bar{\psi})\coloneqq SS^\intercal$ and first moment $m(\bar\psi)\coloneqq m(\rho)$. That is, $\bar \psi$ can be expressed as
    \bb
    \bar\psi = D_{m(\rho)}U_S\ketbra{0}U_S^\dagger D_{m(\rho)}^\dagger\,,
    \ee
    where $\ket{0}$ denotes the $n$-mode vacuum state. Let us upper bound the trace norm between the Gaussian state $\bar\psi$ and $\rho$. To this end, let us proceed in a similar way as in the proof of \cref{lem:rotation invariance G} above. We have that
    \begin{align*}
        \frac{1}{2}\| \rho-\bar\psi  \|_1 & \leqt{(i)}\sqrt{1-\Tr[\rho\bar\psi]}                                                                                                                                                                                              \\
                                          & =\sqrt{1-\Tr[U_{S}^\dagger D_{m(\rho)}^\dagger \rho D_{m(\rho)} U_{S} \,\ketbra{0}]}                                                                                                                                              \\
                                          & \leqt{(ii)} \sqrt{\Tr[U_{S}^\dagger D_{m(\rho)}^\dagger \rho D_{m(\rho)} U_{S} \,\hat{N}]}                                                                                                                                        \\
                                          & =\sqrt{\frac{1}{4}\Tr\!\left[ V\!\left( U_{S}^\dagger D_{m(\rho)}^\dagger \rho D_{m(\rho)} U_{S}\right)-\mathbb{1}\right]+\frac{1}{2}\left\| m\!\left(U_{S}^\dagger D_{m(\rho)}^\dagger \rho D_{m(\rho)} U_{S} \right)\right\|^2} \\
                                          & \eqt{(iii)}\frac{1}{2}\sqrt{\Tr\!\left[ D-\mathbb{1}\right]}                                                                                                                                                                      \\
                                          & =\frac{1}{\sqrt{2}}\sqrt{\sum_{i=1}^n(\nu_i-1)}\,.
    \end{align*}
    Here, in (i) we used Fuchs-van de Graaf inequality. In (ii), we used the operator inequality $\ketbra{0}\ge \mathbb{1}-\hat{N}$, which can be easily proved by expanding in Fock basis, where $\hat{N}$ denotes the total photon number operator. In (iii), we exploited that
    \bb
    V\!\left( U_{S}^\dagger D_{m(\rho)}^\dagger \rho D_{m(\rho)} U_{S}\right)&= S^{-1}V(D_{m(\rho)}^\dagger \rho D_{m(\rho)} )S^{-\intercal}=S^{-1}V(\rho)S^{-\intercal}=D\,,\\
    m\!\left(U_{S}^\dagger D_{m(\rho)}^\dagger \rho D_{m(\rho)} U_{S} \right)&=S^{-1}m\left(D_{m(\rho)}^\dagger \rho D_{m(\rho)}\right)=S^{-1}\left(m(\rho)-m(\rho)\right)=0\,.
    \ee
\end{proof}
In order to deal with estimates of the covariance matrix, we also require a perturbation bound on the symplectic diagonalisation.
Lemma~\ref{lem:wolf}, adapted from~\cite{Idel2016}, provides such a tool by controlling how the symplectic spectrum changes under perturbations of the covariance matrix.
\begin{lemma}[(Perturbation on symplectic diagonalisation~\cite{Idel2016})]\label{lem:wolf}
    Let $V_1,V_2\in\mathbb{R}^{2n\times 2n}$ be two covariance matrices with symplectic diagonalisations $V_1=S_{1}D_1S_{1}^{\intercal}$ and $V_2=S_{2}D_2S_{2}^{\intercal}$, where the elements on the diagonal of $D_1$ and $D_2$ are arranged in descending order. Then
    \bb
    \|D_1-D_2\|_\infty&\le \sqrt{K\!(V_1)K\!(V_2)}\|V_1-V_2\|_\infty\,,\\
    \|D_1-D_2\|_2&\le \sqrt{K\!(V_1)K\!(V_2)}\|V_1-V_2\|_2\,,
    \ee
    where $K\!(V)$ is the condition number of the covariance matrix $V$, defined as $K\!(V)\coloneqq \|V\|_\infty \|V^{-1}\|_\infty$.
\end{lemma}
Equipped with this perturbation bound, we can now extend Lemma~\ref{thm:upper_bound} to the case
where only an estimate of the covariance matrix is available.
\begin{lemma}[(Perturbed upper bound)]\label{lem:upper_perturbed}
    Let $\rho$ be an $n$-mode state with mean energy upper bounded by $n E$. Let $\tilde V$ be a covariance matrix being $\varepsilon$-close to $V(\rho)$ in 2-norm:
    \begin{align}\label{eq:hypothesis_V}
        \|V(\rho)-\tilde V\|_2\leq \varepsilon_V\,.
    \end{align}
    Moreover, let $\tilde{\nu}_1,\ldots,\tilde{\nu}_n$ be the symplectic eigenvalues of $\tilde{V}$. Then, it holds that
    \begin{align}
        \min_{\psi\in\mathcal{G}}\frac 12\|\rho-\psi\|_1\leq \frac{1}{\sqrt 2}\sqrt{\sum_{i=1}^n(\tilde\nu_i-1)+\sqrt{n^3}4E(4n E+\varepsilon_V)\varepsilon_V}  \,,
    \end{align}
    where the minimum is taken over the set $\mathcal{G}$ of all $n$-mode pure Gaussian states and where we denoted as $nE\coloneqq \Tr[\hat{E}\rho]$ the mean energy of $\rho$. In particular, it holds that
    \begin{align}
        \min_{\psi\in\mathcal{G}}\frac 12 \|\rho-\psi\|_1\leq \frac{1}{\sqrt 2}\sqrt{n(\max_i\tilde\nu_i-1)+\sqrt{n^3}4E(4nE+\varepsilon_V)\varepsilon_V}  \,,
    \end{align}
\end{lemma}
\begin{proof}
    The crux of the proof is to exploit Lemma~\ref{thm:upper_bound} and Lemma~\ref{lem:wolf}. To this end, we need to upper bound the quantities $\| V^{-1}\|_\infty$ and $\|\tilde V^{-1}\|_\infty$. Note that for any covariance matrix $V$ it holds that $V^{-1}\le \Omega V \Omega^\intercal$. Indeed, by letting $V=SDS^\intercal$ be the Williamson decomposition of $V$, it holds that
    \bb
    V^{-1} & =S^{-\intercal}D^{-1}S^{-1}                        \\
    & \leq S^{-\intercal}DS^{-1}                         \\
    & \leq S^{-\intercal}\Omega D\Omega^\intercal S^{-1} \\
    & =\Omega SDS^\intercal\Omega^\intercal              \\
    & = \Omega V\Omega^\intercal.
    \ee
    In particular, any covariance matrix $V$ satisfies
    \begin{align}
        \|V^{-1}\|_\infty \leq \|V\|_\infty\,.
    \end{align}
    Consequently, note that the condition number of any covariance matrix $V$ can be upper bounded as
    \bb\label{cond_numb}
    K(V)\le \|V\|_\infty^2\,.
    \ee
    In particular, it holds that
    \bb
    \sqrt{K(V(\rho))K(\tilde{V})}&\le \|V(\rho)\|_\infty \|\tilde V\|_\infty\leqt{(i)} \|V(\rho)\|_\infty\left(\|V(\rho)\|_\infty+\varepsilon_V\right)
    \ee
    where (i) follows from \eqref{eq:hypothesis_V} and from $\|\,\cdot\,\|_\infty\leq \|\,\cdot\,\|_2$.
    Moreover, since it holds that
    \bb\label{cond_numb2}
    \|V(\rho)\|_\infty\le \Tr V(\rho)\le 4\Tr[\hat{E}\rho]\le 4nE\,,
    \ee
    we obtain that
    \bb\label{eq_condxxx}
    \sqrt{K(V(\rho))K(\tilde{V})}&\le 4nE(4nE+\varepsilon_V)\,.
    \ee
    Without loss of generality, let us assume that $\tilde\nu_1,\tilde\nu_2,\ldots,\tilde\nu_n$ are sorted in descending order. Moreover, let $\nu_1,\nu_2,\ldots,\nu_n$ be the symplectic eigenvalues of $V(\rho)$ sorted in decreasing order.
    By exploiting Lemma~\ref{lem:wolf}, we can obtain that
    \bb
    \sqrt{2\sum_{i=1}^n(\nu_i-\tilde\nu_i)^2}\le 4nE(4nE+\varepsilon_V)\varepsilon_V\,.
    \ee
    Consequently, exploiting Lemma~\ref{thm:upper_bound}, we conclude that
    \bb
    \min_{\psi\in\mathcal{G}_{\text{pure}}^{(n)}}\frac{1}{2}\|\rho-\psi\|_1&\leq \frac{1}{\sqrt 2} \sqrt{\sum_{i=1}^n(\nu_i-1})\\
    &= \frac{1}{\sqrt 2}\sqrt{\sum_{i=1}^n(\tilde\nu_i-1)+\sum_{i=1}^n(\nu_i-\tilde\nu_i)}\\
    &\leqt{(i)} \frac{1}{\sqrt 2}\sqrt{\sum_{i=1}^n(\tilde\nu_i-1)+\sqrt{n\sum_{i=1}^n(\nu_i-\tilde\nu_i)^2}}\\
    &\leq \frac{1}{\sqrt 2}\sqrt{\sum_{i=1}^n(\tilde\nu_i-1)+\sqrt{\frac{n^3}{2}}4E(4nE+\varepsilon_V)\varepsilon_V   }\\
    &\leq \frac{1}{\sqrt 2}\sqrt{\sum_{i=1}^n(\tilde\nu_i-1)+\sqrt{n^3}4E(4nE+\varepsilon_V)\varepsilon_V   }\,,
    \ee
    where in (i) we have used that the arithmetic mean is smaller than the quadratic mean.
\end{proof}
So far we have derived an upper bound, ensuring that small deviations in the symplectic spectrum from unity imply closeness to the Gaussian set. To complement this, the following lemma establishes a lower bound: it shows that if one symplectic eigenvalue is far from unity, then the state must be at least a certain distance away from every pure Gaussian state.
\begin{lemma}[(Lower bound on the minimum trace distance from the set of pure Gaussian states)]\label{lemma_lower_bound_min}
    Let $\rho$ be an $n$-mode state with second moment of the energy upper bounded by $E$, i.e.~$\sqrt{\Tr[\hat{E}^2\rho]}\le n E$. Moreover, let $\nu_{\rm{max}}$ be the maximum symplectic eigenvalue of the covariance matrix of $\rho$. Then, it holds that
    \bb
    \min_{\psi\in\mathcal{G}_E}\frac 12 \|\rho-\psi\|_1\geq\frac{(\nu_{\rm{max}}-1)^2}{c\,(nE)^6}\,,
    \ee
    where $c=3\cdot 2^9\cdot 3098$ and the minimum is taken over the set $\mathcal{G}_{\rm pure}^{(n,E)}$ of all $n$-mode pure Gaussian states with mean energy upper bounded by $E$.
\end{lemma}

\begin{proof}

    We have
    \begin{align*}
        \min_{\psi\in\mathcal{G}_E}
        \frac{1}{2}\|\rho-\psi\|_1 & \geqt{(i)} \min_{\psi\in\mathcal{G}_E}\frac{\|V(\rho)-V(\psi)\|_{\infty}^2}{2\cdot 3098\,\max(\Tr[\hat{E}^2\rho],\Tr[\hat{E}^2\psi])}    \\
                                   & \geqt{(ii)} \min_{\psi\in\mathcal{G}_E}\frac{\|D-\mathbb{1}\|_{\infty}^2}{2\cdot 3098\,\max(E^2,\Tr[\hat{E}^2\psi])K(V(\rho))K(V(\psi))} \\
                                   & \geqt{(iii)} \min_{\psi\in\mathcal{G}_E}\frac{\|D-\mathbb{1}\|_{\infty}^2}{2^9\,3098\,\max((n E)^2,\Tr[\hat{E}^2\psi])(nE)^4}            \\
                                   & \geqt{(iv)}  \frac{\|D-\mathbb{1}\|_{\infty}^2}{2^9\,3098\,\max((nE)^2,3(nE)^2)(nE)^4}                                                   \\
                                   & = \frac{\|D-\mathbb{1}\|_{\infty}^2}{c\,(nE)^6}                                                                                          \\
                                   & = \frac{(\nu_{\rm{max}}-1)^2}{c\,(nE)^6},                                                                                                \\
    \end{align*}
    where $c\coloneqq 3\cdot 2^9\cdot 3098$.
    Here, in (i), we employed Lemma~\ref{lower_bound_NG}. In (ii), we applied Lemma~\ref{lem:wolf} and we denoted as $V(\rho)=SDS^\intercal$ the Williamson decomposition of $V(\rho)$. In (iii), we used that the condition number of the covariance matrix $V$ of a quantum state $\sigma$ satisfies
    \bb
    K(V)\le 16\Tr[\hat{E}\sigma]\,,
    \ee
    as follows by combining \eqref{cond_numb} and \eqref{cond_numb2}. Finally, in (iv), we applied \cite[Lemma~S55]{mele2024learning}, which states that the second moment of the energy of any Gaussian state $\sigma$ can be upper bounded in terms of the mean energy as $\Tr[\hat{E}^2\sigma]\le 3(\Tr[\hat{E}\sigma])^2$.
\end{proof}
Similarly to what we did with the upper bound in Lemma~\ref{lem:upper_perturbed}, we need a robust version of the lower bound in Lemma~\ref{lemma_lower_bound_min} that remains valid under perturbations. The following lemma provides such a guarantee.
\begin{lemma}[(Perturbed lower bound)]\label{lemma_low_approx}
    Let $\rho$ be an $n$-mode state with second moment of the energy upper bounded by $E$, i.e.~$\sqrt{\Tr[\hat{E}^2\rho]}\le n E$. Moreover, let $\tilde V$ be a covariance matrix such that it is $\varepsilon$-close to $V(\rho)$ in operator norm:
    \begin{align}\label{eq:hypothesis_V_infty}
        \|V(\rho)-\tilde V\|_\infty\leq \varepsilon_V\,,
    \end{align}
    By denoting as $\tilde{\nu}_{\rm{max}}$ the maximum symplectic eigenvalue of $\tilde V$, it holds that
    \bb
    \min_{\psi\in\mathcal{G}_E}\frac 12\|\rho-\psi\|_1\geq \frac{1}{2c\,(nE)^6}\left[(\tilde{\nu}_{\rm{max}}-1)^2-8nE(4nE+\varepsilon_V)\varepsilon_V\right]\,,
    \ee
    where $c\coloneqq 3\cdot 2^9\cdot 3098$ and where the minimum is taken over the set $\mathcal{G}_E$ of all $n$-mode pure Gaussian states with mean energy upper bounded by $E$.
\end{lemma}

\begin{proof}
    We have that
    \bb
    \min_{\psi\in\mathcal{G}_E}\frac{1}{2}\|\rho-\psi\|_1&\geqt{(i)} \frac{(\nu_{\rm{max}}(\rho)-1)^2}{c\,(nE)^6}\\
    &\geqt{(ii)}\frac{1}{2} \frac{(\tilde{\nu}_{\rm{max}}-1)^2}{c\,(nE)^6}-\frac{(\nu_{\rm{max}}(\rho)-\tilde{\nu}_{\rm{max}})^2}{c\,(nE)^6}\\
    &\ge \frac{1}{2} \frac{(\tilde{\nu}_{\rm{max}}-1)^2}{c\,(nE)^6}-\frac{\|D-\tilde{D}\|_\infty^2}{c\,(nE)^6}\\
    &\geqt{(iii)} \frac{1}{2} \frac{(\tilde{\nu}_{\rm{max}}-1)^2}{c\,(nE)^6}-\frac{ 4nE(4nE+\varepsilon_V)\varepsilon_V}{c(nE)^6}\\
    &= \frac{1}{2c\,(nE)^6}\left[(\tilde{\nu}_{\rm{max}}-1)^2-8nE(4nE+\varepsilon_V)\varepsilon_V\right]\,.
    \ee
    Here, in (i), we employed Lemma~\ref{lemma_lower_bound_min}. In (ii), we used the elementary inequality $(a+b)^2\ge\frac12a^2-b^2$ valid for any $a,b\in\mathbb{R}$. In (iii), we used Lemma~\ref{lem:wolf} together with the fact that $\sqrt{K(V(\rho))K(\tilde{V})}\le 4nE(4nE+\varepsilon_V)$, which can be proved in the exact same way as \eqref{eq_condxxx} by noting that $\Tr[\hat{E}\rho]\le \sqrt{\Tr[\hat{E}^2\rho]}\le n E $.
\end{proof}
With both upper and lower perturbed bounds in place, we can now connect covariance matrix estimation to the tolerant testing task.
The next lemma shows that if the estimated covariance matrix is sufficiently accurate, then one can deterministically decide whether the state is $\varepsilon_A$-close or $\varepsilon_B$-far
from the set of pure Gaussian states.
\begin{lemma}[(A good estimator provides a successful testing strategy)] \label{lem:good_estimator}
    Let $\rho$ be an unknown $n$-mode state satisfying the second-moment constraint $\sqrt{\Tr[\rho \hat{E}^2]}\le nE$. Let $\varepsilon_{B},\varepsilon_A$ such that $1>\varepsilon_{B}>\varepsilon_A\ge0$ and assume that one of the following two hypotheses holds:
    \begin{itemize}
        \item Hypothesis $H_A$: the state $\rho$ is $\varepsilon_A$-close to be a pure energy-constrained Gaussian state:
              \bb \min_{\sigma\in\mathcal{G}_E}\frac 12\|\rho-\sigma\|_1\leq \varepsilon_A\,,
              \ee
        \item Hypothesis $H_B$: the state $\rho$ is $\varepsilon_B$-far to be a pure energy-constrained Gaussian state:
              \bb\min_{\sigma\in\mathcal{G}_E}\frac 12\|\rho-\sigma\|_1>\varepsilon_B\,,
              \ee
    \end{itemize}
    where the minimum is taken over the set $\mathcal{G}_E$ of all pure $n$-mode Gaussian states with mean energy upper bounded by $E$. Let $\eps_V\geq 0$ and let $\tilde{V}$ be a covariance matrix such that $\|\tilde V - V(\rho)\|_\infty\leq \varepsilon_V$.
    If the following conditions holds, then it is possible to (deterministically) decide whether $H_A$ or $H_B$ holds: $\eps_A$, $\eps_B$ and $\eps_V$ satisfy
    \bb\label{cond_eps_crazy}
    \frac{2\varepsilon_B^2}{n}-\frac{\Delta}{\sqrt{n}}>\sqrt{2\Delta+2c(nE)^6\eps_A},
    \ee
    with $c\coloneqq 3\cdot 2^9\cdot 3098$ and $\Delta\coloneqq 4nE(4nE+\eps_V)\eps_V$.

    Specifically, by denoting as $\tilde{\nu}_{\max}$ the maximum symplectic eigenvalue of $\tilde{V}$, we have that if $\tilde{\nu}_{\max}\ge 1+\frac{2\varepsilon_B^2}{n}-\frac{4nE(4nE+\eps_V)\eps_V}{\sqrt{n}}$ is satisfied, then $H_B$ holds, otherwise $H_A$ holds.
\end{lemma}

\begin{proof}
    Let us introduce
    \bb
    \nu_{thr}\coloneqq 1+\frac{2\varepsilon_B^2}{n}-\frac{\Delta}{\sqrt{n}}
    \ee
    We want to prove that, if $\tilde{\nu}_{\max}\ge \nu_{thr}$,
    then $H_B$ holds.
    Let us start by noticing two facts that follow from \eqref{cond_eps_crazy}.
    First, the LHS is positive, therefore we get
    \bb\label{cond_xsquare}
    \frac{2\eps_B^2}{n}-\frac{\Delta}{\sqrt{n}}>0 \qquad\text{i.e.}\qquad \nu_{thr}>1.
    \ee
    Second, we immediately have
    \bb\label{eq:ineq_2}
    \left(\frac{2\varepsilon_B^2}{n}-\frac{\Delta}{\sqrt{n}}       \right)^2-2\Delta>2c(nE)^6\eps_A.
    \ee
    Now let us prove that if $\tilde{\nu}_{\max}\ge \nu_{thr}$
    then $H_B$ holds. Since $H_A$ and $H_B$ are mutually exclusive, this is equivalent to prove that if $\tilde{\nu}_{\max}\ge \nu_{thr}$ holds then $H_A$ does not hold, i.e.~$\min_{\psi\in\mathcal{G}_{\rm pure}^{(n,E)}}\frac{1}{2}\|\rho-\psi\|_1> \varepsilon_A$. To this end, note that
    \bb
    \min_{\psi\in\mathcal{G}_E}\frac{1}{2}\|\rho-\psi\|_1&\geqt{(i)} \frac{1}{2c\,(nE)^6}\left[(\tilde{\nu}_{\rm{max}}-1)^2-2\Delta\right]\\
    &\geqt{(ii)} \frac{1}{2c\,(nE)^6}\left[\left(\frac{2\varepsilon_B^2}{n}-\frac{\Delta}{\sqrt{n}}\right)^2-2\Delta\right]\\
    &\textg{(iii)}\varepsilon_A\,.
    \ee
    Here, in (i), we exploited Lemma~\ref{lemma_low_approx}. In (ii), we exploited the assumption that $\tilde{\nu}_{\max}\ge \nu_{thr}$ together with \eqref{cond_xsquare}. Finally, (iii) follows from \eqref{eq:ineq_2}.
    Now, it remains to prove that if $\tilde{\nu}_{\max}< \nu_{thr}$
    then $H_A$ holds. Similarly, it is sufficient to prove that $H_B$ does not hold, i.e.~$\min_{\psi\in\mathcal{G}_{\rm pure}^{(n,E)}}\frac{1}{2}\|\rho-\psi\|_1\le \varepsilon_B$. To this end, note that
    \bb
    \min_{\psi\in\mathcal{G}_E}\frac{1}{2}\|\rho-\psi\|_1&\leqt{(iv)} \frac{1}{\sqrt 2}\sqrt{n(\tilde\nu_{\max}-1)+\sqrt{n}\Delta}\\
    &<\frac{1}{\sqrt{2}}\sqrt{n\left(\frac{2\varepsilon_B^2}{n}-\frac{\Delta}{\sqrt{n}}\right)+\sqrt{n}\Delta}\\
    &=\varepsilon_B\,\,.
    \ee
    Here, in (iv) we exploited Lemma~\ref{lem:upper_perturbed}. This concludes the proof.
\end{proof}
The previous lemma reduces Gaussian property testing to the task of accurately estimating the covariance matrix.
To turn this reduction into an explicit sample complexity bound, we need to determine how many copies of the state
suffice to estimate the covariance matrix with the required precision.  The following lemma provides this guarantee.
\begin{lemma}[(Sample complexity of estimating the covariance matrix)]\label{correctness_algorithm_cov}
    Let $\varepsilon, \delta\in(0, 1)$ and $E>0$. Let $\rho$ be an $n$-mode quantum state satisfying with second moment of the energy upper bounded by $nE$, i.e.~$\sqrt{\Tr\!\left[\hat{E}^2 \rho \right]}\le nE$. Then, a number
    \bb\label{number_of_copies_cov}
    (n+3)\ceil{  68 \log\!\left(\frac{2 (2n^2+3n)  }{\delta}\right) \frac{200(8n^2 E^2+3n)}{\varepsilon^2}}=\bigO\!\left( \log\!\left(\frac{n^2}{\delta}\right)\frac{n^3E^2}{\varepsilon^2}   \right)\,,
    \ee
    of copies of $\rho$ are sufficient to build a vector $\tilde{\mathbf{m}}\in\mathbb{R}^{2n}$ and a symmetric matrix $\tilde{V}'\in\mathbb{R}^{2n,2n}$ such that
    \bb\label{eq:bounds_estimators}
    \Pr\left(  \|\tilde{V}'-V\!(\rho)\|_2\le \varepsilon\quad \text{ and }\quad \tilde{V}'+i\Omega \ge 0 \quad \text{ and }\quad \|\tilde{\mathbf{m}}-\mathbf{m}(\rho)\|\le \frac{\varepsilon}{10\sqrt{8nE}}\right)\ge 1-\delta\,.
    \ee
    Such procedure only requires single-copy measurements.
\end{lemma}
The above result is a slight improvement over \cite[Lemma~S54]{mele2024learning}. The only difference with \cite[Lemma~S54]{mele2024learning} is that the guarantee on the estimator of the covariance matrix is provided with respect the $2$-norm, specifically as $\|\tilde{V}'-V\!(\rho)\|_2\le \varepsilon$. In contrast, in~\cite{mele2024learning}, the guarantee was provided with respect the operator norm, specifically as $\|\tilde{V}'-V\!(\rho)\|_\infty\le \varepsilon$. Our improvement stands from the fact that the operator norm is always less or equal to the $2$-norm.
\begin{proof}
    The proof is completely analogous to the one of \cite[Lemma~S54]{mele2024learning}. Indeed, it easily follows by intervening in the proof of \cite[Lemma~S54]{mele2024learning} in Eq.~(S278) and Eq.~(S280) of \cite{mele2024learning} by considering the $2$-norm instead of the operator norm.
\end{proof}

We are now in position to combine all the above ingredients and obtain an explicit sample complexity upper bound. The result is stated in the following theorem.

\begin{boxed}{}
    \begin{thm}[(Sample complexity for the learning approach in the pure state scenario)]\label{thm:pure_testing}

        Consider the task of testing Gaussianity as in \cref{prob:testing}, with $\mathcal{P}=\mathcal{G}_E$ on $n$ modes, and with parameters $0\leq \eps_A<\eps_B$ and $\delta\in(0,1)$, on a possibly mixed state $\rho$ satisfying the energy bound $\sqrt{\Tr[\hat{E}^2 \rho]} \leq n E$.
        Let $c\coloneqq 3\cdot 2^9\cdot 3098$. Then, if
        \bb
        \eta \coloneqq\eps_B^4-\frac{c}{2} n^8E^6\eps_A>0,
        \ee
        then single-copy measurements on
        \[ N=\bigO\!\left( \log\!\left(\frac{n}{\delta}\right)\frac{n^7E^6}{\eta^2}   \right) \]
        copies of $\rho$ suffice to decide whether $\rho$ is $\epsilon_A$-close or $\epsilon_B$-far from $\mathcal{G}_{E}$ with success probability at least $1-\delta$.
    \end{thm}
\end{boxed}

\begin{proof}
    Let $\Delta\coloneqq 4nE(4nE+\eps_V)\eps_V$. 
    We want to prove that the requirement $\eps_B^4>\frac{c}{2}n^8E^6\eps_A$ ensures that \eqref{cond_eps_crazy} in Lemma \ref{lem:good_estimator} can be satisfied for $\Delta$ (i.e.~$\eps_V$) small enough. More precisely, let us rewrite \eqref{cond_eps_crazy} as
    \bb\label{eq:saturate}
    \mu -\left(\frac{4\eps_B^2}{n\sqrt n}+2\right)\Delta+\frac{\Delta^2}{n}>0\qquad\text{where}\qquad\mu\coloneqq\left(\frac{2\eps_B^2}{n}\right)^2-2c(nE)^6\eps_A.
    \ee
    The inequality in \eqref{eq:saturate} is satisfied for
    \bb
    0<\Delta\leq\Delta^\ast\coloneqq \frac{\mu}{6}
    \ee
    as $\frac{4\eps_B^2}{n\sqrt n}+2\leq 6$.
    So, for any $0<\Delta\leq \Delta^\ast$, we have a valid $\eps_V$ satisfying the hypotheses of Lemma \ref{lem:good_estimator}. In particular,
    \bb
    \eps_V=\frac{\mu}{24nE(4nE+1)}
    \ee
    is a valid choice of $\eps_V$, as
    \bb
    \Delta=4nE(4nE+\eps_V)\eps_V\leq 4nE(4nE+1)\eps_V=\frac{\mu}{6}.
    \ee
    By Lemma \ref{correctness_algorithm_cov}, using single-copy measurements, with probability at least $1-\delta$ we can produce a covariance matrix $\tilde{V}$ such that
    \begin{align}
        \|\tilde{V}-V\!(\rho)\|_2\le \varepsilon_V
    \end{align}
    by using at most
    \bb
    N=(n+3)\ceil{  68 \log\!\left(\frac{2 (2n^2+3n)  }{\delta}\right) \frac{200(8n^2 E^2+3n)}{(\varepsilon_V)^2}}=\bigO\!\left( \log\!\left(\frac{n}{\delta}\right)\frac{n^7E^6}{\left(2\eps_B^4- cn^8E^6\eps_A\right)^2}   \right)\,,
    \ee
    copies of $\rho$, according to \eqref{number_of_copies_cov} and to the fact that $nE\ge \frac{n}{2}$. In particular, by Lemma \ref{lem:good_estimator} we can correctly guess the hypothesis with probability at least $1-\delta$.
\end{proof}
\subsection{Non-tolerant pure Gaussian property testing by tomography}\label{app:tom}
When we consider the particular case of non-tolerant testing, i.e.~$\eps_A = 0$, a recent tomography algorithm for Gaussian states \cite{bittel2025energyindependent} can be used as a subroutine of a Gaussianity testing procedure yielding a better sample complexity.
\begin{boxed}{}
    \begin{thm}[(Sample complexity of non-tolerant testing by Gaussian tomography)]\label{thm:pure_testing_by_learning}

        Consider the task of testing Gaussianity as in \cref{prob:testing} with $\mathcal{P}=\mathcal{G}_E$, where we set $\epsilon_A=0$, $\epsilon_B> 0$ and $\delta\in (0,1)$ on a pure state $\rho$.  There is a quantum algorithm that uses
        \begin{align}
            N =\bigO\!\left(\frac{n^3 + n \log\log E}{\eps_B^2}\log(\delta^{-1})\right).
        \end{align}
        copies of $\rho$ to decide whether $\rho$ is a Gaussian state or $\epsilon_B$-far from $\mathcal{G}_{E}$ with success probability at least $1-\delta$.
    \end{thm}
\end{boxed}
\begin{proof}
  If we set $\epsilon_A = 0$, we can obtain a testing bound using Gaussian tomography. Specifically, we employ the energy-independent tomography method described in~\cite{bittel2025energyindependent}. For a Gaussian state $\rho$, this yields an estimate $\hat\rho$ satisfying
\begin{align}
\frac{1}{2} \| \hat\rho - \rho \|_1 \leq \frac{1}{2} \epsilon,
\end{align}
with the required number of copies bounded by
\begin{align}
N = \mathcal{O}\left(\frac{n^3 + n \log\log E}{\epsilon^2} \log(\delta^{-1})\right),
\end{align}
where $\epsilon$ will be chosen later to be $\epsilon=\epsilon_B/2$.
We first consider Case A where the target state $\rho_A$ is a pure Gaussian state. The tomography algorithm may yield an estimate $\hat\rho_A$ that is not pure. However, let $S_{\epsilon/2}(\hat\rho_A)$ be the set of pure Gaussian states which are $\epsilon/2$-close to $\hat \rho_A$. This set is non-empty as $\rho_A\in S_{\epsilon/2}(\hat\rho_A)$. Choose any $\bar\rho_A\in S_{\epsilon/2}(\hat\rho_A)$, which by definition satisfies
\begin{align}
\frac{1}{2} \| \bar\rho_A - \hat\rho_A \|_1 \leq \frac{1}{2} \epsilon.
\end{align}
Since computational cost is not a concern, finding such a $\bar\rho_A$ is acceptable in this context.
By the triangle inequality, we obtain:
\begin{align}
\frac{1}{2} \| \bar\rho_A - \rho_A \|_1 
&\leq \frac{1}{2} \| \bar\rho_A - \hat\rho_A \|_1 + \frac{1}{2} \| \hat\rho_A - \rho_A \|_1 \\
&\leq \frac{1}{2} \epsilon + \frac{1}{2} \epsilon = \epsilon.
\end{align}
We now perform a swap test between $\bar\rho_A$ and $\rho_A$, which estimates their overlap:
\begin{align}
\Tr(\bar\rho_A \rho_A) 
= 1 - \left( \frac{1}{2} \| \bar\rho_A - \rho_A \|_1 \right)^2 
\geq 1 - \epsilon^2,
\end{align}
using the identity relating trace distance and fidelity for pure states. Hence, in Case A, the swap test yields a value close to 1.

Now we consider Case B where the target state $\rho_B$ is $\epsilon_B$-far from $\mathcal{G}_E$. We proceed as above: first, the Gaussian tomography algorithm produces a guess $\hat \rho_B$; second,
\begin{itemize}
    \item if $S_{\epsilon/2}(\hat\rho_B)$ is empty (i.e.~there is no pure state $\epsilon/2$-close to $\hat \rho_B$), then we conclude Case B directly;
    \item otherwise, we obtain a pure Gaussian state estimate $\bar\rho_B\in S_{\epsilon/2}(\hat\rho_B)$. In this case, we again perform a swap test.
\end{itemize}
By assumption, any pure Gaussian state is at least $\epsilon_B$ away from $\rho_B$, so
\begin{align}
\Tr(\bar\rho_B \rho_B) 
= 1 - \left( \frac{1}{2} \| \bar\rho_B - \rho_B \|_1 \right)^2 
\leq 1 - \epsilon_B^2.
\end{align}
The problem now reduces to distinguishing two outcomes
\begin{align}
\Tr(\bar\rho_i \rho_i) \begin{cases}
\geq 1 - \epsilon^2 = p_A & \text{if } i = A, \\
\leq 1 - \epsilon_B^2 = p_B & \text{if } i = B.
\end{cases}
\end{align}
This is equivalent to distinguishing between two Bernoulli distributions, a well-known hypothesis testing task. 
Here we can directly apply \cref{lem:amplify error test}, by setting $\epsilon = \epsilon_B / 2$ this yields a sampling complexity of $O(\epsilon_B^{-2}\log(\delta^{-1}))$.
This shows that the sampling complexity for the Bernoulli hypothesis test is negligible compared to the copies required for the tomography step. Substituting $\epsilon = \epsilon_B / 2$ into the tomography complexity yields the final overall bound.
\end{proof}
\section{Mixed state setting}\label{sec:mixed}
In the previous sections, we analyzed the case of pure states, where efficient Gaussianity testing is possible under suitable constraints.  We now turn to the more general setting of mixed states.  Our goal is to understand whether efficient property testers still exist, or whether fundamental limitations arise.

We consider two natural formulations of the problem, depending on the
distance measure used to quantify closeness to the Gaussian set.
In subsection \ref{subsec:mixed1}, we first study the case where closeness is measured in trace distance. Here we show, via a reduction from a classical hard testing problem, that any tester must use exponentially many copies of the state.
In subsection \ref{subsec:mixed2}, we analyze the analogous problem defined using relative entropy as the distance measure, and prove that the same exponential lower bound persists.
Together, these results demonstrate that Gaussianity testing in the mixed-state setting is intrinsically hard, in sharp contrast to the pure-state regime.\\

Throughout this section, we denote by $\mathcal{G}_E$ the set of all $n$-mode Gaussian states $\rho$ with mean energy per mode at most $E$, namely
\begin{align}
    \Tr[\rho \hat{E}] \leq n E\,,
\end{align}
where $\hat{E}$ denotes the energy operator. The testing scenario we are going to study is \cref{prob:testing} (see also \cref{sec:intro_to_property_testing}) where $\mathcal{P}=\mathcal{G}_{E,\mixed}$.

\subsection{Trace-distance based problem}\label{subsec:mixed1}

Our hardness result will be based on a classical reduction.
We recall a lemma due to Ref.~\cite{Valiant}, which constructs families of distributions that are information-theoretically hard to distinguish.  This lemma will serve as the cornerstone of our lower bound.
\begin{lemma}[(Classical identity testing~\cite{Valiant})]\label{thm:N^n} Let $q$ be a probability distribution over $\mathbb{N}^n$ with $\|q\|_\infty \leq 1/2$ and let $\eps\in (0,1)$. Then there is a family $\mathcal{F}_{q,\eps}$ of probability distributions over $\mathbb{N}^n$ with the following properties:
    \begin{itemize}
        \item if $p\in\mathcal{F}_{q,\eps}$, then there exists $z\in\{-1,+1\}^{\mathbb{N}^n}$ such that
              \bb\label{eq:characterization}
              p(i)=\frac{(1+4\eps z(i))q(i)}{\sum_{j\in\mathbb{N}^{n}}(1+4\eps z(j))q(j)}\qquad \forall\,i\in\mathbb{N}^n
              \ee
        \item $\frac{1}{2}\|p-q\|_1>\eps$ for any $p\in\mathcal{F}_{q,\eps}$;
        \item testing whether a distribution $p$ is $q$ or belongs to $\mathcal{F}_{q,\eps}$ with failure probability smaller than $1/3$ requires at least $\Omega\left(\frac{1}{\eps^2\|q\|_2}\right)$ samples of $p$.
    \end{itemize}
\end{lemma}
As a simple but useful corollary of the previous characterization,
one can upper bound the distributions in the hard family in terms of the reference distribution $q$. We record this observation for later use.
\begin{rem} For any $\varepsilon\in(0,\frac{1}{4})$ and $p\in\mathcal{F}_{q,\eps}$, it holds the uniform upper bound
    \bb\label{eq:upper_pq}
    p\leq \frac{1+4\eps}{1-4\eps}q
    \ee
    as an immediate consequence of \eqref{eq:characterization}.
\end{rem}

Another central ingredient in our analysis is the notion of \emph{Gaussianification}, namely the Gaussian state $G(\rho)$ that shares with a given state $\rho$ the same first and second moments. The next lemma provides estimates on the trace distance between a state and its Gaussianification, by proving that closeness to the Gaussian set guarantees closeness to the Gaussianification. This two-sided control is crucial for establishing hardness in the relative-entropy testing problem.
\begin{lemma}[(Trace distance between a state and its Gaussianification)]\label{lem:gauss}
    Let $\rho$ be an $n$-mode quantum state such that $\sqrt{\Tr[\hat{E}^2\rho]}\le nE$.
    Let $G(\rho)$ denote the Gaussianification of $\rho$. By definition of the minimum, one has
    \[
        \min_{\sigma\in\mathcal{G}_{E,\mixed}}\frac12\|\rho-\sigma\|_1
        \;\le\; \frac12\|\rho-G(\rho)\|_1\,,
    \]
    i.e.~if $\rho$ is close to its Gaussianification then it is also close to the set of Gaussian states.
    Interestingly, a converse bound also holds:
    \[
        \frac12\|\rho-G(\rho)\|_1
        \;\le\; \Biggl(1+ \frac{1+\sqrt{3}}{2}\sqrt{6}\,\sqrt{1549}\,(nE)^2
        +4\sqrt{2\sqrt{3}}\,(nE)^{3/2}\Biggr)\,
        \sqrt{ \min_{\sigma\in\mathcal{G}_{E,\mixed}}\frac12\|\rho-\sigma\|_1 }\,,
    \]
    i.e.~if $\rho$ is close to the set of Gaussian states then it is close to its Gaussianification.
\end{lemma}

\begin{proof}
    Define
    \[
        \varepsilon \coloneqq \min_{\sigma\in\mathcal{G}_{E,\mixed}}\frac12\|\rho-\sigma\|_1\,.
    \]
    Let $\sigma\in\mathcal{G}_{E,\mixed}$ be such that $\tfrac12\|\rho-\sigma\|_1=\varepsilon$. Then
    \begin{align*}
        \tfrac12\|\rho-G(\rho)\|_1
         & \leqt{(i)} \varepsilon+\tfrac12\|\sigma-G(\rho)\|_1                                                                    \\
         & \leqt{(ii)} \varepsilon+\frac{1+\sqrt{3}}{8}\max\!\bigl(\Tr V(\sigma),\Tr V(\rho)\bigr)\,
        \|V(\sigma)-V(\rho)\|_\infty
        + \frac{1}{2}\sqrt{2\min\!\bigl(\|V(\sigma)\|_\infty,\|V(\rho)\|_\infty\bigr)}\,\|\mathbf{m}(\rho)-\mathbf{m}(\sigma)\|_2 \\
         & \leqt{(iii)}  \varepsilon+ \frac{1+\sqrt{3}}{2}\,nE \,\|V(\sigma)-V(\rho)\|_\infty
        + \sqrt{2nE}\,\|\mathbf{m}(\rho)-\mathbf{m}(\sigma)\|_2                                                                   \\
         & \leqt{(iv)} \varepsilon+ \frac{1+\sqrt{3}}{2}(nE)^2\sqrt{3\times1549}\,\sqrt{\|\rho-\sigma\|_1}
        + 4\sqrt{2\sqrt{3}\,nE}\,(nE)\sqrt{\|\rho-\sigma\|_1}                                                                     \\
         & = \varepsilon+ \frac{1+\sqrt{3}}{2}\sqrt{6}\,\sqrt{1549}\,(nE)^2\sqrt{\varepsilon}
        + 4\sqrt{2\sqrt{3}}\,(nE)^{3/2}\sqrt{\varepsilon}                                                                         \\
         & \le \Biggl(1+ \frac{1+\sqrt{3}}{2}\sqrt{6}\,\sqrt{1549}\,(nE)^2
        +4\sqrt{2\sqrt{3}}\,(nE)^{3/2}\Biggr)\sqrt{\varepsilon}\,.
    \end{align*}
    Here:
    (i) follows from the triangle inequality;
    (ii) uses the upper bound on the trace distance between Gaussian states from Lemma~\ref{lem:upper_gaussian};
    (iii) employs the fact that
    \[
        \|V(\rho)\|_\infty \le \Tr V(\rho)
        = 4\Tr[\hat{E}\rho]-2\|\mathbf{m}(\rho)\|_2^2
        \le 4\sqrt{\Tr[\hat{E}^2\rho]} \le 4nE\,;
    \]
    (iv) relies on the lower bound on the trace distance between non-Gaussian states given in Lemma~\ref{lower_bound_NG}.
\end{proof}

We are now ready to state our main lower bound in the trace-distance setting. By embedding the classical hard instances into bosonic Fock space, we obtain the following exponential lower bound on the sample complexity.

\begin{boxed}{}
    \begin{thm}[(Hardness for mixed states)]\label{thm:hardness}
        Consider the task of testing Gaussianity as in \cref{prob:testing}, with $\mathcal{P}=\mathcal{G}_{E,\mixed}$ on $n$ modes, and with parameters $\eps_A$, $\eps_B$ and $\delta$, on a possibly mixed state $\rho$ satisfying the energy bound $\sqrt{\Tr[\hat{E}^2 \rho]} \leq n E$. Let
        \bb\label{eq:defCe}
        c_{nE}\coloneqq \sqrt {2}+ 8\sqrt[4]{3}(nE)^{3/2}+\frac{1+\sqrt{3}}{2} \sqrt{3098} (nE)^2 =\Theta\big((nE)^2\big)
        \ee
        Then, if $\varepsilon_{A}<\varepsilon_B<\frac{1}{32c_{nE}^2}$, deciding whether $\rho$ is $\epsilon_A$-close or $\epsilon_B$-far from $\mathcal{G}_{E,\mixed}$ with success probability larger than $2/3$ requires at least
        \[N=\Omega\left(\frac{E^n}{(c_{nE}^2\eps_B)^2}\right)=\Omega\left(\frac{E^n}{n^4E^4\eps_B^2}\right)\]
        copies of $\rho$.
    \end{thm}
\end{boxed}

\begin{proof}
    Suppose we have an algorithm that acts as a Gaussian property tester for mixed states, with a failure probability of at most \(\delta = 1/3\). Assume this algorithm uses \(N(n, E, \varepsilon_A, \varepsilon_B, 1/3)\) copies of the state to be tested.
    Our claim is that this algorithm can be incorporated as a subroutine of a statistical test to distinguish classical distributions as in Theorem \ref{thm:N^n}. This will provide, up to polynomial corrections, an exponential lower bound on $N(n,E,\varepsilon_{A},\varepsilon_B,1/3)$.\\

    Let $\mathcal{P}(\mathbb{N}^n)$ be the space of probability distributions on $\mathbb{N}^n$ and let $\mathcal{P}_2(\mathbb{N}^n)\subset \mathcal{P}(\mathbb{N}^n)$ be the space of probability distributions on $\mathbb{N}^n$ such that the second moment is finite. For any $p\in\mathcal{P}_2(\mathbb{N}^n)$ we define the corresponding $n$-mode bosonic state
    \bb
    \rho(p)\coloneqq \sum_{k\in\mathbb{N}^n}p(k)\ketbra{k},
    \ee
    where $\ket{k}$ is the Fock state
    \bb
    \ket{k}\coloneqq\ket{k_1\,k_2\,\dots k_n}.
    \ee
    In particular, we are going to call
    \bb
    q_\nu(k)\coloneqq \frac{1}{(\nu+1)^n}\left(\frac{\nu}{\nu+1}\right)^{\|k\|_1}
    \ee
    the geometric distribution corresponding to the i.i.d. thermal state
    \bb
    \rho(q_\nu)=\tau_\nu^{\otimes n}=\left(\frac{1}{\nu+1}\sum_{k_1=0}^\infty \left(\frac{\nu}{\nu+1}\right)^{k_1}\ketbra{k_1}\right)^{\otimes n}.
    \ee
    The second moment of the energy for this state is
    \bb
    \Tr[\tau_{\nu}^{\otimes n}\hat{E}^{2}]
    &= \Tr[\tau^{\otimes n}_{\nu}(\hat{N}^{2}+n\hat{N}+\frac{n^{2}}{4}\id)] \\
    &\eqt{(i)}\Tr[\tau_\nu^{\otimes n}\hat{N}^{2}] + n^{2}\nu + \frac{n^{2}}{4} \\
    &\eqt{(ii)}n\Tr[\tau_{\nu} (a^\dagger a)^2] + n(n-1)\nu^{2}+n^{2}\nu+\frac{n^{2}}{4} \\
    &\eqt{(iii)}n\nu(2\nu+1) + n(n-1)\nu^{2}+n^{2}\nu+\frac{n^{2}}{4}.
    \ee
    In the first step we used the equality $\hat{E}=\hat{N}+\frac{n}{2}$. In (ii) we distinguished if the $\hat{N}_{i}\hat{N}_{j}$ are the same or different, and in (iii) we used that $\mathrm{tr}[\tau_{\nu}(a^\dagger a)^{2}]=\frac{1}{\nu+1}\sum_{n=0}^{\infty}n^{2}(\frac{\nu}{\nu+1})^{n}=\nu(2\nu+1)$. For any $\eps\in(0,\frac{1}{4})$, we are going to call $\nu_{E,\eps}$ the solution of the equation
    \bb
    n\nu_{E,\eps}(2\nu_{E,\eps}+1) + n(n-1)\nu_{E,\eps}^{2}+n^{2}\nu_{E,\eps}+\frac{n^{2}}{4} = \frac{1-4\eps}{1+4\eps}n^2E^2.
    \ee
    It is clear that, asymptotically in $n$ and $E$, the solution $\nu_{E,\eps}$ is $\Omega(E)$.\\
    Let us now define the testing problem. Let us fix $E$ and let $\eps\coloneqq 8c_{nE}^2\eps_B\in(0,\frac{1}{4})$, since $\eps_B<\frac{1}{32c_{nE}^2}$. We are going to call $q\in\mathcal{P}_2(\mathbb{N}^k)$ the probability distribution
    \bb
    q\coloneqq q_{\nu_{E,\eps}}
    \ee
    that has to be tested against $\mathcal{F}_{q,\eps}$ as in Theorem \ref{thm:N^n}. According to the statement of the theorem and to our definitions, we have the following guarantees:
    \begin{itemize}
        \item upper bound on the second moment of the energy of state corresponding to $q$, namely
              \bb
              \Tr[\rho(q)\hat{E}^{2}]=\frac{1-4\eps}{1+4\eps}n^2E^2\leq n^2E^2;
              \ee
        \item upper bound on the second moment of the energy of state corresponding to the distributions $p\in\mathcal{F}_{q,\eps}$, as
              \bb
              \Tr[\rho(p)\hat{E}^{2}]\leqt{(iv)} \frac{1+4\eps}{1-4\eps}\Tr[\rho(q)\hat{E}^{2}]\leq n^2E^2,
              \ee
              where (iv) is due to \eqref{eq:upper_pq} and to the linearity of $p\mapsto\rho(p)$;
        \item trace distance bound
              \bb
              \frac{1}{2}\|\rho(p)-\rho(q)\|_1=\frac{1}{2}\|p-q\|_1>\eps
              \ee
              for any $p\in\mathcal{F}_{q,\eps}$.
    \end{itemize}
    Given $p\in\mathcal{P}_2(\mathbb{N}^k)$, we want to identify a classical algorithm that solves the property testing problem $p=q$ ($H_0$) or $p\in\mathcal{F}_{q,\eps}$ ($H_1$) with failure probability at most $\delta$.\\
    \noindent\textbf{Step zero: embedding in the quantum system}. The statistics of the state $\rho(p)^{\otimes N_{tot}}$ are equivalent to the statistics of the state obtained by preparing $N_{tot}$ independent $n$-mode bosonic systems in the states $\ket{X_i}$ (with $i=1,\dots, N_{tot}$ and $X_i\in\mathbb{N}^n$) according to the outcome of $N_{tot}$ i.i.d. samples $X_1,\dots,X_{N_{tot}}$ of $p$. So, in the following, we can assume to deal with $N_{tot}$ copies of $\rho(p)$.\\
    \noindent\textbf{Step one: test of the covariance}. Both $H_0$ and $H_1$ would yield states $\rho(p)$ with zero mean by their very definition. However, a first property of $\rho(p)$ that could be tested is its covariance matrix $V(p)$. If the estimator $\tilde V(p)$ of $V(p)$ differs from $V(q)$ more than a fixed threshold, then with high probability we can claim that $p\neq q$, whence $p\in\mathcal{F}_{q,\eps}$. More precisely, using $N_{cov}$ copies of the $N_{tot}$ available copies of $\rho(p)$ we can produce an estimator $\tilde V(p)$ such that
    \bb
    \|\tilde V(p)-V(p)\|_2\leq \eps_V
    \ee
    with probability at least $1-1/6$, where
    \begin{itemize}
        \item we choose $\eps_V\coloneqq\frac{\eps}{4(1+\sqrt{3})nE}$
        \item therefore $N_{cov}=O\!\left( \log\!\left(n^2\right)\frac{n^5E^4}{\varepsilon^2}   \right)$ according to Lemma \ref{correctness_algorithm_cov},
    \end{itemize}
    If $\|\tilde V(p)-V(q)\|_2> \eps_V$, we can declare that $H_1$ holds since, in such case,
    \bb
    \|V(p)-V(q)\|_2\geq \|\tilde V(p)-V(q)\|_2-\|\tilde V(p)-V(p)\|_2>0\,,
    \ee
    and
    \bb
    \|\rho(p)-\rho(q)\|_1&\ge \frac{\|V(p)-V(q)\|_{\infty}^2}{1549\,\max(\Tr[\hat{E}^2\rho(p)],\Tr[\hat{E}^2\rho(q)])}\geq \frac{\|V(p)-V(q)\|_{2}^2/n^2}{1549n^2E^2}>0\,,
    \ee
    i.e.~it is not possible that $p=q$.\\
    \noindent\textbf{Step two: the Gaussianity test}. In this last step, we want to prove that Gaussianity test allows to identify whether $H_0$ or $H_1$ hold when $\|\tilde V(p)-V(q)\|_2\le \eps_V$. In this way, we will have a lower bound on the number of samples needed in this testing algorithm and, as a corollary, on $N(n,E, \eps_A,\eps_B,\delta)$. Let us run the Gaussianity test with $\eps_A<\eps_B$ and failure probability at most $\frac{1}{6}$. The property we are going to test are
    \begin{itemize}
        \item $H_A$: $\min_{\sigma\in\mathcal{G}^{n,E}}\|\rho(p)-\sigma\|\leq \eps_A$,
        \item $H_B$: $\min_{\sigma\in\mathcal{G}^{n,E}}\|\rho(p)-\sigma\|> \eps_B$.
    \end{itemize}
    Both $H_0$ and $H_1$ yield states that satisfy the second moment energy constraint. If, conditioning on $\|\tilde V(p)-V(q)\|_2\le \eps_V$, it holds that $H_0\implies H_A$ and $H_1\implies H_B$, then the Gaussianity test will conclude the solution of the classical property testing problem.
    Clearly $H_0$ implies $H_A$, since $\rho(q)=\tau_{\nu_{E,\eps}}^{\otimes n}$ is a Gaussian state. On the other hand, let us now suppose that $p\in\mathcal{F}_{q,\eps}$ and prove that $H_1\implies H_B$ by contradiction, If $H_B$ doesn't hold, i.e.~$\min_{\sigma\in\mathcal{G}^{n,E}}\|\rho(p)-\sigma\|_1\leq \eps_B$, then
    \bb\label{eq:stop}
    \|p-q\|_1&=\|\rho(p)-\rho(q)\|_1\\
    &\leq \|\rho(p)-G(\rho(p))\|_1+\|G(\rho(p))-\rho(q)\|_1\\
    &\leqt{(v)} c_{nE} \left(\min_{\sigma\in\mathcal{G}_{E,\mixed}}\|\rho(p)-\sigma\|_1 \right)^{1/2}+\|G(\rho(p))-\rho(q)\|_1\\
    &\leqt{(vi)}c_{nE} \sqrt{2\eps_B}
    + \frac{1+\sqrt{3}}{4}\max(\Tr V(p) ,\Tr V(q)) \,\|V(p)-V(q)\|_\infty\\
    &\leqt{(vii)} \frac{1}{2}\eps
    + (1+\sqrt{3})nE \| V(p)- V(q)\|_2\\
    &\leq \frac{1}{2}\eps
    + (1+\sqrt{3})nE(2\eps_V)\\
    & < \eps,
    \ee
    so $H_1$ cannot hold as well. In particular, in (v) we have introduced the Gaussianification of $\rho(p)$ and used Lemma \ref{lem:gauss}, in (vi) we have used the assumption that $H_B$ does not hold and Lemma \ref{lem:upper_gaussian}, and in (vii) we have used the fact that $\eps_B=\frac{\eps}{8c_{nE}^2}$ by the very definition of $\eps$.\\
    \noindent\textbf{Step three: back to classical}. Let us consider the channel $\Delta: \mathcal{L}(\mathcal{H}^{\otimes n})\to \mathcal{P}(\mathbb{N}^n)$, acting on $n$-mode bosonic states and producing classical distributions,
    \bb
    \Delta(\rho)\coloneqq \sum_{k\in\mathbb{N}^n}\Tr[\rho\ketbra{k}]\delta_k,
    \ee
    where $\delta_k$ is the Dirac delta in $k$, defined as
    \bb
    \delta_k(k')=\begin{cases}
        1 & k'=k     \\
        0 & k'\neq k
    \end{cases}.
    \ee
    It's easy to see that $\rho(p)=\Delta^\dagger(p)$.
    Let $\{A,\id-A\}$ be the POVM representing the algorithm of the first two steps, where $\Tr[A\rho(p)^{\otimes N_{tot}}]$ is the probability of accepting $H_0$. Then, calling $p^{N_{tot}}\in\mathcal{P}\left((\mathbb{N}^n)^{N_{tot}}\right)$ the distribution given by $N_{tot}$ copies of $p$, we have
    \bb
    \Tr[A\rho(p)^{\otimes N_{tot}}]=\Tr[A(\Delta^\dagger)^{\otimes N_{tot}}(p^{N_{tot}})]=\langle \Delta^{\otimes N_{tot}}(A),p^{N_{tot}}\rangle,
    \ee
    so $\{\Delta^{\otimes N_{tot}}(A),1-\Delta^{\otimes N_{tot}}(A)\}$ is a classical POVM acting on $N_{tot}$ samples of $p$ which produces the same statistics of the quantum algorithm testing $\rho(p)$. Therefore, by Theorem \ref{thm:N^n} we have
    \bb
    N_{cov}+N(n,E,\eps_A,\eps_B,1/3)=N_{tot}=\Omega\left(\frac{1}{\eps^2\|q\|_2}\right)=\Omega\left(\frac{(\nu_{E,\eps}+1)^n}{\eps^2}\right)=\Omega\left(\frac{E^n}{\eps^2}\right).
    \ee
    Being $N_{cov}$ polynomial in $n$, we conclude that
    \bb
    N(n,E,\eps_A,\eps_B,1/3)=\Omega\left(\frac{E^n}{\eps^2}\right).
    \ee
    By exploiting the definition of $\varepsilon$, we conclude the proof of the hardness of the mixed case.
\end{proof}

We conjecture that the requirement $\eps_B< 1/c_{nE}$ can be relaxed to $\eps_B<1$ by some nontrivial refinements in the proof strategy. It would be interesting to study whether the instance-optimal results of~\cite{odonnell2025instanceoptimalquantumstate} can be generalized to the continuous-variable setting. This could be a promising tool to relax the assumptions of \cref{thm:hardness_intro}.

\subsection{Relative-entropy based problem}\label{subsec:mixed2}
The hardness result of Theorem~\ref{thm:hardness} establishes that testing Gaussianity in the mixed-state regime requires exponentially many copies, but only under the additional assumption that the parameter $\varepsilon_B$ is sufficiently small. More precisely, $\varepsilon_B$ must be smaller than a quantity that decreases polynomially with both the energy $E$ and the number of
modes $n$ (cf.~\eqref{eq:defCe}).

When switching to the formulation based on the quantum relative entropy, this limitation disappears. In the relative-entropy setting we will prove an exponential lower bound for \emph{all} values of $\varepsilon_B$ up to a fixed constant, without any dependence on $n$ or $E$ (see Theorem~\ref{thm:h_rel_ent}). In this sense, the relative-entropy formulation provides a technically cleaner hardness result, showing that the intractability of mixed-state Gaussianity testing is not an
artifact of parameter tuning.

We now formalize the property testing task in the relative-entropy setting. This mirrors the trace-distance formulation, but replaces the distance measure with the quantum relative entropy. The problem definition below makes this precise.

\begin{problem}[(Mixed Gaussian property testing with relative entropy)]\label{test_rel_ent}
Let $\rho$ be an unknown $n$-mode state satisfying the second-moment constraint $\sqrt{\Tr[\rho \hat{E}^2]}\le nE$. Let $0\leq \varepsilon_A<\varepsilon_{B}<1$ and assume that one of the following two hypotheses holds:
\begin{itemize}
    \item Hypothesis $H_A$: the state $\rho$ is $\varepsilon_A$-close to the set of energy-constrained Gaussian states according to the relative entropy:
          \bb \min_{\sigma \in \mathcal{G}_{E,\mixed}}D(\rho\|\sigma)\leq \varepsilon_A\,,
          \ee
    \item Hypothesis $H_B$: the state $\rho$ is $\varepsilon_B$-far to the set of energy-constrained Gaussian states according to the relative entropy:
          \bb \min_{\sigma\in\mathcal{G}_{E,\mixed}}D(\rho\|\sigma)> \varepsilon_B\,,
          \ee
\end{itemize}
where the minimum is taken over the set $\mathcal{G}_{E,\mixed}$ of all (possibly-mixed) $n$-mode Gaussian states with mean energy upper bounded by $E$.
\end{problem}

Let us denote by $\mathcal{G}_{\mixed}$ the set of all $n$-mode Gaussian states, with no energy constraint. Leveraging the fact that the minimizer of the second argument of the quantum relative entropy among Gaussian states is the Gaussianification of the first argument (see \cite{Genoni2008,Marian2013}), namely
\bb\label{eq:gaussianif}
\min_{\sigma\in\mathcal{G}_{\mixed}}D(\rho\|\sigma)=D(\rho\|G(\rho)),
\ee
we can adapt the reduction technique developed for the trace-distance setting to the case of relative entropy. Importantly, this approach avoids the technical restriction encountered before: the lower bound will now hold uniformly for all $\varepsilon_B$ up to a fixed constant, without requiring it to vanish as a function of the number of modes $n$ or the
energy $E$. We can therefore state the following hardness result.
\begin{thm}[(Hardness for mixed Gaussianity testing in the relative-entropy setting)]\label{thm:h_rel_ent} Let $n,E,\varepsilon_{A},\varepsilon_B$ as in Problem \ref{test_rel_ent}. Let us suppose that $\eps_B<\frac{1}{2}$. Then Problem \ref{test_rel_ent} requires at least \[N=\Omega\left(\frac{E^n}{n^4E^4\eps_B^2}\right)\]
    copies of the state to be tested in order to be solved with success probability at least $2/3$
\end{thm}

\begin{proof}
    The proof is very similar to the one of proof of Theorem \ref{thm:hardness}. However, we will now define $\eps\coloneqq \sqrt{2\eps_B}$.
    Then we proceed as in the steps zero and one in the proof of Theorem \ref{thm:hardness}, where $\eps_V\coloneqq\frac{\eps}{4(1+\sqrt{3})nE}$ has now been defined using the new value of $\eps$. In step two, eq. \eqref{eq:stop} now becomes
    \bb
    \|p-q\|_1&= \|\rho(p)-\rho(q)\|_1\\
    &\leq \|\rho(p)-G(\rho(p))\|_1+\|G(\rho(p))-\rho(q)\|_1\\
    &\overset{\mathrm{(i)}}{\le} \sqrt{\frac{1}{2}D(\rho(p)\|G(\rho(p)))} +\frac{1}{2}\eps\\
    &\eqt{(ii)} \sqrt{\frac{1}{2}\min_{\sigma\in\mathcal{G}_{\mixed}}D(\rho(p)\|\sigma)} +\frac{1}{2}\eps\\
    &\leqt{(iii)} \sqrt{\frac{1}{2}\min_{\sigma\in\mathcal{G}_{E,\mixed}}D(\rho(p)\|\sigma)} +\frac{1}{2}\eps\\
    &\leqt{(vi)} \sqrt{\frac{1}{2}\eps_B} +\frac{1}{2}\eps\\
    & \eqt{(v)} \eps,
    \ee
    where
    \begin{itemize}
        \item in (i) we bound $\|G(\rho(p))-\rho(q)\|_1<\frac{1}{2}\eps$ as in the proof of Theorem \ref{thm:hardness}, while for $\|\rho(p)-G(\rho(p))\|_1$ we have used the quantum Pinsker inequality;
        \item in (ii) we have used \eqref{eq:gaussianif};
        \item in (iii) we have used that $\mathcal{G}_{E,\mixed}\subseteq \mathcal{G}_{\mixed}$;
        \item in (iv) we have used the assumption that $\displaystyle{\min_{\sigma\in\mathcal{G}_{E,\mixed}}}D(\rho(p)\|\sigma)\leq \eps_B$;
        \item in (v) we have used that $\eps_B=\frac{\eps^2}{2}$ by the very definition of $\eps$.
    \end{itemize}
    By step three of proof of Theorem \ref{thm:hardness}, we conclude again the proof.
\end{proof}
In summary, testing Gaussianity of mixed states is exponentially hard also in the relative-entropy setting. Compared to the trace-distance case, the result is technically stronger, since the exponential lower bound persists uniformly across a broad range of $\varepsilon_B$.

\end{document}